\documentclass[11pt]{article}
\usepackage[utf8]{inputenc}
\usepackage[top=3cm,left=3cm,right=3cm,bottom=3cm]{geometry}
\usepackage{amsmath,amsthm,amssymb}
\usepackage[hyperindex,breaklinks]{hyperref}
\usepackage[sorting=nyt,giveninits=true,doi=false,url=false,isbn=false,maxbibnames=99]{biblatex}
\usepackage{biblatex-shortfields}

\usepackage{caption} 
\usepackage[titletoc,title]{appendix}

\DeclareSourcemap{
  \maps[datatype=bibtex]{
    \map{
      \step[fieldsource=shortjournal,
            match={Journal of Geometry and Physics},
            replace={J. Geom. Phys.}]
    }
    \map{
      \step[fieldsource=shortjournal,
            match={Commun. Math. Phys.},
            replace={Comm. Math. Phys.}]
    }
    \map{
      \step[fieldsource=shortjournal,
            match={Commun.Math. Phys.},
            replace={Comm. Math. Phys.}]
    }
    \map{
      \step[fieldsource=journaltitle,
            match={American Journal of Mathematics},
            replace={Amer. J. Math.}]
    }
    \map{
      \step[fieldsource=shortjournal,
            match={Journal of Mathematical Physics},
            replace={J. Math. Phys.}]
    }
    \map{
      \step[fieldsource=shortjournal,
            match={Annals of Physics},
            replace={Ann. Physics}]
    }
    \map{
      \step[fieldsource=journaltitle,
            match={Communications on Pure and Applied Mathematics},
            replace={Comm. Pure Appl. Math.}]
    }
    \map{
      \step[fieldsource=journaltitle,
            match={Annals of Mathematics},
            replace={Ann. of Math.}]
    }
    \map{
      \step[fieldsource=journaltitle,
            match={Journal of the European Mathematical Society},
            replace={J. Eur. Math. Soc.}]
    }
    \map{
      \step[fieldsource=journaltitle,
            match={Memoirs of the American Mathematical Society},
            replace={Mem. Amer. Math. Soc.}]
    }
    \map{
      \step[fieldsource=journaltitle,
            match={Journal of Differential Geometry},
            replace={J. Differential Geom.}]
    }
    \map{
      \step[fieldsource=shortjournal,
            match={Physics Letters B},
            replace={Phys. Lett. B}]
    }
    \map{
      \step[fieldsource=journaltitle,
            match={Advances in Physics},
            replace={Adv. in Physics}]
    }
  }
}

\AtEveryBibitem{\clearlist{language}}
\AtEveryBibitem{\clearfield{month}}
\AtEveryBibitem{\clearfield{day}}

\makeatletter
\newcommand{\address}[1]{\gdef\@address{#1}}
\newcommand{\email}[1]{\gdef\@email{\url{#1}}}
\newcommand{\@endstuff}{\par\vspace{\baselineskip}\noindent\small
\begin{tabular}{@{}l}\scshape\@address\\\textit{E-mail address:} \@email\end{tabular}}
\AtEndDocument{\@endstuff}
\makeatother

\newtheorem{theorem}{Theorem}[section]
\newtheorem{lemma}[theorem]{Lemma}
\newtheorem{proposition}[theorem]{Proposition}

\theoremstyle{definition}

\theoremstyle{definition}
\newtheorem{definition}[theorem]{Definition}

\theoremstyle{definition}
\newtheorem{remark}[theorem]{Remark}

\newcommand{\diver}{\mathrm{div}}
\newcommand{\ric}{\mathrm{Ric}}
\newcommand{\tr}{\mathrm{tr}}

\newcommand{\lie}{\mathcal{L}}
\newcommand{\n}{\nabla}
\newcommand{\sn}{\overline{\n}}
\newcommand{\p}{\partial}
\newcommand{\s}{\varphi}
\newcommand{\R}{\mathbb{R}}
\newcommand{\cn}{\mathcal{N}}
\newcommand{\rt}{\widetilde R}
\newcommand{\srt}{\overline{\rt}}
\newcommand{\shatr}{\overline{\widehat{R}}}
\newcommand{\mck}{\mathcal{K}}
\newcommand{\mfx}{\mathfrak{X}}
\newcommand{\mch}{\mathcal{H}}
\newcommand{\mfe}{\mathfrak{E}}
\newcommand{\ci}{\mathcal{I}}

\newcommand{\tsum}{\textstyle\sum}
\newcommand{\bbh}{\mathbb{H}}
\newcommand{\bbl}{\mathbb{L}}
\newcommand{\I}{\mathrm{\mathbf{I}}}
\newcommand{\J}{\mathrm{\mathbf{J}}}
\newcommand{\K}{\mathrm{\mathbf{K}}}

\newcommand{\mfp}{\mathfrak{p}}
\newcommand{\subc}{\mathfrak{S}}
\newcommand{\crit}{\mathfrak{C}}
\newcommand{\ca}{\mathcal{A}}
\newcommand{\ct}{\mathcal{T}}
\newcommand{\ve}{\Vec{e}}
\newcommand{\densho}[1]{\rho_{(#1)}}
\newcommand{\lowen}[1]{\mfe_{(#1)}}
\newcommand{\side}{\mathcal{S}}

\newcommand{\dens}[2][\ell]{\rho_{(#2),#1}}
\newcommand{\lpar}{\mathfrak{a}}
\newcommand{\ce}{\mathcal{E}}
\newcommand{\sob}[3][]{\lVert #3 \rVert_{H^{#2}#1}}
\newcommand{\bigsob}[3][]{\big\lVert #3 \big\rVert_{H^{#2}#1}}
\newcommand{\sobprod}[4][]{\big\langle #3, #4 \big\rangle_{H^{#2}#1}}
\newcommand{\ck}[3][]{\lVert #3 \rVert_{C^{#2}#1}}
\newcommand{\bigck}[3][]{\big\lVert #3 \big\rVert_{C^{#2}#1}}
\newcommand{\pI}[1][]{\p_{\hspace{1pt}\I_{#1}}}
\newcommand{\cf}{\mathcal{F}}
\newcommand{\mfkr}{\mathring{\mathfrak{K}}}
\newcommand{\psir}{\mathring{\Psi}}
\newcommand{\phir}{\mathring{\Phi}}
\newcommand{\nr}{\mathring{N}}

\newcommand{\md}{\mathrm{md}}
\newcommand{\pr}{\mathring{p}}
\newcommand{\mchr}{\mathring{\mathcal{H}}}
\newcommand{\mckr}{\mathring{\mathcal{K}}}
\newcommand{\mfk}{\mathfrak{K}}
\newcommand{\mfr}{\mathfrak{R}}
\newcommand{\refmetric}{h_{\mathrm{ref}}}
\newcommand{\muref}{\mu_{\mathrm{ref}}}
\newcommand{\vol}{\mathrm{vol}}

\allowdisplaybreaks

\title{Localized formation of quiescent big bang singularities}
\author{Andrés Franco-Grisales}
\date{}
\address{Department of Mathematics, KTH, 100 44 Stockholm, Sweden}
\email{anfg@kth.se}

\begin{document}

\maketitle

\begin{abstract}
    We prove a localized big bang formation result, which does not require proximity of the initial data to any background solution. Suppose that we are given initial data for the Einstein--nonlinear scalar field equations on an open set $U \subset \R^3$. We identify a general condition on the initial data such that if the condition is satisfied in a large enough neighborhood of $x \in U$, then the corresponding maximal globally hyperbolic development has a local quiescent big bang singularity with curvature blow up to the past of $x$.

    We achieve the localization by introducing a new kind of foliation by spacelike hypersurfaces, given by the level sets of a time function satisfying a certain second order differential equation. This time function allows us to synchronize the singularity while at the same time yielding a symmetric hyperbolic formulation of Einstein's equations. Our new formulation also has two key advantages over previous localized big bang stability results. First, it is independent of the matter model, so it is possible that it could be used to prove big bang formation with matter models different from a scalar field. And second, it allows us to conclude that our solutions induce geometric initial data on the singularity, thus giving a complete description of the asymptotics towards the big bang.  
\end{abstract}

\tableofcontents

\section{Introduction}

In the study of cosmological singularities, it is natural to look at the heuristics developed by Belinskii, Khalatnikov and Lifschitz (BKL); see, e.g., \cite{belinskii_oscillatory_1970,belinskii_general_1982,belinskii_effect_1973}, or \cite{damour_cosmological_2003,heinzle_cosmological_2009} for recent refinements. What BKL proposed is that the behavior near the singularity should be spatially local and either oscillatory or quiescent. By quiescent we mean, roughly speaking, that the asymptotics of the spacetime are convergent towards the singularity. 

In \cite{oude_groeniger_formation_2023}, Oude Groeniger, Petersen and Ringström identified a general condition on initial data for Einstein's equations, such that the corresponding development exhibits quiescent big bang formation. Roughly speaking, they proved the following result. Suppose that we are given constant mean curvature (CMC) initial data for Einstein's equations, coupled to a nonlinear scalar field, which are subcritical in a sense analogous to \cite{fournodavlos_stable_2023}. If the mean curvature of the initial data is large enough, relative to the spatial variation, then the corresponding maximal globally hyperbolic development terminates at a quiescent big bang singularity to the past. This result is remarkable, since it is the first big bang formation result that does not require the initial data to be close to the data induced by any particular solution. Similarly as in most previous big bang stability results, they use a CMC foliation for their proof. These foliations are natural in this setting, since they synchronize the singularity. However, they have one issue, the CMC condition introduces an elliptic equation to the system. This means that the analysis cannot be localized in space in a straightforward way. This is at odds with the BKL heuristics, which suggest that the behavior near the big bang localizes. 

In \cite{beyer_localized_2025,zheng_localized_2026}, Beyer, Oliynyk and Zheng provide localized versions of the big bang stability results of \cite{fournodavlos_stable_2023}, for the Einstein--scalar field equations. In order to achieve the localization, they use the scalar field as a time function. This allows them to synchronize the singularity, while at the same time providing a fully hyperbolic formulation of Einstein's equations, which allows for the energy estimates to be done on localized spacetime domains. However, their approach has two limitations. First, since they use the scalar field as a time coordinate, it is not clear how to extend the results to other matter models. This is relevant, for instance, in the vacuum setting with spacetime dimensions $\geq 11$, where quiescent behavior is expected to occur. And second, their conclusions are incompatible with the recent asymptotic results of \cite{franco_complete_asymptotics_2026}, so it is unclear whether their solutions induce geometric initial data on the singularity as in \cite{ringstrom_initial_2025} or \cite{franco-grisales_developments_2025}; see \cite[Remark~1.4]{zheng_localized_2026}.

Here we prove a localized version of the big bang formation result of \cite{oude_groeniger_formation_2023}, in 4 spacetime dimensions. A rough version of our result is the following. Suppose that we are given initial data for the Einstein--nonlinear scalar field equations. We prove that if the mean curvature at a point $x$ is large enough, relative to the spatial variation on a large enough neighborhood of $x$, and if the subcriticality condition is satisfied in that neighborhood, then the corresponding maximal globally hyperbolic development has a local quiescent big bang singularity with curvature blow up to the past of $x$. To obtain this result, we introduce a new formulation of Einstein's equations, which relies on a foliation by spacelike hypersurfaces given by the level sets of a time function which satisfies a particular second order differential equation. Similarly to the previous localized results, this new time function allows us to synchronize the singularity while at the same time providing a fully hyperbolic formulation of Einstein's equations. Moreover, our formulation solves both of the issues that we mentioned with the previous localized results. It is independent of the matter model. So it is possible that it could be used to prove localized big bang formation with matter models different from a scalar field. And crucially, it is compatible with the results of \cite{franco_complete_asymptotics_2026}, which allows us to conclude that the local big bang singularity induces geometric initial data on the singularity as in \cite{ringstrom_initial_2025} or \cite{franco-grisales_developments_2025}.

\subsection{The Einstein--nonlinear scalar field equations}

Let $(M,g)$ be a 4-dimensional spacetime. We are interested in the Einstein--nonlinear scalar field equations,
\begin{subequations} \label{seq:nlsf sys}
\begin{align}
    \ric_g - \tfrac{1}{2} S_gg + \Lambda g &= T,\label{eq: einsten equation}\\
    \Box_g \s &= V' \circ \s.\label{eq: matter equation} 
\end{align}
\end{subequations}
Here $\ric_g$ and $S_g$ denote the Ricci and scalar curvatures of $g$ respectively, $\Lambda \in \R$ denotes the cosmological constant, $V \in C^\infty(\R)$ denotes the potential, $\s \in C^\infty(M)$ denotes the scalar field, $\Box_g$ is the wave operator associated with $g$, and 
\[
T = d\s \otimes d\s - \big( \tfrac{1}{2}g(d\s,d\s) + V \circ \s \big)g
\]
is the energy-momentum tensor of $\s$. Note that \eqref{eq: einsten equation} is equivalent to
\[
\ric_g = d\s \otimes d\s + (V\circ\s)g + \Lambda g.
\]
It is clear that the cosmological constant term can be included in the potential. Hence, without loss of generality, we can set $\Lambda = 0$. As is well known, these equations have an initial value problem formulation.

\begin{definition} \label{def: initial data}
    Let $(\Sigma,\bar h)$ be a 3-dimensional Riemannian manifold, $\bar k$ a symmetric
    covariant $2$-tensor on $\Sigma$ and $\bar\varphi, \bar\psi \in C^\infty(\Sigma)$ satisfying
    \begin{subequations}\label{seq:constraints}
        \begin{align}
        S_{\bar h} - |\bar k|^2_{\bar h} + (\tr_{\bar h} \bar k)^2
        &= \bar \psi^2 + |d\bar \varphi|_{\bar h}^{2} + 2 V \circ \bar \varphi, \label{eq: Hamiltonian constraint} \\
        \diver_{\bar h}\bar k - d(\tr_{\bar h} \bar k)
        &= \bar \psi d \bar \varphi. \label{eq: momentum constraint}
        \end{align}
    \end{subequations}
    Then $(\Sigma,\bar h,\bar k,\bar\varphi,\bar\psi)$ are called \emph{initial data for the Einstein--nonlinear scalar field equations}. If $(M,g,\varphi)$ is a solution to \eqref{seq:nlsf sys} and $\iota:\Sigma\rightarrow M$ is an embedding such that: $\iota^*g=\bar h$, so that $\iota(\Sigma)$ is a spacelike hypersurface in $(M,g)$; if $k$ is the second fundamental form induced on $\iota(\Sigma)$ and $U$ is the future pointing unit normal to $\iota(\Sigma)$, then
    \[
    \iota^*k=\bar k, \qquad \varphi\circ \iota= \bar\varphi, \qquad (U\varphi)\circ \iota=\bar\psi.
    \]
    Then $(M,g,\s)$ is a \emph{development} of the data.
\end{definition}

The fundamental result \cite{choquet-bruhat_global_1969} of Choquet-Bruhat and Geroch ensures that given initial data for the Einstein--nonlinear scalar field equations, there is a corresponding maximal globally hyperbolic development. In the quiescent big bang setting, these equations also have a singular initial value problem formulation. The relevant notion of initial data was introduced in \cite{ringstrom_initial_2025}. 

\begin{definition} \label{def:idos}
    Let $(\Sigma,\mchr)$ be a 3-dimensional Riemannian manifold, $\mckr$ a $(1,1)$-tensor on $\Sigma$ and $\phir, \psir \in C^\infty(\Sigma)$. Assume that the following hold: 
    \begin{enumerate}
        \item $\tr\mckr = 1$ and $\mckr$ is symmetric with respect to $\mchr$.
        \item $\tr\mckr^2 + \psir^2 = 1$ and $\diver_{\mchr} \mckr = \psir d\phir$.
        \item The eigenvalues $\pr_1,\pr_2,\pr_3$ of $\mckr$ are everywhere distinct and satisfy 
        \begin{equation} \label{eq: subcritical condition ids}
          \pr_I + \pr_J - \pr_K < 1
        \end{equation}
        for all $I,J$ and $K$ such that $I \neq J$.
    \end{enumerate}
    Then $(\Sigma,\mchr,\mckr,\phir,\psir)$ are \emph{robust nondegenerate quiescent initial data on the singularity for the Einstein--nonlinear scalar field equations}. We usually refer to $(\Sigma,\mchr,\mckr,\phir,\psir)$ as just initial data on the singularity for short.
\end{definition}

\begin{remark}
    The condition \eqref{eq: subcritical condition ids} corresponds to the subcritical condition of \cite{fournodavlos_stable_2023}. This requirement can be omitted from the definition. However, it is then necessary to impose an integrability condition on the eigenspaces of $\mckr$; see \cite[Definition~1.10]{ringstrom_initial_2025}. Moreover, the solutions corresponding to data that do not satisfy \eqref{eq: subcritical condition ids} are expected to be unstable. Here we are only interested in the stable setting. The condition that the eigenvalues of $\mckr$ be distinct is a technical condition that we require for the proof, but we expect similar results to hold in the absence of this condition. 
\end{remark}

That given initial data on the singularity, there is a corresponding maximal globally hyperbolic development is verified in \cite{franco-grisales_developments_2025}. In connection with Definitions~\ref{def: initial data} and \ref{def:idos} is the notion of expansion normalized initial data. This gives us the objects that are expected to converge to the initial data on the singularity, along a suitable foliation by spacelike hypersurfaces of a quiescent big bang spacetime.

\begin{definition} \label{def:exp norm id}
    Let $(M,g,\s)$ be a solution to the Einstein--nonlinear scalar field equations. Suppose that $\Sigma \subset M$ is a spacelike hypersurface with induced metric $h$, second fundamental form $k$, future pointing unit normal $U$ and mean curvature $\theta:=\tr_hk$. Assuming that $\theta > 0$ and $K := k^\sharp$, define
    \[
    \mck := \frac{1}{\theta}K, \qquad \mch(X,Y) := h(\theta^\mck X,\theta^\mck Y), \qquad
    \Psi := \frac{1}{\theta}U\s, \qquad
    \Phi := \s + \Psi\ln \theta,
    \]
    for all $X,Y\in T\Sigma$, where
    \[
    \theta^\mck X := \sum_{m=0}^\infty \frac{(\ln\theta)^m}{m!} \mck^m X.
    \]
    Then $\mck$ is called the
    \textit{expansion normalised Weingarten map}, and $(\Sigma,\mch,\mck,\Phi,\Psi)$ are referred to as the \textit{expansion normalized initial data} associated with the induced initial data $(\Sigma,h,k,\s|_\Sigma,U\s|_\Sigma)$. 
\end{definition}

\subsection{The main result}

Similarly as in \cite{oude_groeniger_formation_2023}, we need to restrict the class of potentials that we work with.

\begin{definition}
    Let $\sigma > 0$ and $V \in C^\infty(\R)$. We say that $V$ is a \emph{$\sigma$-admissible potential} if it is nonnegative, and there are constants $c_k$ such that
    \[
    |V^{(k)}(x)| \leq c_k e^{2(1-\sigma)|x|}
    \]
    for all $x \in \R$, and every nonnegative integer $k$.
\end{definition}

Before stating the main result, we introduce the type of localized spacetime domain where we construct our solutions.

\begin{definition} \label{def: spacetime domain}
    For $x = (x^1,x^2,x^3) \in \R^3$, denote by $|x| = \sqrt{\textstyle\sum_i (x^i)^2}$ its Euclidean norm. For some positive numbers $\delta$ and $\varepsilon$, and $t \in [0,\infty)$, define the domain
    \[
    \Omega_t := \Big\{ x \in \R^3 : |x| < \frac{1}{3\delta}t^{3\delta} + \varepsilon \Big\}.
    \]
    Moreover, for $I \subset (0,\infty)$ an interval, define the spacetime domain
    \[
    \Omega_I := \bigcup_{t \in I} \big(\{t\} \times \Omega_t\big) \subset (0,\infty) \times \R^3.
    \]
    If $\{s\} \times \Omega_t \subset \Omega_I$, we usually make the identification $\{s\} \times \Omega_t \cong \Omega_t$, whenever there is no danger of confusion. 
\end{definition}

\begin{theorem} \label{thm: main theorem}
    Fix $\delta \in (0,\frac{1}{3}]$, $\varepsilon > 0$ and $\zeta_0 > 0$, and let $V$ be a $3\delta$-admissible potential. Then there is a positive integer $k_1$, depending only on $\delta$, and a $\zeta_1 > 0$, depending only on $\delta,\varepsilon$ and $\zeta_0$, such that the following holds: \vspace{0.1cm}\\
    Let $U \subset \R^3$ be an open set that contains the origin and $\mathfrak{I} = (U,\bar h, \bar k, \bar\s,\bar\psi)$ be initial data for the Einstein--nonlinear scalar field equations with potential $V$. If $\bar\theta$ is the associated mean curvature, set the initial time to be $T := 1/\bar\theta(0)$ and assume that $\overline{\Omega}_T \subset U$. Let $(U,\bar\mch,\bar\mck,\bar\Phi,\bar\Psi)$ denote the associated expansion normalized initial data, and suppose that the subcritical condition holds:
    \begin{equation} \label{eq: subcritical condition}
        \bar p_I(x) + \bar p_J(x) - \bar p_K(x) < 1-6\delta
    \end{equation}
    for all $I \neq J$ and all $x \in \overline{\Omega}_T$, where $\bar p_I$ denote the eigenvalues of $\bar\mck$. Moreover, assume that
    \begin{subequations} \label{expansion normalized bounds}
    \begin{align}
        \|\bar\mch\|_{H^{k_1+1}(\Omega_T)} + \|\bar\mch^{-1}\|_{C^0(\overline{\Omega}_T)} + \|\bar\mck\|_{H^{k_1+1}(\Omega_T)} + \bar\theta(0)^{-1+4\delta}\|d\bar\theta\|_{H^{k_1}(\Omega_T)} &\leq \zeta_0,\\
        \|\bar\Phi\|_{H^{k_1+1}(\Omega_T)} + \|\bar\Psi\|_{H^{k_1+1}(\Omega_T)} &\leq \zeta_0;
    \end{align}
    \end{subequations}
    that $|\bar p_I(x) - \bar p_J(x)| \geq 1/\zeta_0$ for all $I \neq J$ and all $x \in \overline{\Omega}_T$; and that $\bar\theta(0) > \zeta_1$. Then the maximal globally hyperbolic development of $\mathfrak{I}$, say $(M,g,\s)$, with associated embedding $\iota : U \to M$ has a past local crushing big bang singularity in the following sense:\vspace{1ex}
    
    \noindent\textbf{Local crushing foliation:} There is a map $\Xi: \Omega_{(0,T]} \to D^-\big(\iota(\Omega_T)\big) \subset M$, which is a diffeomorphism onto its image, such that $\Xi(\{T\} \times \Omega_T) = \iota(\Omega_T)$, and the mean curvature $\theta(t,x)$ of the hypersurfaces $\Xi(\{t\} \times \Omega_t)$ blows up as $(t,x)$ approaches $\{0\} \times \Omega_0$.\vspace{1ex}
    
    \noindent\textbf{Asymptotic data:} There are initial data on the singularity for the Einstein--nonlinear scalar field equations $(\overline{\Omega}_0,\mchr,\mckr,\phir,\psir)$, and constants $C_\ell$ and $\sigma > 0$, such that
    \[
    \ck[(\overline{\Omega}_0)]{\ell}{\mch(t) - \mchr} + \ck[(\overline{\Omega}_0)]{\ell}{\mck(t) - \mckr} + \ck[(\overline{\Omega}_0)]{\ell}{\Phi(t) - \phir} + \ck[(\overline{\Omega}_0)]{\ell}{\Psi(t) - \psir} \leq C_\ell t^\sigma
    \]
    for all $t \in (0,T]$ and all $\ell$, where $(\overline{\Omega}_0,\mch(t),\mck(t),\Phi(t),\Psi(t))$ are the expansion normalized initial data induced on $\{t\} \times \overline{\Omega}_0 \cong \overline{\Omega}_0$.\vspace{1ex}
    
    \noindent\textbf{Curvature blow up:} Let $\pr_I$ denote the eigenvalues of $\mckr$ and $R_g$ denote the Riemann curvature tensor of $g$. Then there are constants $C_\ell$ and $\sigma > 0$ such that $\mfk_g := R_{g,\alpha\beta\mu\nu} R_g^{\alpha\beta\mu\nu}$ and $\mfr_g := \ric_{g,\alpha\beta} \ric_g^{\alpha\beta}$ satisfy 
    \begin{align*}
        \bigck[(\overline{\Omega}_0)]{\ell}{\theta(t)^{-4}\mfk_g(t) - 4\big( \tsum_I \pr_I^2(1-\pr_I)^2 + \tsum_{I < J} \pr_I^2\pr_J^2 \big)} &\leq C_\ell t^\sigma,\\
        \ck[(\overline{\Omega}_0)]{\ell}{\theta(t)^{-4} \mfr_g(t) - \psir^4} &\leq C_\ell t^\sigma
    \end{align*}
    for all $t \in (0,T]$ and all $\ell$, so that $(M,g)$ is $C^2$ inextendible across the boundary $\{0\} \times \overline{\Omega}_0$. Moreover, $J^-\big(\iota(0)\big) \subset \Xi\big((0,T] \times \Omega_0\big)$ and every inextendible causal geodesic starting at $\iota(0)$ is past incomplete. 
\end{theorem}

\begin{remark}
    An analog global statement also holds, for initial data defined on a closed (compact without boundary) manifold $\Sigma$. Fix a reference Riemannian metric $\refmetric$ on $\Sigma$. The constant $\zeta_1$ is now also allowed to depend on $(\Sigma,\refmetric)$. The definition of the initial time and the lower bound on the initial mean curvature change to
    \[
    \frac{1}{T} := \frac{1}{\vol(\Sigma)} \int_\Sigma \bar\theta\,\muref > \zeta_1,
    \]
    where $\vol(\Sigma)$ denotes the volume of $\Sigma$ with respect to $\refmetric$ and $\muref$ is the volume form associated with $\refmetric$. The assumed bound on $d\bar\theta$ changes accordingly to
    \[
    T^{1-4\delta}\sob[(\Sigma)]{k_1+1}{d\bar\theta} \leq \zeta_0
    \]
    with $T$ as above. The manifold where we construct the solution changes from $\Omega_{(0,T]}$ to $(0,T] \times \Sigma$. The statement of the global result is very similar to \cite[Theorem~11]{franco_complete_asymptotics_2026}, so we do not spell it out here. The important difference is that the foliation that we construct is not CMC. The proof of the global result requires only minor modifications from the localized proof. We will make appropriate remarks throughout the proof to highlight the most relevant differences.
\end{remark}

\begin{remark}
    The only place where the restriction to 4 spacetime dimensions is relevant is in the propagation of constraints for short-time existence. In particular, Lemma~\ref{cyclic sum lemma} below. Otherwise, the proof would work in $n+1$ spacetime dimensions, $n \geq 3$, essentially unchanged. 
\end{remark}

\begin{remark}
    The condition that the $\bar p_I$ be distinct is required in two places. First, for the existence of an appropriate orthonormal frame for the initial metric $\bar h$; see Section~\ref{ssec: bounds on initial data}. And second, for the application of \cite[Theorem~24]{franco_complete_asymptotics_2026} that ensures the existence of the data on the singularity.
\end{remark}

\begin{remark}
    The condition that the potential $V$ be nonnegative is only needed to prove Lemma~\ref{lemma: initial estimate on norm der of phi} below.
\end{remark}

\subsection{Related works}

Hawking's singularity theorem asserts that if we have initial data for Einstein's equations, satisfying the timelike convergence condition and with mean curvature bounded below by a positive constant, then the corresponding maximal globally hyperbolic development has a singularity in the sense of past timelike geodesic incompleteness. However, the nature of the incompleteness remains unclear. As we mentioned at the beginning, in order to resolve this question, it is natural to look at the heuristics of BKL, which propose that the singularity should be local in space and, in general, oscillatory. Mathematical results in the oscillatory setting have only been achieved for spatially homogeneous solutions of Einstein's equations; see, e.g.,\cite{weaver_dynamics_2000,ringstrom_curvature_2000,ringstrom_bianchi_2001,beguin_aperiodic_2010,beguin_chaotic_2023,brehm_bianchi_2016,liebscher_ancient_2011,liebscher_oscillatory_2013}. However, as pointed out in \cite{belinskii_effect_1973,barrow_quiescent_1978,demaret_non-oscillatory_1985}, in certain special cases, such as in the presence of a scalar field or a stiff fluid, or in the vacuum setting with spacetime dimensions $\geq 11$, the asymptotics are expected to be quiescent in the direction of the singularity. In fact, the subcriticality condition \eqref{eq: subcritical condition} originates in \cite{demaret_non-oscillatory_1985}.

There has been a lot of progress towards the understanding of quiescent big bang singularities. See, e.g., \cite{chrusciel_strong_1990,ringstrom_existence_2006,ringstrom_strong_2009} for results in the $\mathbb{T}^3$-Gowdy setting. Furthermore, several results have appeared where the authors introduce a notion of data on the singularity, and then prove the existence of a corresponding solution to Einstein's equations; see, e.g., \cite{ames_quasilinear_2013,ames_class_2017,andersson_quiescent_2001,damour_kasner-like_2002,isenberg_asymptotic_1999,isenberg_asymptotic_2002,kichenassamy_analytic_1998,klinger_new_2015,rendall_fuchsian_2000,stahl_fuchsian_2002} for results in symmetric settings and in the analytic setting. The first result of this type, without symmetry or analyticity assumptions, was introduced by Fournodavlos and Luk in \cite{fournodavlos_asymptotically_2023}. This result was localized by Athanasiou and Fournodavlos in \cite{athanasiou_localized_2025}. However, there is one issue with these results. The notions of data on the singularity that the authors introduce are tied to the specific gauge that they use. This makes it difficult to see how all of these results are related to each other. To resolve this issue, in \cite{ringstrom_wave_2025,ringstrom_geometry_2026}, Ringström introduced a geometric framework for understanding big bang singularities. Furthermore, in \cite{ringstrom_initial_2025}, he introduced a geometric notion of initial data on big bang singularities, which is expected to provide a unified understanding of the previously mentioned results. In \cite{franco-grisales_developments_2025} it was verified that given geometric initial data on the singularity, as in \cite{ringstrom_initial_2025}, there is a corresponding development, and that developments of the same data on the singularity are isometric in a neighborhood of the singularity. That geometric initial data on the singularity can indeed be used to completely characterize quiescent big bang solutions was verified for the spatially homogeneous Bianchi class A spacetimes in \cite{ringstrom_initialsymmetric_2024,ringstrom_structure_2025}.  

The stability of big bang singularities is another area where a lot of progress has been made recently. This started with the works \cite{rodnianski_regime_2018,rodnianski_stable_2018} by Rodnianski and Speck about the past stability of the FLRW solutions. In \cite{speck_maximal_2018}, Speck proved a similar result with $\mathbb{S}^3$ spatial topology. These results were complemented in \cite{fajman_cosmic_2025}, where Fajman and Urban treated the negative spatial curvature case. In \cite{rodnianski_nature_2021}, Rodnianski and Speck studied the past stability of moderately anisotropic Kasner spacetimes in high dimensions. And finally, they were joined by Fournodavlos in \cite{fournodavlos_stable_2023}, where they proved the stability of the big bang singularity of the Kasner--scalar field spacetimes, in the full subcritical range in which stability is expected to occur. The key new ingredient in \cite{fournodavlos_stable_2023} was the introduction of a Fermi-Walker frame, which has been used extensively in later works. These results have been extended to include other matter models; see \cite{fajman_past_2025,an_stability_2025}. We also mention the stability result \cite{urban_quiescent_2024} in 3 spacetime dimensions.

One shortcoming of the stability results we have mentioned so far, is that they only prove big bang formation for initial data that are close to that of specific spatially homogeneous solutions. Recently, in \cite{oude_groeniger_formation_2023}, Oude Groeniger, Petersen and Ringström introduced the first big bang formation result that does not make reference to any background solution. Finally, connecting these results with those that start from data on the singularity, in \cite{franco_complete_asymptotics_2026}, together with Ringström, we obtained complete asymptotics for the solutions of \cite{oude_groeniger_formation_2023}. Specifically, we proved that these solutions induce geometric initial data on the singularity as in \cite{ringstrom_initial_2025} or \cite{franco-grisales_developments_2025}.

As we discussed at the beginning, all previously mentioned big bang formation and stability results use CMC foliations for their proofs, which means that they cannot be localized in space in a straightforward way. Efforts to overcome this issue started with Beyer and Oliynyk in \cite{beyer_localized_2023}, where they proved localized big bang stability of FLRW solutions with a scalar field, by using the scalar field as a time function. They later extended this result to include a perfect fluid in \cite{beyer_past_2024}. Together with Zheng in \cite{beyer_localized_2025}, they proved localized past stability of the 4-dimensional subcritical Kasner--scalar field spacetimes. Similar methods were used in the big bang stability results \cite{dong_paststability_2026,beyer_isotropisation_2026,zheng_localized_2026}.

\subsection{Strategy for the proof}

\paragraph{Formulation of Einstein's equations.} We use a foliation by spacelike hypersurfaces with zero shift, given by the level sets of a time function $t$. So that the metric takes the form
\[
g = -N^2dt \otimes dt + h,
\]
where $h$ denotes the family of induced metrics on the level sets of $t$ and $N>0$ denotes the lapse function. Similarly as in \cite{fournodavlos_stable_2023,oude_groeniger_formation_2023}, we treat the family of induced metrics $h$ by introducing a family of orthonormal frames $\{e_I\}$ that satisfies the Fermi-Walker transport equation. The key new ingredient for the proof is our choice of time function $t$. We want to choose $t$ such that the singularity is synchronized at $t = 0$. The formation of a big bang singularity is signaled by the blow up of the mean curvature, so we should choose $t$ in such a way that the mean curvature of its level sets can be expected to blow up as $t \to 0$. For that purpose, we require $t$ to satisfy the second order differential equation
\[
\Box_g \ln t = \frac{\lpar}{tN^2}\Big(N\theta-\frac{1}{t}\Big),
\]
where $\theta$ denotes the mean curvature of the level sets of $t$ and $\lpar > 0$ is a constant. Even though this looks like a wave equation for $\ln t$, note that it is not, since $\theta$ contains second derivatives of $t$. The crucial property of this time function is the following. If we define
\[
\Theta := N\theta,
\]
then, as a consequence of Einstein's equations, $\Theta$ and the lapse $N$ satisfy equations of the form
\[
\p_t(t\Theta-1) = \frac{\lpar}{t}(t\Theta-1) + \cdots, \qquad \p_tN = (\lpar + 1)\Big( \Theta - \frac{1}{t} \Big)N,
\]
where $\cdots$ stands for terms that are expected to be asymptotically negligible. What this means is that $t\Theta-1$ is expected to converge to zero as $t^\sigma$, for some $\sigma > 0$, as $t \to 0$. In other words, $\Theta$ is expected to be asymptotic to $\frac{1}{t}$. While at the same time, $N$ is expected to remain bounded, since the expression in front of $N$ on the right-hand side of the corresponding equation would be of order $t^{-1+\sigma}$, and thus integrable in time all the way to $t = 0$. If these expectations are realized, then the mean curvature $\theta$ along the foliation would blow up as $t \to 0$, so that our choice of $t$ can indeed be expected to synchronize the singularity. We go over the details of our formulation of Einstein's equations in Section~\ref{sec: gauge}.

\paragraph{Short-time existence.} The first step of the proof is to verify that our formulation indeed yields a system of fully hyperbolic equations, that can be used to prove short-time existence and a continuation principle for Einstein's equations on a localized spacetime domain as in Definition~\ref{def: spacetime domain}. This is done in Section~\ref{sec: short time existence}. It turns out that, as it stands, the system of equations that we obtain is not hyperbolic. See Proposition~\ref{prop: reduced equations struct coeffs} below. However, following the ideas of \cite{beyer_localized_2025,fournodavlos21}, we show that it is possible to add multiples of the Hamiltonian and momentum constraints to the evolution equations, in such a way that we obtain a symmetric hyperbolic system; see Subsection~\ref{ssec: symmetric hyperbolic formulation}. Due to these modifications, the step from a solution to the modified equations to a solution to Einstein's equations is non trivial. We show how to do this in Subsections~\ref{constructing a solution to einsteins equations} and \ref{ssec: propagation of constraints}. The main difficulty lies in showing that the constraint equations are satisfied. To do this we follow the ideas of \cite{fournodavlos21}. The idea is that, after defining the metric $g$, one can use the solution to define a connection on the spacetime. This connection is compatible with $g$ from the beginning, however, it is a priori not known whether it is torsion free. One can then use the analogs of the Bianchi identities for a connection with torsion, to derive a linear and homogeneous symmetric hyperbolic system for the constraint quantities and the torsion. It follows that all these quantities vanish, as long as they vanish initially, and one gets a solution to Einstein's equations in the end. This argument for the propagation of constraints is done in Subsection~\ref{ssec: propagation of constraints}. Finally, in Subsection~\ref{short time existence localized}, we prove short-time existence for Einstein's equations in a domain as in Definition~\ref{def: spacetime domain}. This is done through an application of the general results of \cite{schochet_euler_1986}, for symmetric hyperbolic systems.   

\paragraph{Global existence.} The next step is to prove that, under the assumptions of Theorem~\ref{thm: main theorem}, one gets global existence in a domain of the form $\Omega_{(0,T]}$. This is done through a bootstrap argument that is, for the most part, similar to that of \cite{oude_groeniger_formation_2023} and \cite{fournodavlos_stable_2023}. The basic idea is to control a low order energy $\bbl$, given in terms of $C^2$ norms, and a high order energy $\bbh$, given in terms of $H^{k_1}$ norms. The estimates for $\bbh$ are done, mostly, in the standard way for symmetric hyperbolic systems. However, in order to close the argument, one needs to use both the Hamiltonian and momentum constraints; see Lemma~\ref{lemma: constraint equations trick}. The estimates for $\bbl$ are of ODE type, where a loss of derivatives is avoided by using interpolation inequalities and estimating in terms of $\bbh$; see Lemma~\ref{interpolation lemma}. Now we highlight the two main differences with the arguments of \cite{oude_groeniger_formation_2023}. The first obvious difference is that, since we do not use a CMC foliation, we do not need to do elliptic estimates, which is what makes the localization straightforward; and second, when applying the divergence theorem for the high order energy estimates, we need to deal with integrals on the ``side" part of the boundary of $\Omega_{(0,T]}$. Let $\nu$ denote a unit normal to the side boundary and let $\{e_I\}$ be the family of Fermi-Walker transported orthonormal frames for $h$. Since we expect a bound of the form 
\[
\ck[(\overline{\Omega}_t)]{2}{e_I} \leq Ct^{-1+4\delta},
\]
then we can show that
\[
g(\nu,\nu) \leq -\frac{1}{N^2}t^{2(-1+3\delta)} + \tsum_I g(\nu,e_I)^2 \leq -\displaystyle\frac{1}{N^2}t^{2(-1+3\delta)} + Ct^{2(-1+4\delta)},
\]
for some constant $C$. This means that if $N$ is bounded and the initial time $T$ is small enough, then $\nu$ is timelike. Hence, the side boundary is spacelike and ingoing, and the total contribution of the boundary integrals has a good sign. Geometrically, this is reflected in  Theorem~\ref{thm: main theorem}, in the statement that $\Xi(\Omega_{(0,T]})$ is contained in the past domain of dependence of $\iota(\Omega_T)$. In other words, what happens on the side boundary of $\Omega_{(0,T]}$ can be predicted from the initial data, without the need for additional boundary conditions. This fact is already present in the proof of short-time existence of course, although in a somewhat less evident way. This is precisely the role of the constant $\mathfrak{b}$ in Theorem~\ref{thm: short time existence}. 

We note that our global existence proof is quite different from those of the other localized results \cite{beyer_localized_2023,beyer_localized_2025,zheng_localized_2026}. In these works, the authors apply general Fuchsian results to obtain global existence of solutions, as opposed to the more direct approach that we take. The global existence result is discussed in Section~\ref{sec: global existence}. Section~\ref{sec: energy estimates} is entirely devoted to deriving the energy estimates required for the proof of global existence. 

In the global existence result, the estimates that we obtain for the variables are only for a finite number of derivatives. So the next step is to upgrade these conclusions to estimates for derivatives of all orders. This is done in a similar way as in \cite[Section~2]{franco_complete_asymptotics_2026}. The energy estimates for this part of the proof are similar but simpler that those for the proof of global existence. The main differences are that the argument is inductive in this case, and that there is no low order energy, only a Sobolev energy. This is the purpose of Section~\ref{sec: higher order estimates}.

\paragraph{Proof of the main result.} Finally, the proof of Theorem~\ref{thm: main theorem} is done in Section~\ref{sec: proof of main theorem}. The only thing left to do at this point is to derive the asymptotics. We start by obtaining the asymptotics for the eigenvalues of $\mck$ and for the scalar field. Then we are in a position to apply \cite[Theorem~24]{franco_complete_asymptotics_2026}, which implies the existence of the initial data on the singularity. The proofs about curvature blow up, causality and geodesic incompleteness of the solution are straightforward consequences of the previously obtained estimates.

\subsection*{Acknowledgements}

The author would like to thank Hans Ringström for suggesting the topic, the helpful discussions and the careful reading of the manuscript. This research was funded by the Swedish Research Council (Vetenskapsrådet), dnr. 2022-03053; and supported by foundations managed by The Royal Swedish Academy of Sciences.

\section{Formulation of the Einstein--nonlinear scalar field equations} \label{sec: gauge}

Let $(M,g)$ be a 4-dimensional spacetime. For our formulation of Einstein's equations, we assume that $(M,g)$ is equipped with a foliation $\{\Sigma_t\}$ by spacelike hypersurfaces, given by the level sets of a time function $t$, such that the metric $g$ takes the form
\[
g = -N^2 dt \otimes dt + h,
\]
where $h$ denotes the family of induced metrics on the $\Sigma_t$ hypersurfaces and $N$ is called the \emph{lapse function}. We choose the time orientation of the spacetime to be such that $t$ is increasing along future directed causal curves. Denote by $\n$ and $\sn$ the Levi-Civita connections of $g$ and $h$ respectively, and by $e_0 := \frac{1}{N} \p_t$ the future pointing unit normal of the $\Sigma_t$ hypersurfaces. If $X,Y \in \mfx(M)$ are tangential to the $\Sigma_t$, then the family of second fundamental forms of the $\Sigma_t$ is given by
\[
k(X,Y) := g(\n_X e_0,Y).
\]
Denote by $K := k^\sharp$ the Weingarten map and by $\theta := \tr_h k = \tr K$ the mean curvature. Assume that there are vector fields $E_i$ and one forms $\eta^i$, for $i = 1,2,3$, on $M$, such that $[\p_t,E_i] = 0$ for all $i$, and $\{E_i|_{\Sigma_t}\}$ is a frame for $\Sigma_t$ for all $t$, with dual frame $\{\eta^i|_{\Sigma_t}\}$. The family of frames $\{E_i\}$ serves the role of a fixed time independent reference frame. The normal Lie derivative of $K$ is then given by
\[
\lie_{e_0}K = e_0(K_i{}^j) \eta^i \otimes E_j;
\]
see \cite[Subsection~A.2, pp. 225-227]{ringstrom_wave_2025}. In this setting, Einstein's equations take the following form.

\begin{proposition} \label{prop: ADM equations}
    Under the assumptions just described, let $\s \in C^\infty(M)$. Then $(M,g,\s)$ solves the Einstein--nonlinear scalar field equations with potenial $V$, if and only if the following system of equations holds. The evolution equations for $K$ and $\s$:
    \begin{subequations}
        \begin{align}
            \lie_{e_0} K + \ric_h^\sharp + \theta K - \tfrac{1}{N}\sn^2N^\sharp &= d\s \otimes \sn \s + (V \circ \s)I,\label{eq: general equation for k}\\
            -e_0e_0 \s + \Delta_h \s - \theta e_0 \s + h\big(d\s,d(\ln N)\big) &= V' \circ \s,\label{eq: general equation for phi}
        \end{align}
    \end{subequations}
    where $I$ denotes the family of identity $(1,1)$-tensors on the $\Sigma_t$ hypersufraces and $\sn \s$ denotes the gradient of $\s$ with respect to $h$; and the constraint equations:
    \begin{subequations} \label{eq: general constraint equations}
        \begin{align}
            S_h + \theta^2 - \tr K^2 &= e_0(\s)^2 + |d\s|_h^2 + 2V \circ \s,\label{hamiltonian constraint}\\
            \diver_h K - d\theta &= e_0(\s)d\s.\label{momentum constraint}
        \end{align}
    \end{subequations}
    Equation~\eqref{hamiltonian constraint} is called the Hamiltonian constraint equation and Equation~\eqref{momentum constraint} is called the momentum constraint equation.
\end{proposition}

\begin{proof}
    The equivalence between \eqref{eq: matter equation} and \eqref{eq: general equation for phi} follows by a direct computation. That \eqref{eq: general equation for k} and \eqref{eq: general constraint equations} are equivalent to \eqref{eq: einsten equation} is a direct consequence of \cite[Lemma~113 and Remark~114]{ringstrom_geometry_2021}, \cite[Proposition~13.3]{ringstrom_cauchy_2009} and the definition of $T$. Recall that we set $\Lambda = 0$.
\end{proof}

In order to fix the lapse, we introduce the following second order differential equation for the time function $t$:
\begin{equation} \label{equation for ln t}
    \Box_g \ln t = \frac{\lpar}{t N^2}\Big(N\theta - \frac{1}{t}\Big),
\end{equation}
where $\lpar > 0$ is a constant. 

Next, following \cite{fournodavlos_stable_2023,oude_groeniger_formation_2023}, we introduce a Fermi-Walker frame. That is, an orthonormal frame $\{e_I\}$, for $I = 1,2,3$, for the family of induced metrics $h$ satisfying the Fermi-Walker transport equation
\begin{equation} \label{fermi walker transport equation}
    \n_{e_0} e_I = e_I(\ln N)e_0.
\end{equation}
Note that Equation~\eqref{fermi walker transport equation} ensures that if the frame $\{e_I\}$ is tangent to $\Sigma_{t_0}$ and orthonormal, at some initial time $t_0$, then it remains tangent to $\Sigma_t$ and orthonormal for all $t$. From now on, we adopt the convention to sum over repeated capital indices, regardless of whether they are both lower or upper indices. The only exception to this rule is when the functions $\bar p_I$ are involved. In that case, we always state explicitly when there is a sum over the corresponding index. In order to obtain a system of equations that we can work with, we also need to consider the structure coefficients of the frame $\{e_I\}$,
\[
\gamma_{IJK} := g([e_I,e_J],e_K).
\]
Then we arrive at the following formulation of Einstein's equations.

\begin{proposition} \label{prop: reduced equations struct coeffs}
    Under the assumptions described above, let $\s \in C^\infty(M)$ and $V \in C^\infty(\R)$. Suppose that we have vector fields $\{e_I\}$, for $I = 1,2,3$, on $M$, tangent to the $\Sigma_t$ hypersurfaces, such that for each $t$ they form an orthonormal frame for the family of induced metrics $h$. Define $k_{IJ} := k(e_I,e_J)$ and $\gamma_{IJK} := g([e_I,e_J],e_K)$. Let $\{\omega^I\}$ be the dual frame of $\{e_I\}$. Finally, let $e_I = e_I^iE_i$ and $\omega^I = \omega^I_i\eta^i$. Then $(M,g,\s)$ is a solution to the Einstein--nonlinear scalar field equations with potential $V$, the time function $t$ satisfies \eqref{equation for ln t} and $\{e_I\}$ is a Fermi-Walker frame if and only if the following system of equations holds:\vspace{0.2cm}\\
    The evolution equations for the frame and the dual frame:
    \begin{subequations} \label{frame equations reduced}
        \begin{align} 
            e_0(e_I^i) &= -k_{IJ}e_J^i,\label{frame equation global}\\
            e_0(\omega^I_i) &= k_{IJ}\omega^J_i.\label{dual frame equation global}
        \end{align}
    \end{subequations}
    The lapse equation:
    \begin{equation} \label{lapse equation global}
        e_0 N = (\lpar+1)\Big( N\theta - \frac{1}{t} \Big).
    \end{equation}
    The evolution equations for the second fundamental form and the structure coefficients:
    \begin{subequations}
        \begin{align}
        \begin{split}
            e_0k_{IJ} &= -e_K\gamma_{K(IJ)} - e_{(I}\gamma_{J)KK} + e_{(I}e_{J)}(\ln N) - \theta k_{IJ}\\
            &\quad + \gamma_{I(KL)}\gamma_{J(KL)} - \tfrac{1}{4}\gamma_{KLI}\gamma_{KLJ} + \gamma_{KLL}\gamma_{K(IJ)}\\
            &\quad + e_I(\ln N)e_J(\ln N) - \gamma_{K(IJ)}e_K(\ln N) + e_I(\s)e_J(\s) + (V\circ\s)\delta_{IJ},
        \end{split}\label{k equation global}\\
        \begin{split}
            e_0\gamma_{IJK} &= e_Jk_{IK} - e_Ik_{JK} + k_{KL}\gamma_{IJL} + k_{IL}\gamma_{JLK}\\
            &\quad - k_{JL}\gamma_{ILK} + e_J(\ln N)k_{IK} - e_I(\ln N)k_{JK}.
        \end{split}\label{gamma equation global}
        \end{align}
    \end{subequations}
    The wave equation for the scalar field:
    \begin{equation} \label{eq: phi equation global}
        e_0e_0\s  = e_Ie_I\s - \theta e_0\s - \gamma_{JII}e_J\s + e_I(\s)e_I(\ln N) - V' \circ \s.
    \end{equation}
    The Hamiltonian constraint:
    \begin{equation} \label{hamiltonian constraint with structure coeffs}
        \begin{split}
        2e_I\gamma_{IJJ} - \tfrac{1}{4} \gamma_{IJK}\gamma_{IJK} - \tfrac{1}{2}\gamma_{IJK}\gamma_{IKJ} - \gamma_{IJJ}&\gamma_{IKK} + \theta^2 - k_{IJ}k_{IJ}\\
        &= e_0(\s)^2 + e_I(\s)e_I(\s) + 2V\circ\s.
        \end{split}
    \end{equation}
    The momentum constraint:
    \begin{equation} \label{momentum constraint with structure coeffs}
        e_Jk_{JI} - \gamma_{JKK}k_{JI} - \gamma_{JIK}k_{JK} - e_I\theta = e_0(\s)e_I(\s).
    \end{equation}
\end{proposition}

\begin{proof}
    The equations \eqref{frame equations reduced} are a straightforward consequence of \eqref{fermi walker transport equation} and the fact that the matrix with components $e_I^i$ is the inverse of the matrix with components $\omega_i^I$. A direct computation shows that
    \[
    e_0 N = N\theta - \frac{1}{t} + tN^2 \Box_g \ln t. 
    \]
    It follows that $t$ solves \eqref{equation for ln t} if and only if \eqref{lapse equation global} is satisfied. As a consequence of \cite[Lemma~72]{oude_groeniger_formation_2023}, the spatial Ricci and scalar curvatures are given by
    \begin{align*}
        \ric_h(e_I,e_J) &= e_K\gamma_{K(IJ)} + e_{(I}\gamma_{J)KK} - \gamma_{I(KL)}\gamma_{J(KL)} + \tfrac{1}{4}\gamma_{KLI}\gamma_{KLJ} - \gamma_{KLL}\gamma_{K(IJ)},\\
        S_h &= 2e_I\gamma_{IJJ} - \tfrac{1}{4}\gamma_{IJK}\gamma_{IJK} - \tfrac{1}{2}\gamma_{IJK}\gamma_{IKJ} - \gamma_{IJJ}\gamma_{IKK}.
    \end{align*}
    Furthermore, we can write
    \[
    \tfrac{1}{N}\sn^2N(e_I,e_J) = e_{(I}e_{J)}(\ln N) + e_I(\ln N)e_J(\ln N) - \gamma_{K(IJ)}e_K(\ln N).
    \]
    Therefore, \eqref{k equation global}, \eqref{hamiltonian constraint with structure coeffs} and \eqref{momentum constraint with structure coeffs} are direct consequences of Proposition~\ref{prop: ADM equations}. Note that if $X \in \mfx(M)$ and $\overline{X}$ denotes the part of $X$ that is tangent to the $\Sigma_t$, then
    \[
    \begin{split}
        e_0k_{IJ} = e_0 g(K(e_I),e_J) &= g(\lie_{e_0}K(e_I),e_J) + 2k(K(e_I),e_J) +  k(\overline{[e_0,e_I]},e_J) + k(e_I,\overline{[e_0,e_J]})\\
        &= (\lie_{e_0}K)_I{}^J + 2k_{IK}k_{KJ} - k(\n_{e_I}e_0,e_J) - k(e_I,\n_{e_J}e_0),
    \end{split}
    \]
    where we have used \eqref{fermi walker transport equation}, and the last three terms on the far right-hand side cancel out. Next, \eqref{gamma equation global} follows by the Jacobi identity and \eqref{fermi walker transport equation}. Finally, \eqref{eq: phi equation global} follows directly from \eqref{eq: general equation for phi}.
\end{proof}

\section{Short-time existence} \label{sec: short time existence}

The main result of this section is Theorem~\ref{thm: short time existence} below, which is a short-time existence statement along with a continuation principle, for Einstein's equations in the formulation of Proposition~\ref{prop: reduced equations struct coeffs}, in spacetime domains as in Definition~\ref{def: spacetime domain}. According to the discussion of Section~\ref{sec: gauge}, it makes sense to look for a solution to the Einstein--nonlinear scalar field equations by looking for solutions to the equations of Proposition~\ref{prop: reduced equations struct coeffs}. The issue is that, as it stands, this system of equations is not hyperbolic. This can be seen, for instance, by looking at the second term on the right-hand side of \eqref{k equation global}, which does not have a symmetric counterpart in \eqref{gamma equation global}. We begin by modifying the evolution equations in such a way that we obtain a symmetric hyperbolic system. Then we show that we can do this in such a way that the constraints satisfy a linear and homogeneous symmetric hyperbolic system. These facts, plus general results for symmetric hyperbolic systems, allow us to prove Theorem~\ref{thm: short time existence}.

\subsection{A symmetric hyperbolic formulation of the evolution equations} \label{ssec: symmetric hyperbolic formulation}

We work in the setting described in Section~\ref{sec: gauge}. Introduce the connection coefficients of the frame $\{e_I\}$,
\[
\Gamma_{IJK} := g(\n_{e_I}e_J,e_K) = h(\sn_{e_I}e_J,e_K).
\]
Note that the $\Gamma_{IJK}$ satisfy $\Gamma_{IJK} = -\Gamma_{IKJ}$. The connection coefficients are related with the structure coefficients by the relations
\begin{equation} \label{eq: relation between structure and connection coeffs}
    \Gamma_{IJK} = \tfrac{1}{2}(\gamma_{IJK} + \gamma_{KIJ} - \gamma_{JKI}), \qquad \gamma_{IJK} = \Gamma_{IJK} - \Gamma_{JIK},
\end{equation}
which follow from the Koszul formula and the fact that $\n$ is torsion free. For the proof of short-time existence, it turns out to be more convenient to work with the connection coefficients, rather than with the structure coefficients. For that reason, we introduce the following alternative to Proposition~\ref{prop: reduced equations struct coeffs}.

\begin{proposition} \label{reduced equations}
    Under the assumptions of Proposition~\ref{prop: reduced equations struct coeffs}, let $\Gamma_{IJK} := g(\n_{e_I}e_J,e_K)$. Then $(M,g,\s)$ is a solution to the Einstein--nonlinear scalar field equations with potential $V$, the time function $t$ satisfies \eqref{equation for ln t} and $\{e_I\}$ is a Fermi-Walker frame if and only if the following system of equations holds:\vspace{0.2cm}\\
    The evolution equation for the second fundamental form and the connection coefficients:
    \begin{subequations}
        \begin{align}
        \begin{split}
            e_0 k_{IJ} &= -e_K\Gamma_{IJK} + e_I\Gamma_{KJK} + e_Ie_J(\ln N) - \theta k_{IJ} - \Gamma_{IJL}\Gamma_{KLK} + \Gamma_{KIL}\Gamma_{LJK} \\
            &\quad + e_I(\ln N)e_J(\ln N) - \Gamma_{IJK}e_K(\ln N) + e_I(\s)e_J(\s) + (V \circ \s)\delta_{IJ}, 
        \end{split}\label{k equation reduced}\\
        \begin{split}
            e_0\Gamma_{IJK} &= e_Jk_{KI} - e_Kk_{JI} + k_{KL}\Gamma_{JLI} - k_{JL}\Gamma_{KLI} - k_{IL}\Gamma_{JKL}\\
            &\quad + k_{IL}\Gamma_{KJL} - k_{IL}\Gamma_{LJK} + e_J(\ln N)k_{IK} - e_K(\ln N)k_{IJ}.
        \end{split}\label{gamma equation}
        \end{align}
    \end{subequations}
    The wave equation for the scalar field:
    \begin{equation} \label{phi equation reduced}
        e_0e_0\s = e_Ie_I\s - \theta e_0\s - \Gamma_{IIJ}e_J\s + e_I(\s)e_I(\ln N) - V' \circ \s.
    \end{equation}
    The constraint equations:
    \begin{subequations} \label{constraint equations reduced}
        \begin{align} 
            2e_K\Gamma_{IIK} + \Gamma_{IIL}\Gamma_{KLK} - \Gamma_{KIL}\Gamma_{LIK} + \theta^2 - k_{IJ}k_{IJ} &= e_0(\s)^2 + e_I(\s)e_I(\s) + 2V \circ \s,\label{hamiltonian constraint reduced}\\
            e_Jk_{JI} - \Gamma_{JJK}k_{KI} - \Gamma_{JIK}k_{JK} - e_I\theta &= e_0(\s)e_I(\s).\label{momentum constraint reduced}
        \end{align}
    \end{subequations}
    In addition to \eqref{frame equations reduced} and \eqref{lapse equation global}.
\end{proposition}

\begin{proof}
    Note that in terms of the connection coefficients, the spatial Ricci and scalar curvatures are given by
    \begin{align*}
        \ric_h(e_I,e_J) &= e_K\Gamma_{IJK} - e_I\Gamma_{KJK} + \Gamma_{IJL}\Gamma_{KLK} - \Gamma_{KIL}\Gamma_{LJK}\\
        S_h &= 2e_K\Gamma_{IIK} + \Gamma_{IIL}\Gamma_{KLK} - \Gamma_{KIL}\Gamma_{LIK}.
    \end{align*}
    Equations~\eqref{k equation reduced}, \eqref{hamiltonian constraint reduced} and \eqref{momentum constraint reduced} now follow directly from Proposition~\ref{prop: ADM equations}. In order to deduce \eqref{gamma equation}, we compute
    \[
    \begin{split}
        e_0\Gamma_{IJK} = e_0 g(\n_{e_I}e_J,e_K) &= g(\n_{e_I}\n_{e_0}e_J,e_K) + R_g(e_0,e_I,e_J,e_K)\\
        &\quad + g(\n_{[e_0,e_I]}e_J,e_K) + g(\n_{e_I}e_J,\n_{e_0}e_K).
    \end{split}
    \]
    We deal with the curvature term, 
    \[
    \begin{split}
        R_g(e_0,e_I,e_J,e_K) &= R_g(e_J,e_K,e_0,e_I)\\
        &= g\big( \n_{e_J}\n_{e_K}e_0 - \n_{e_K}\n_{e_J}e_0 - \n_{[e_J,e_K]}e_0,e_I \big)\\
        &= g\big( \n_{e_J}(k_{KL}e_L) - \n_{e_K}(k_{JL}e_L) - \n_{\sn_{e_J}e_K}e_0 + \n_{\sn_{e_K}e_J}e_0,e_I \big)\\
        &= e_Jk_{KI} + k_{KL}\Gamma_{JLI} - e_Kk_{JI} - k_{JL}\Gamma_{KLI} - \Gamma_{JKL}k_{LI} + \Gamma_{KJL}k_{LI}.
    \end{split}
    \]
    Equation~\eqref{gamma equation} now follows in a straightforward way by using \eqref{fermi walker transport equation}. Finally, \eqref{phi equation reduced} follows by a direct computation.
\end{proof}

In order to formulate our equations, we introduce some notation. Working in the setting of Proposition~\ref{reduced equations}, define the following variables:
\[
N_I := e_I(\ln N), \qquad \s_t := \p_t \s, \qquad \s_I := e_I \s.
\]
From Equations~\eqref{lapse equation global} and \eqref{phi equation reduced}, they satisfy
\begin{align}
    e_0 N_I &= (\lpar+1)e_I\theta + (\lpar+1)\theta N_I - k_{IJ} N_J,\label{preliminary equation for N_I}\\
    \p_t \s_t &= N^2 e_I\s_I - N\theta \s_t + (\lpar+1)\Big( N\theta - \frac{1}{t} \Big)\s_t - N^2\Gamma_{IIK}\s_K + N^2\s_IN_I - N^2 V'\circ\s,\\
    \p_t\s_I &= e_I\s_t - Nk_{IJ}\s_J,
\end{align}
where we have used that
\begin{equation} \label{eq: commutation formula}
    e_0e_If = e_Ie_0f + e_I(\ln N)e_0f - k_{IJ}e_Jf
\end{equation}
for all $f \in C^\infty(M)$. Define the \emph{shear} $\Sigma$ as the trace-free part of the second fundamental form $k$. That is, 
\[
\Sigma_{IJ} = k_{IJ} - \tfrac{1}{3}\theta\delta_{IJ},
\]
where $\Sigma_{IJ} := \Sigma(e_I,e_J)$. Then, due to \eqref{k equation reduced} and with the notation introduced in Subsection~\ref{ssec: conventions appendix}, $\Sigma_{IJ}$ and $\theta$ satisfy the equations
\begin{align}
    \begin{split}
        e_0 \Sigma_{IJ} &= -e_K\Gamma_{\langle IJ \rangle K} + e_{\langle I|}\Gamma_{K|J\rangle K} + e_{\langle I}N_{J\rangle} - \theta \Sigma_{IJ} - \Gamma_{\langle IJ\rangle L}\Gamma_{KLK}\\
        &\quad + \Gamma_{K\langle I|L}\Gamma_{L|J\rangle K} + N_{\langle I}N_{J\rangle} - \Gamma_{\langle IJ\rangle K}N_K + \s_{\langle I}\s_{J\rangle}, 
    \end{split}\\
    \begin{split}
            e_0 \theta &= -2e_K \Gamma_{IIK} + e_IN_I - \theta^2 - \Gamma_{IIL}\Gamma_{KLK} + \Gamma_{KIL}\Gamma_{LIK}\\
            &\quad + N_IN_I - \Gamma_{IIK}N_K + \s_I \s_I + 3V \circ \s.
    \end{split}\label{preliminary equation for theta}
\end{align}
In order to make the system symmetric hyperbolic, we use ideas similar to those of \cite{fournodavlos21,beyer_localized_2025}. We add appropriate multiples of the constraint quantities given by \eqref{constraint equations reduced}, to \eqref{gamma equation}, \eqref{preliminary equation for theta} and \eqref{preliminary equation for N_I}. Moreover, since after solving the modified equations, it is not known a priori whether $\Sigma_{IJ}$ is trace-free, we also need to introduce some terms involving $\Sigma_{II}$ to \eqref{gamma equation} and \eqref{preliminary equation for N_I}. We introduce the parameters
\begin{equation*}
    \alpha = \alpha(\lpar,\rho) := \frac{\lpar+1}{\rho + \frac{2}{3}}, \qquad \beta = \beta(\rho) := \frac{1}{6\rho} + 1,
\end{equation*}
for $\rho > 0$. Then the system of equations that we wish to solve is:
\begin{subequations} \label{equations for short time existence}
    \begin{align}
        \p_te_I^i &= -Nk_{IJ}e_J^i,\\
        \p_t\omega^I_i &= Nk_{IJ}\omega^J_i,\\
        \p_tN &= (\lpar+1)\Big( N\theta - \frac{1}{t} \Big)N,\label{N equation short time}\\
        \begin{split}
            \p_t \Sigma_{IJ} &= - N\theta \Sigma_{IJ} -Ne_K\Gamma_{\langle IJ \rangle K} + Ne_{\langle I|}\Gamma_{K|J\rangle K} + Ne_{\langle I}N_{J\rangle} - N\Gamma_{\langle IJ\rangle L}\Gamma_{KLK}\\
            &\quad + N\Gamma_{K\langle I|L}\Gamma_{L|J\rangle K} + NN_{\langle I}N_{J\rangle} - N\Gamma_{\langle IJ\rangle K}N_K + N\s_{\langle I}\s_{J\rangle},
        \end{split}\label{k equation short time}\\
        \begin{split}
            \p_t\Gamma_{IJK} &= - Nk_{IL}\Gamma_{LJK} + Ne_Jk_{KI} - Ne_Kk_{JI} + Nk_{KL}\Gamma_{JLI} - Nk_{JL}\Gamma_{KLI}\\
            &\quad - Nk_{IL}\Gamma_{JKL} + Nk_{IL}\Gamma_{KJL} + NN_Jk_{IK} - NN_Kk_{IJ}\\
            &\quad + \tfrac{2}{3}\delta_{IJ}Ne_K\Sigma_{LL} - \tfrac{2}{3}\delta_{IK}Ne_J\Sigma_{LL}\\
            &\quad +\delta_{IK}\big( Ne_Lk_{LJ} - N\Gamma_{LLM}k_{MJ} - N\Gamma_{LJM}k_{LM} - Ne_J\theta - \s_t\s_J \big)\\
            &\quad -\delta_{IJ}\big( Ne_Lk_{LK} - N\Gamma_{LLM}k_{MK} - N\Gamma_{LKM}k_{LM} - Ne_K\theta - \s_t\s_K \big),
        \end{split}\label{gamma equation short time}\\
        \begin{split}
            \p_tN_I &= (\lpar+1)Ne_I\theta + (\lpar+1)N\theta N_I - Nk_{IJ}N_J - \tfrac{1}{3}\alpha Ne_I\Sigma_{JJ}\\
            &\quad + \alpha\big( Ne_Jk_{JI} - N\Gamma_{JJK}k_{KI} - N\Gamma_{JIK}k_{JK} - Ne_I\theta - \s_t\s_I \big),
        \end{split}\label{lapse derivatives equation short time}\\
        \begin{split}
            \p_t\theta &= - 2Ne_K\Gamma_{IIK} + Ne_IN_I - N\theta^2 - N\Gamma_{IIL}\Gamma_{KLK} + N\Gamma_{KIL}\Gamma_{LIK} + NN_IN_I\\
            &\quad - N\Gamma_{IIK}N_K + N\s_I\s_I +3NV \circ \s + \beta\big( 2Ne_K\Gamma_{IIK} + N\Gamma_{IIL}\Gamma_{KLK} \\
            &\quad - N\Gamma_{KIL}\Gamma_{LIK} + N\theta^2 - Nk_{IJ}k_{IJ} - \tfrac{1}{N}\s_t^2 - N\s_I\s_I - 2NV \circ \s\big),\label{theta equation short time}
        \end{split}\\
        \begin{split}
            \p_t \s_t &= N^2 e_I\s_I - N\theta \s_t + (\lpar+1)\Big( N\theta - \frac{1}{t} \Big)\s_t\\
            &\quad - N^2\Gamma_{IIK}\s_K + N^2\s_IN_I - N^2 V' \circ \s,
        \end{split}\label{phi equation short}\\
        \p_t\s_I &= e_I\s_t - Nk_{IJ}\s_J,\\
        \p_t\s &= \s_t,
    \end{align}
\end{subequations}
where we use the shorthand $k_{IJ} = \Sigma_{IJ} + \frac{1}{3}\theta\delta_{IJ}$. Now we verify that these equations are indeed hyperbolic.

\begin{definition}
    Let $\Sigma$ be a manifold and $\{E_i\}$ a global frame on $\Sigma$. Consider a system of equations of the form
    \begin{equation} \label{eq: model symmetric hyperbolic system}
        \p_t u = A^i(t,x,u) E_i u + F(t,x,u).
    \end{equation}
    Here we assume that $u = u(t,x)$ takes values in an open subset $U$ of $\R^k$; $F$ takes values in $\R^k$; $t$ takes values in an open interval $I$; $x \in \Sigma$; the $A^i$ are $k \times k$ matrices; and that $A^i$ and $F$ are smooth in $I \times \Sigma \times U$. If there is a $k \times k$ symmetric matrix $A^0(t,x,u)$, smooth in $I \times \Sigma \times U$, such that
    \begin{equation} \label{eq: lower A0 bound}
        A_0(t,x,u) \geq cI
    \end{equation}
    for some constant $c > 0$ and for all $(t,x,u) \in I \times \Sigma \times U$, and $A^0A^i$ is symmetric for all $i$, then \eqref{eq: model symmetric hyperbolic system} is called a \emph{quasilinear symmetrizable hyperbolic system}. 
\end{definition}

\begin{proposition} \label{prop: the system is hyperbolic}
    Let $a > 0$. If we restrict the values of $t$ and $N$ in \eqref{equations for short time existence} to the interval $(a,\infty)$, then the system of equations \eqref{equations for short time existence} is a quasilinear symmetrizable hyperbolic system for every $\lpar > 0$ and $\rho > 0$, with $A^0$ diagonal and depending only on $\lpar$, $\rho$ and $N$.
\end{proposition}

\begin{proof}
    For $\xi = (\xi_1,\xi_2,\xi_3) \in \R^3$, the principal symbol $\sigma_\xi$ of the system \eqref{equations for short time existence} is given by
    \begin{align*}
        \sigma_\xi \Sigma_{IJ} &= -N\xi_K\Gamma_{\langle IJ \rangle K} + N\xi_{\langle I|} \Gamma_{K|J\rangle K} + N\xi_{\langle I}N_{J\rangle},\\
        \begin{split}
            \sigma_\xi \Gamma_{IJK} &= N\xi_Jk_{KI} - N\xi_Kk_{JI} + \tfrac{2}{3} \delta_{IJ}N\xi_K\Sigma_{LL} - \tfrac{2}{3}\delta_{IK}N\xi_J\Sigma_{LL}\\
            &\quad + \delta_{IK}( N\xi_Lk_{LJ} - N\xi_J\theta ) - \delta_{IJ}( N\xi_Lk_{LK} - N\xi_K\theta ),
        \end{split}\\
        \sigma_\xi N_I &= (\lpar + 1)N\xi_I\theta - \tfrac{1}{3} \alpha N\xi_I\Sigma_{JJ} + \alpha( N\xi_Jk_{JI} - N\xi_I\theta ),\\
        \sigma_\xi \theta &= -2N\xi_K\Gamma_{IIK} + N\xi_IN_I +  2\beta N\xi_K\Gamma_{IIK},\\
        \sigma_\xi \s_t &= N^2\xi_I\s_I,\\
        \sigma_\xi \s_I &= \xi_I\s_t,\\
        \sigma_\xi e_I^i &= \sigma_\xi \omega_i^I = \sigma_\xi N = \sigma_\xi \s = 0.
    \end{align*}
    Let $u = (e_I^i,\omega_i^I,N,\Sigma_{IJ},\Gamma_{IJK},N_I,\theta,\s_t,\s_I,\s)$ and $\widetilde u = (\widetilde e_I^{\,i}, \widetilde \omega_i^I, \widetilde N, \widetilde\Sigma_{IJ},\widetilde\Gamma_{IJK},\widetilde N_I,\widetilde\theta,\widetilde\s_t,\widetilde\s_I, \widetilde\s)$ such that $\Sigma_{IJ} = \Sigma_{JI}$ and $\Gamma_{IJK} = -\Gamma_{IKJ}$, and similarly for $\widetilde u$. Let the matrices $A^I$ be defined by writing \eqref{equations for short time existence} in the form
    \[
    \p_t u = A^Ie_Iu + F.
    \]
    We show that there is a matrix $A^0$ such that $A^0A^I$ is symmetric for all $I$. Note that the principal symbol $\sigma_\xi$ satisfies
    \[
    \langle \widetilde u, A^I\xi_I u \rangle = \widetilde\Sigma_{IJ}\cdot\sigma_\xi\Sigma_{IJ} +\widetilde\Gamma_{IJK}\cdot\sigma_\xi\Gamma_{IJK} + \widetilde N_I\cdot\sigma_\xi N_I + \widetilde\theta\cdot\sigma_\xi\theta + \widetilde\s_t\cdot\sigma_\xi\s_t + \widetilde\s_I\cdot\sigma_\xi\s_I,
    \]
    where $\langle\,\cdot\,,\,\cdot\, \rangle$ denotes the Euclidean inner product. We compute,
    \begin{align*}
    \begin{split}
        \widetilde\Sigma_{IJ} \cdot\sigma_\xi \Sigma_{IJ} &= \tfrac{1}{2} \big( -\widetilde\Sigma_{IJ} N\xi_K\Gamma_{IJK} - \widetilde\Sigma_{IJ}N\xi_K\Gamma_{JIK} + \widetilde\Sigma_{IJ}N\xi_I\Gamma_{KJK}\\
        &\quad + \widetilde\Sigma_{IJ}N\xi_J\Gamma_{KIK} + \widetilde\Sigma_{IJ}N\xi_IN_J + \widetilde\Sigma_{IJ}N\xi_JN_I \big)\\
        &\quad + \tfrac{1}{3}\widetilde\Sigma_{II}N\xi_K\Gamma_{LLK} - \tfrac{1}{3}\widetilde\Sigma_{II}N\xi_L\Gamma_{KLK} - \tfrac{1}{3}\widetilde\Sigma_{II}N\xi_KN_K\\
        &= -\widetilde\Sigma_{IJ}N\xi_K\Gamma_{IJK} + \widetilde\Sigma_{IJ}N\xi_I\Gamma_{KJK} + \widetilde\Sigma_{IJ}N\xi_IN_J\\
        &\quad + \tfrac{2}{3} \widetilde\Sigma_{II}N\xi_K\Gamma_{LLK} - \tfrac{1}{3}\widetilde\Sigma_{II}N\xi_KN_K,
    \end{split}\\
    \begin{split}
        \widetilde\Gamma_{IJK} \cdot\sigma_\xi \Gamma_{IJK} &= \widetilde\Gamma_{IJK} N\xi_J\Sigma_{KI} - \widetilde\Gamma_{IJK} N\xi_K \Sigma_{JI} + \widetilde\Gamma_{KJK} N\xi_L\Sigma_{LJ} - \widetilde\Gamma_{IIK}N\xi_L\Sigma_{LK}\\
        &\quad + \tfrac{2}{3}\widetilde\Gamma_{IIK}N\xi_K\Sigma_{LL} - \tfrac{2}{3}\widetilde\Gamma_{KJK}N\xi_J\Sigma_{LL} - \tfrac{1}{3}\widetilde\Gamma_{KJK}N\xi_J\theta + \tfrac{1}{3}\widetilde\Gamma_{IIK}N\xi_K\theta\\
        &= -2\widetilde\Gamma_{IJK}N\xi_K\Sigma_{IJ} + 2\widetilde\Gamma_{KJK}N\xi_L\Sigma_{LJ} + \tfrac{4}{3} \widetilde\Gamma_{IIK}N\xi_K\Sigma_{LL} + \tfrac{2}{3}\widetilde\Gamma_{IIK}N\xi_K\theta,
    \end{split}\\
    \begin{split}
        \widetilde N_I \cdot\sigma_\xi N_I &= \big( \lpar + 1 - \tfrac{2}{3}\alpha \big)\widetilde N_IN\xi_I\theta - \tfrac{1}{3}\alpha \widetilde N_IN\xi_I\Sigma_{JJ} + \alpha \widetilde N_IN\xi_J\Sigma_{JI}\\
        &= \rho \alpha \widetilde N_IN\xi_I\theta - \tfrac{1}{3}\alpha \widetilde N_IN\xi_I\Sigma_{JJ} + \alpha \widetilde N_IN\xi_J\Sigma_{JI},
    \end{split}\\
    \begin{split}
        \widetilde\theta \cdot\sigma_\xi \theta &= 2(\beta - 1)\widetilde\theta N\xi_K\Gamma_{IIK} + \widetilde\theta N\xi_IN_I\\
        &= \tfrac{1}{3\rho}\widetilde\theta N\xi_K\Gamma_{IIK} + \widetilde\theta N\xi_IN_I.
    \end{split}    
    \end{align*}
    Putting the above information together, it follows that
    \[
    \begin{split}
        &\widetilde\Sigma_{IJ}\cdot\sigma_\xi\Sigma_{IJ} + \tfrac{1}{2}\widetilde\Gamma_{IJK}\cdot\sigma_\xi\Gamma_{IJK} + \tfrac{1}{\alpha}\widetilde N_I\cdot\sigma_\xi N_I + \rho \widetilde\theta\cdot\sigma_\xi\theta + \widetilde\s_t\cdot\sigma_\xi\s_t + N^2 \widetilde\s_I\cdot\sigma_\xi\s_I\\
        &\hspace{1cm} = - N\xi_K( \widetilde\Sigma_{IJ}\Gamma_{IJK} + \widetilde\Gamma_{IJK}\Sigma_{IJ} ) + N\xi_I(\widetilde\Sigma_{IJ}\Gamma_{KJK} + \widetilde\Gamma_{KJK}\Sigma_{IJ})\\
        &\hspace{1cm}\quad + N\xi_I( \widetilde\Sigma_{IJ}N_J + \widetilde N_J \Sigma_{IJ} ) + \tfrac{2}{3}N\xi_K(\widetilde\Sigma_{II}\Gamma_{LLK} + \widetilde\Gamma_{LLK}\Sigma_{II})\\
        &\hspace{1cm}\quad - \tfrac{1}{3}N\xi_K(\widetilde\Sigma_{II}N_K + \widetilde N_K\Sigma_{II}) + \tfrac{1}{3}N\xi_K(\widetilde\Gamma_{IIK}\theta + \widetilde\theta\Gamma_{IIK})\\
        &\hspace{1cm}\quad+ \rho N\xi_I(\widetilde N_I\theta + \widetilde\theta N_I) + N^2\xi_I(\widetilde\s_t\s_I + \widetilde\s_I\s_t)\\
        &\hspace{1cm}= \Sigma_{IJ}\cdot\sigma_\xi\widetilde\Sigma_{IJ} + \tfrac{1}{2}\Gamma_{IJK}\cdot\sigma_\xi\widetilde\Gamma_{IJK} + \tfrac{1}{\alpha} N_I\cdot\sigma_\xi\widetilde N_I + \rho \theta\cdot\sigma_\xi\widetilde\theta + \s_t\cdot\sigma_\xi\widetilde\s_t + N^2 \s_I\cdot\sigma_\xi\widetilde\s_I.
    \end{split}
    \]
    If we define $A^0 := \mathrm{diag}(1,1,1,1, \frac{1}{2}, \frac{1}{\alpha}, \rho, 1, N^2, 1 )$, this calculation shows that
    \[
    \langle \widetilde u, A^0A^I\xi_I u \rangle = \langle A^0A^I\xi_I\widetilde u, u \rangle
    \]
    for all $\xi \in \R^3$. Thus $A^0A^I$ is symmetric for all $I$. But then $A^0A^Ie_I^i$ is symmetric for all $i$ as well. Now we check the regularity conditions. The only potential issue is that $F$ contains expressions of the form $\frac{1}{t}$ and $\frac{1}{N}$. This is why we need the restrictions on the values of $t$ and $N$. Let $U$ be the open set of values for $u$ where $N > a$. Then $A^0, A^I$ and $F$ are smooth in $(a,\infty) \times \Sigma \times U$. Finally, having restricted the possible values of $N$, the condition \eqref{eq: lower A0 bound} holds with $c = \min\{\frac{1}{2},\frac{1}{\alpha},\rho,a^2\}$.  
\end{proof}

\subsection{Constructing a solution to the Einstein--nonlinear scalar field equations} \label{constructing a solution to einsteins equations}

Let $(\Sigma,\bar h, \bar k, \bar \s, \bar \psi)$ be initial data for the Einstein--nonlinear scalar field equations, where $\Sigma$ is a closed orientable 3-dimensional manifold. Fix a reference Riemannian metric $\refmetric$ on $\Sigma$ and let $\{E_i\}$ be a global orthonormal frame for $(\Sigma,\refmetric)$, with dual frame $\{\eta^i\}$ (see Remark~\ref{rmk: paralellizable}). Moreover, let $\{\bar e_I\}$ be an global orthonormal frame for $\bar h$ with dual frame $\{\bar\omega^I\}$. We show how to construct a corresponding development as described in Proposition~\ref{reduced equations}, with $M = (T-\delta,T+\delta) \times \Sigma$ for some $T, \delta > 0$. We do this globally in space first for convenience. The modifications required to deal with the localized setting will be discussed later. The reference frame $\{E_i\}$ can be extended to $(T-\delta,T+\delta) \times \Sigma$ by requiring that $[\p_t,E_i] = 0$ for all $i$. We obtain initial data for \eqref{equations for short time existence} as follows. Let $\bar N \in C^\infty(\Sigma)$ be any positive function. Let $\bar e_I = \bar e_I^iE_i$ and $\bar \omega^I = \bar\omega_i^I\eta^i$. Let $\bar\n$ be the Levi-Civita connection of $\bar h$, and denote by $\bar\theta = \tr_{\bar h}\bar k$ the initial mean curvature. We set
\begin{gather*}
    e_I^i(T) = \bar e_I^i, \quad \omega_i^I(T) = \bar\omega_i^I, \quad \Sigma_{IJ}(T) = \bar k(\bar e_I,\bar e_J) - \tfrac{1}{3} \bar\theta \delta_{IJ}, \quad \Gamma_{IJK}(T) = \bar h(\bar\n_{\bar e_I} \bar e_J,\bar e_K),\\
    N(T) = \bar N, \quad \quad N_I(T) = \bar e_I(\ln \bar N), \quad \theta(T) = \bar\theta, \quad \s_t(T) = \bar N \bar \psi, \quad \s_I(T) = \bar e_I \bar\s, \quad \s(T) = \bar\s. 
\end{gather*}
Then, by standard existence theory for quasilinear symmetric hyperbolic systems (see, for instance, \cite[Chapter~16, Corollary~1.6 and Proposition~2.1; and Exercise~1 in Section~16.1]{taylor_pdeiii_2011}), we get a smooth solution 
\[
(e_I^i,\omega_i^I,N,\Sigma_{IJ},\Gamma_{IJK},N_I,\theta,\s_t,\s_I,\s)
\]
to $\eqref{equations for short time existence}$, defined on $(T-\delta,T+\delta) \times \Sigma$ for some $\delta > 0$. 

\begin{lemma} \label{lemma: solution before constraints}
    Within the setting described above, the following holds. $\Sigma_{IJ}$ is symmetric and $\Gamma_{IJK}$ is antisymmetric in $J$ and $K$. Define $e_I := e_I^iE_i$ and $\omega^I := \omega_i^I\eta^i$. Then $\{e_I(t)\}$ is a frame for $\Sigma$, with dual frame $\{\omega^I(t)\}$, and $N(t) > 0$ for all $t \in (T-\delta,T+\delta)$. In particular, 
    \[
    g := -N^2 dt \otimes dt + \omega^I \otimes \omega^I
    \]
    is a Lorentzian metric on $(T-\delta,T+\delta) \times \Sigma$. If $e_0 := \frac{1}{N}\p_t$, then $\{e_0,e_I\}$ is an orthonormal frame for $g$, and $\{e_I\}$ is Fermi-Walker. $\Sigma_{II} = 0$ and $k_{IJ} = \Sigma_{IJ} + \frac{1}{3}\theta\delta_{IJ}$ are the components of the second fundamental form of the $\Sigma_t = \{t\} \times \Sigma$ hypersurfaces with respect to the frame $\{e_I\}$. In particular, $\theta$ is the corresponding mean curvature. Finally, $\s_I = e_I\s$ and the equation
    \begin{equation} \label{phi equation before constraints}
        e_0e_0\s = e_Ie_I\s - \theta e_0\s - \Gamma_{IIK}e_K\s + e_I(\s)N_I - V'\circ \s
    \end{equation}
    holds.
\end{lemma}

\begin{proof}
    Since, excluding the first term, the right-hand side of \eqref{k equation short time} is symmetric in $I$ and $J$, and the right-hand side of \eqref{gamma equation short time} is antisymmetric in $J$ and $K$, then $\Sigma_{IJ}$ and $\Gamma_{IJK}$ have the corresponding symmetries as long as the initial data have them.  
    
    Turning our attention to $e_I$ and $\omega^I$, note that
    \[
    \p_t( e_I^i\omega_i^J - \delta_I^J) = -Nk_{IK}( e_K^i \omega_i^J - \delta_K^J ) + Nk_{JK}( e_I^i \omega_i^K - \delta_I^K ).
    \]
    Since $e_I^i(T)\omega_i^J(T) = \bar e_I^i \bar \omega_i^J = \delta_I^J$, it follows that $e_I^i \omega_i^J = \delta_I^J$ for all $t \in (T-\delta,T+\delta)$. This means that $\{e_I\}$ is a frame on $\Sigma$ for all $t \in (T-\delta,T+\delta)$, with dual frame $\{\omega^I\}$. Moreover, note that from \eqref{N equation short time}, it follows that if $N(t) > 0$ for some $t \in (T-\delta,T+\delta)$, then $N(t) > 0$ for all $t \in (T-\delta,T+\delta)$. These observations imply that both $g$ and $e_0$ are well defined. Note that by definition, $\{e_0,e_I\}$ is an orthonormal frame for $g$.

    We continue with the statements about the second fundamental form. By taking the trace of \eqref{k equation short time}, we obtain
    \[
    \p_t\Sigma_{II} = - N\theta\Sigma_{II}.
    \]
    Since $\Sigma_{II}(T) = 0$, we conclude that $\Sigma_{II} = 0$ for all $t \in (T-\delta,T+\delta)$. Next, note that
    \[
    [e_0,e_I] = e_0(e_I^i)E_i + e_I^i[e_0,E_i] = -k_{IJ}e_J^iE_i + \frac{1}{N^2}e_I(N)\p_t = -k_{IJ}e_J + e_I(\ln N)e_0.
    \]
    Let $\n$ be the Levi-Civita connection of $g$, then
    \[
    2g(\n_{e_I}e_0,e_J) = \lie_{e_0}g(e_I,e_J) = -g([e_0,e_I],e_J) - g(e_I,[e_0,e_J]) = 2k_{IJ},
    \]
    so that $k_{IJ} = \Sigma_{IJ} + \frac{1}{3}\theta\delta_{IJ}$ are indeed the components of the second fundamental form and $\theta$ is the mean curvature. Moreover,
    \[
    \n_{e_0} e_I = \n_{e_I}e_0 + [e_0,e_I] = e_I(\ln N)e_0,
    \]
    so that the frame $\{e_I\}$ is Fermi-Walker.

    Turning our attention to $\s$, we have
    \[
    \p_t(e_I\s - \s_I) = -Nk_{IJ}(e_J\s - \s_J).
    \]
    Since $e_I\s - \s_I$ is equal to zero initially, we conclude that $\s_I = e_I\s$ for all $t \in (T-\delta,T+\delta)$. Furthermore, \eqref{phi equation before constraints} is a consequence of the fact that
    \begin{equation*}
        e_0e_0\s = \frac{1}{N^2}\p_t^2\s - \frac{1}{N^2} e_0(N)\p_t\s,
    \end{equation*}
    plus \eqref{phi equation short} and \eqref{N equation short time}.    
\end{proof}

Let $(e_I^i,\omega_i^I,N,\Sigma_{IJ},\Gamma_{IJK},N_I,\theta,\s_t,\s_I,\s)$ be as in the lemma above. Then, in view of Proposition~\ref{reduced equations}, in order to show that $(g,\s)$ is the desired solution to the Einstein--nonlinear scalar field equations, it remains to show that $\Gamma_{IJK} = g(\n_{e_I}e_J,e_K)$, that $N_I = e_I(\ln N)$, and that the constraint equations \eqref{constraint equations reduced} are satisfied. Due to the way in which we modified the equations with the constraints to obtain \eqref{equations for short time existence}, all of these facts need to be proved at the same time.

To that end, define
\[
\cn_I := e_I(\ln N) - N_I.
\]
Following the arguments of \cite[Section 4]{fournodavlos21}, define a connection $D$ on $(T-\delta,T+\delta) \times \Sigma$ by
\[
D_{e_0}e_0 := e_I(\ln N)e_I, \quad D_{e_0}e_I := e_I(\ln N)e_0, \quad D_{e_I}e_0 := k_{IJ}e_J, \quad D_{e_I}e_J := k_{IJ}e_0 + \Gamma_{IJK}e_K.
\]
It follows from the definition that $D$ is compatible with $g$. However, we do not know $D$ to be torsion free a priori. Define the torsion tensor $C$ by
\[
C(X,Y) := [X,Y] - D_X Y + D_Y X.
\]
Similarly as in \cite{fournodavlos21}, we need to consider a modified curvature of $D$. Define
\[
\widehat R_{\alpha\beta\mu\nu} := g(D_{e_\alpha}D_{e_\beta}e_\mu - D_{e_\beta}D_{e_\alpha}e_\mu - D_{D_{e_\alpha} e_\beta - D_{e_\beta} e_\alpha} e_\mu,e_\nu),
\]
where the greek indices run from $0$ to $3$. Note that $\widehat R$ is not a tensor in its third index. Also define $\widehat R_{\alpha\beta} := \widehat R_{\gamma\alpha\beta}{}^\gamma$, $\widehat R := \widehat R_\gamma{}^\gamma$. Finally, we introduce
\begin{equation*}
    E_{\alpha\beta} := \widehat R_{\alpha\beta} - e_\alpha(\s)e_\beta(\s) - (V \circ \s)g_{\alpha\beta}.
\end{equation*}
Now we can appeal to Proposition~\ref{evolution equations for constrants} below. It is clear from their definitions that both the torsion $C$ of $D$, and the $\cn_I$ vanish at $t = T$. Moreover, from Lemma~\ref{second bianchi identity for error tensor} below, it is clear that $E_{\alpha\beta}$ also vanishes at $t = T$. Then Proposition~\ref{evolution equations for constrants}, Lemma~\ref{second bianchi identity for error tensor} and uniqueness for symmetric hyperbolic systems imply that the $\cn_I$, $C$ and $E_{\alpha\beta}$ vanish in all of $(T-\delta,T+\delta)\times \Sigma$. From this fact, we conclude the following. First, that $N_I = e_I(\ln N)$. Next, since $C=0$, and $D$ is compatible with $g$, then $D$ is equal to $\n$, the Levi-Civita connection of $g$. In particular, we have that $\Gamma_{IJK} = g(\n_{e_I}e_J,e_K)$. Furthermore, from \eqref{phi equation before constraints}, it now follows that $\Box_g \s = V'\circ\s$. Finally, $E_{\alpha\beta} = 0$ and $D = \n$ imply that $\ric_g = d\s \otimes d\s + (V\circ\s)g$. Hence, $(g,\s)$ is indeed a solution to the Einstein--nonlinear scalar field equations with potential $V$ in $(T-\delta,T+\delta)\times\Sigma$, and the initial data induced on $\Sigma_T$ by $(g,\s)$ coincide with $(\bar h, \bar k, \bar\s, \bar\psi)$.    

\subsection{Propagation of constraints} \label{ssec: propagation of constraints}

We continue to work in the setting described in the previous subsection. Our objective now is to prove Proposition~\ref{evolution equations for constrants}. Introducing the notation $\Gamma_{\alpha\beta\gamma} := g(D_{e_\alpha} e_\beta, e_\gamma)$, we have
\[
\Gamma_{I0J} = -\Gamma_{IJ0} = k_{IJ}, \quad \Gamma_{00I} = -\Gamma_{0I0} = e_I(\ln N), \quad \Gamma_{0IJ} = \Gamma_{I00} = \Gamma_{000} = 0.
\]
Moreover, define
\[
C_{\alpha\beta\gamma} := g( C(e_\alpha,e_\beta),e_\gamma ) = g([e_\alpha,e_\beta] - D_{e_\alpha}e_\beta + D_{e_\beta}e_\alpha, e_\gamma) = -C_{\beta\alpha\gamma}.
\]
Finally, let $\bar D$ be the induced connection on the $\Sigma_t$ hypersurfaces and let $\rt_{\alpha\beta\gamma\delta}$, $\rt_{\alpha\beta}$ and $\rt$ denote the orthonormal frame components of the Riemann, Ricci and scalar curvatures of $D$ respectively. Similarly, denote by $\srt_{IJKL}$, $\srt_{IJ}$ and $\srt$ the components of the curvatures associated with $\bar D$. 

\begin{lemma} \label{curvature identities with torsion}
    The following identities hold:
    \begin{align*}
        C_{0I\alpha} &= C_{IJ0} = 0, \qquad \rt_{\alpha\beta\mu\nu} = -\rt_{\alpha\beta\nu\mu}, \qquad \srt_{IJKL} = -\srt_{IJLK},\\
        \begin{split}
            0 &= \rt_{\alpha\beta\mu\nu} + \rt_{\beta\mu\alpha\nu} + \rt_{\mu\alpha\beta\nu} + D_\alpha C_{\beta\mu\nu} + D_\beta C_{\mu\alpha\nu} + D_\mu C_{\alpha\beta\nu}\\
            &\quad - C_{\alpha\beta}{}^\gamma C_{\gamma\mu\nu} - C_{\beta\mu}{}^\gamma C_{\gamma\alpha\nu} - C_{\mu\alpha}{}^\gamma C_{\gamma\beta\nu}, 
        \end{split}\\
        0 &= D_\mu\rt_{\alpha\beta\gamma\delta} + D_\alpha\rt_{\beta\mu\gamma\delta} + D_\beta\rt_{\mu\alpha\gamma\delta} - C_{\mu\alpha}{}^\nu\rt_{\nu\beta\gamma\delta} - C_{\alpha\beta}{}^\nu\rt_{\nu\mu\gamma\delta} - C_{\beta\mu}{}^\nu\rt_{\nu\alpha\gamma\delta}.
    \end{align*}
    Moreover, the Gauss and Codazzi equations are
    \begin{align*}
        \rt_{IJKL} &= \srt_{IJKL} + k_{IL}k_{JK} - k_{IK}k_{JL},\\
        \rt_{IJ0K} &= \bar D_I k_{JK} - \bar D_J k_{IK} - C_{IJL}k_{LK}.
    \end{align*}
\end{lemma}

\begin{proof}
    See the proof of \cite[Lemma~4.1]{fournodavlos21}.
\end{proof}

From the definition of the modified curvature $\widehat R_{\alpha\beta\mu\nu}$, we see that
\[
\widehat R_{\alpha\beta\mu\nu} = \rt_{\alpha\beta\mu\nu} + C_{\alpha\beta}{}^\gamma \Gamma_{\gamma\mu\nu}.
\]
We define the modified curvatures $\shatr_{IJKL}$, $\shatr_{IJ}$ and $\shatr$ associated with $\bar D$, similarly as we did $\widehat R_{\alpha\beta\mu\nu}$. The following lemma is a direct consequence of Lemma~\ref{curvature identities with torsion}.

\begin{lemma} \label{identities for modified curvature}
    The following identities hold:
    \begin{align}
        \widehat R_{\alpha\beta\mu\nu} &= -\widehat R_{\alpha\beta\nu\mu}, \qquad \shatr_{IJKL} = -\shatr_{IJLK},\\
        0 &= \widehat R_{\alpha\beta\mu\nu} + \widehat R_{\beta\mu\alpha\nu} + \widehat R_{\mu\alpha\beta\nu} + D_\mu C_{\alpha\beta\nu} + D_\alpha C_{\beta\mu\nu} + D_\beta C_{\mu\alpha\nu} + F,\label{first bianchi identity for modified curvature}\\
        \widehat R_{\alpha\beta} - \widehat R_{\beta\alpha} &= -D_\mu C_{\alpha\beta}{}^\mu - D_\alpha C_{\beta\mu}{}^\mu - D_\beta C_{\mu\alpha}{}^\mu + F,\label{antisymmetric part of modified ricci}\\
        0 &= D_\mu \widehat R_{\alpha\beta\gamma\delta} + D_\alpha \widehat R_{\beta\mu\gamma\delta} + D_\beta \widehat R_{\mu\alpha\gamma\delta} + ( \widehat R_{\alpha\beta\mu}{}^\nu + \widehat R_{\beta\mu\alpha}{}^\nu + \widehat R_{\mu\alpha\beta}{}^\nu )\Gamma_{\nu\gamma\delta} + F,\label{second bianchi identity for modified curvature}\\
        \widehat R_{IJKL} &= \shatr_{IJKL} + k_{IL}k_{JK} - k_{IK}k_{JL},\label{gauss equation for modified curvature}\\
        \widehat R_{IJ0K} &= \bar D_I k_{JK} - \bar D_J k_{IK},\label{codazzi equation for modified curvature}\\
        \widehat R_{I0} &= \bar D_J k_{IJ} - e_I \theta,\label{trace of codazzi for modified curvature}\\
        \shatr_{IJ} &= e_K\Gamma_{IJK} - e_I\Gamma_{KJK} + \Gamma_{IJL}\Gamma_{KLK} - \Gamma_{KIL}\Gamma_{LJK}, 
    \end{align}
    where $F$ stands for sums of terms which are either linear or quadratic in the $C_{\alpha\beta\gamma}$, with coefficients depending only on the solution to \eqref{equations for short time existence} and its derivatives. 
\end{lemma}

\begin{proof}
    See \cite[Lemma~4.2]{fournodavlos21}.
\end{proof}

In order to obtain a homogeneous system for the constraint quantities, we need to deal with the cyclic sum inside the parentheses in \eqref{second bianchi identity for modified curvature}. 

\begin{lemma} \label{cyclic sum lemma}
    The cyclic sum $3\widehat R_{[\alpha\beta\mu]\nu} = \widehat R_{\alpha\beta\mu\nu} + \widehat R_{\beta\mu\alpha\nu} + \widehat R_{\mu\alpha\beta\nu}$ satisfies
    \begin{align*}
    \begin{split}
        3\widehat R_{[IJK]L} &= ( \widehat R_{IJ} - \widehat R_{JI} )\delta_{KL} + (\widehat R_{KI} - \widehat R_{IK})\delta_{JL} + (\widehat R_{JK} - \widehat R_{KJ})\delta_{IL}\\
        &\quad + (C_{IJM}\delta_{KL} + C_{KIM}\delta_{JL} + C_{JKM}\delta_{IL})e_M(\ln N),
    \end{split}\\
    3\widehat R_{[0IJ]0} &= C_{IJL}e_L(\ln N),\\
    3\widehat R_{[IJK]0} &= 0,\\
    3\widehat R_{[0IJ]K} &= \delta_{IK}\big(\widehat R_{J0} - e_0(\s)e_J(\s)\big) - \delta_{JK}\big(\widehat R_{I0} - e_0(\s)e_I(\s)\big) + \cn_I k_{JK} - \cn_J k_{IK}.
    \end{align*}
\end{lemma}

\begin{proof}
    For the first identity, note that if any of $I$, $J$ or $K$ coincide, then both sides vanish. So suppose that $I$, $J$ and $K$ are distinct. Then $L$ has to coincide with one of them, since $\Sigma_t$ is 3-dimensional. Assuming $L = K$, we have
    \[
    \begin{split}
        \widehat R_{IJ} - \widehat R_{JI} = \widehat R_{\lambda IJ}{}^\lambda - \widehat R_{\lambda JI}{}^\lambda &= \widehat R_{KIJK} + \widehat R_{JKIK} + \widehat R_{0I0J} - \widehat R_{0J0I}\\
        &= 3\widehat R_{[IJK]L} + \widehat R_{0I0J} - \widehat R_{0J0I}.
    \end{split}
    \]
    Furthermore,
    \begin{equation} \label{eq: k equation propagation of constraints}
    \begin{split}
        \widehat R_{0I0J} &= g( D_{e_0}D_{e_I}e_0 - D_{e_I}D_{e_0}e_0 - D_{D_{e_0} e_I - D_{e_I} e_0} e_0, e_J )\\
        &= g\big( D_{e_0}(k_{IM}e_M) - D_{e_I}( e_M(\ln N)e_M ) - e_I(\ln N)D_{e_0}e_0 + k_{IM}D_{e_M}e_0, e_J \big)\\
        &= e_0k_{IJ} - e_Ie_J(\ln N) - e_M(\ln N)\Gamma_{IMJ} - e_I(\ln N)e_J(\ln N) + k_{IM}k_{MJ}.
    \end{split}
    \end{equation}
    Hence, by symmetry of the $k_{IJ}$,
    \[
    \widehat R_{0I0J} - \widehat R_{0J0I} = -[e_I,e_J](\ln N) + (\Gamma_{IJM} - \Gamma_{JIM})e_M(\ln N) = -C_{IJM}e_M(\ln N).
    \]
    Combining the above observations we obtain the first identity. Furthermore, note that the equality above corresponds to the second identity. The third identity follows from \eqref{codazzi equation for modified curvature}.

    For the final identity, we can use \eqref{trace of codazzi for modified curvature} to write Equation~\eqref{gamma equation short time} as 
    \begin{equation} \label{eq: gamma equation prop of const}
    \begin{split}
        e_0\Gamma_{IJK} &= - k_{IL}\Gamma_{LJK} + \bar D_J k_{KI} - \bar D_K k_{JI} + N_Jk_{IK} - N_Kk_{IJ}\\
        &\quad + \delta_{IK}\big( \widehat R_{J0} - e_0(\s)e_J(\s) \big) - \delta_{IJ}\big( \widehat R_{K0} - e_0(\s)e_K(\s) \big).
    \end{split}
    \end{equation}
    Now we compute,
    \begin{align*}
        3\widehat R_{[0IJ]K} &= \widehat R_{0IJK} + \widehat R_{IJ0K} - \widehat R_{0JIK}\\
        &= g( D_{e_0}D_{e_I}e_J - D_{e_I}D_{e_0}e_J - D_{D_{e_0}e_I - D_{e_I}e_0} e_J, e_K ) + \bar D_I k_{JK} - \bar D_Jk_{IK}\\
        &\quad - g( D_{e_0}D_{e_J}e_I - D_{e_J}D_{e_0}e_I - D_{D_{e_0}e_J - D_{e_J}e_0}e_I,e_K )\\
        & = e_0\Gamma_{IJK} - e_J(\ln N)k_{IK} + k_{IL}\Gamma_{LJK} + \bar D_I k_{JK} - \bar D_J k_{IK}\\
        &\quad - e_0\Gamma_{JIK} + e_I(\ln N)k_{JK} - k_{JL}\Gamma_{LIK}\\
        &= -\bar D_K k_{JI} + \delta_{IK}\big( \widehat R_{J0} - e_0(\s)e_J(\s) \big) - e_J(\ln N)k_{IK} + N_Jk_{IK} - N_Kk_{IJ}\\
        &\quad + \bar D_K k_{IJ} - \delta_{JK}\big( \widehat R_{I0} - e_0(\s)e_I(\s) \big) + e_I(\ln N)k_{JK} - N_I k_{JK} + N_Kk_{JI}\\
        &= \delta_{IK}\big( \widehat R_{J0} - e_0(\s)e_J(\s) \big) - \delta_{JK}\big( \widehat R_{I0} - e_0(\s)e_I(\s) \big) + \cn_Ik_{JK} - \cn_Jk_{IK},
    \end{align*}
    thus finishing the proof.
\end{proof}

Next, we deduce an evolution equation for the torsion.

\begin{lemma}
    Let $\{Y_i\}$ be a local frame on $\Sigma$ with dual frame $\{\zeta^i\}$, and let $\gamma_{ij}^k = \zeta^k\big([Y_i,Y_j]\big)$ be the corresponding structure coefficients. If $e_I = \widetilde e_I^{\,i}Y_i$ and $\omega^I = \widetilde \omega_i^I\zeta^i$, then the $C_{IJK}$ satisfy
    \begin{align}
        C_{IJK} &= e_I( \widetilde e_J^{\,i} ) \widetilde\omega_i^K - e_J( \widetilde e_I^{\,i} )\widetilde\omega_i^K - \widetilde e_I^{\,i} \widetilde e_J^{\,j} \widetilde\omega_k^K \gamma_{ij}^k - \Gamma_{IJK} + \Gamma_{JIK},\label{torsion identity}\\
        \begin{split}
            e_0 C_{IJK} &= k_{KL}C_{IJL} - k_{IL}C_{LJK} - k_{JL}C_{ILK} - k_{JK}\cn_I + k_{IK}\cn_J\\
            &\quad - \delta_{IK}\big( \widehat R_{J0} - e_0(\s)e_J(\s) \big) + \delta_{JK}\big( \widehat R_{I0} - e_0(\s)e_I(\s) \big).
        \end{split}\label{evolution equation for torsion}
    \end{align}
\end{lemma}

\begin{proof}
    The identity \eqref{torsion identity} comes directly from the definition of $C_{IJK}$. For \eqref{evolution equation for torsion}, let $(x^i)$ be local coordinates on $\Sigma$, and set $Y_i = \p_i$, $\zeta^i = dx^i$. Note that for this choice of $Y_i$, the structure coefficients vanish. Moreover, note that the components $\widetilde e_I^{\,i}$ and $\widetilde\omega_i^I$ also satisfy the equations \eqref{frame equations reduced}. We compute,
    \begin{align*}
        e_0C_{IJK} &= e_Ie_0(\widetilde e_J^{\,i}) \widetilde\omega_i^K + e_I(\ln N)e_0( \widetilde e_J^{\,i} )\widetilde\omega_i^K - k_{IL}e_L( \widetilde e_J^{\,i} )\widetilde\omega_i^K + k_{KL}e_I( \widetilde e_J^{\,i} )\widetilde\omega_i^L\\
        &\quad - e_Je_0(\widetilde e_I^{\,i})\widetilde\omega_i^K - e_J(\ln N) e_0(\widetilde e_I^{\,i})\widetilde\omega_i^K + k_{JL} e_L( \widetilde e_I^{\,i} )\widetilde\omega_i^K - k_{KL}e_J( \widetilde e_I^{\,i} ) \widetilde\omega_i^L\\
        &\quad - \Big( -k_{IL}\Gamma_{LJK} + \bar D_Jk_{KI} - \bar D_Kk_{JI} + N_Jk_{IK} - N_Kk_{IJ}\\
        &\quad + \delta_{IK}\big( \widehat R_{J0} - e_0(\s)e_J(\s)\big) - \delta_{IJ}\big( \widehat R_{K0} - e_0(\s)e_K(\s) \big) \Big)\\
        &\quad + \Big( -k_{JL}\Gamma_{LIK} + \bar D_Ik_{KJ} - \bar D_Kk_{IJ} + N_Ik_{JK} - N_Kk_{JI}\\
        &\quad + \delta_{JK}\big( \widehat R_{I0} - e_0(\s)e_I(\s) \big) - \delta_{JI}\big( \widehat R_{K0} - e_0(\s)e_K(\s) \big) \Big)\\
        &= k_{KL}C_{IJL} - k_{IL}C_{LJK} - k_{JL}C_{ILK} - k_{JK}\cn_I + k_{IK}\cn_J\\
        &\quad - \delta_{IK}\big( \widehat R_{J0} - e_0(\s)e_J(\s) \big) + \delta_{JK}\big( \widehat R_{I0} - e_0(\s)e_I(\s) \big),
    \end{align*}
    where we have used \eqref{eq: commutation formula}, \eqref{eq: gamma equation prop of const} and the fact that the matrix with components $\widetilde e_I^{\,i}$ is the inverse of the matrix with components $\widetilde\omega_i^I$.
\end{proof}

Now we prove the following consequence of Lemma~\ref{identities for modified curvature}.

\begin{lemma} \label{second bianchi identity for error tensor}
    Recall that $E_{\alpha\beta}$ is given by
    \begin{equation*}
        E_{\alpha\beta} = \widehat R_{\alpha\beta} - e_\alpha(\s)e_\beta(\s) - (V \circ \s)g_{\alpha\beta}.
    \end{equation*}
    We have
    \begin{align}
        \begin{split}
            E_{00} &= e_I\cn_I - \cn_K\Gamma_{IIK} + \cn_Ie_I(\ln N) + N_I\cn_I\\
            &\quad - \frac{\beta}{1 + 6\rho}\big( \shatr + \theta^2 - k_{KL}k_{KL} - e_0(\s)^2 - e_K(\s) e_K(\s) - 2V \circ \s \big),
        \end{split}\\
        E_{I0} &= \bar D_Jk_{JI} - e_I\theta - e_I(\s)e_0(\s),\\
        \begin{split}
            E_{(IJ)} &= - e_{(I}\cn_{J)} + \cn_K\Gamma_{(IJ)K} - \cn_Ie_J(\ln N) - N_I\cn_J\\
            &\quad + \tfrac{1}{3}(1 + 6\rho)\big( e_L\cn_L - \cn_K\Gamma_{LLK} + \cn_L e_L(\ln N) + N_L\cn_L - E_{00} \big)\delta_{IJ},
        \end{split}\\
        E_{[IJ]} &= -\tfrac{1}{2}( e_KC_{IJK} + e_IC_{JKK} + e_JC_{KIK} ) + L(C).
    \end{align}
    Moreover,
    \begin{equation*}
        D^\alpha E_{\beta\alpha} - \tfrac{1}{2}e_\beta E_\alpha{}^\alpha = L(E,C,\cn),
    \end{equation*}
    where $L(C)$ and $L(E,C,\cn)$ denote terms that are linear and homogeneous in $C_{IJK}$, and in $E_{I0}$, $E_{[IJ]}$, $C_{IJK}$ and $\cn_I$, respectively, with coefficients depending only on the solution to \eqref{equations for short time existence}. Here, $D^\alpha E_{\beta\alpha}$ should be interpreted as if $E_{\alpha\beta}$ were the components of a tensor.
\end{lemma}

\begin{proof} We start by deriving the equation for $E_{(IJ)}$. From \eqref{k equation short time} and \eqref{theta equation short time}, it follows that
    \begin{align}
    \begin{split}
        e_0k_{IJ} &= - \theta k_{IJ} -\shatr_{(IJ)} + e_{(I}N_{J)} + N_IN_J - \Gamma_{(IJ)K}N_K + \s_I\s_J + (V \circ \s)\delta_{IJ}\\
        &\quad + \tfrac{1}{3}\beta\big( \shatr + \theta^2 - k_{KL}k_{KL} - e_0(\s)^2 - \s_K\s_K - 2V\circ\s \big)\delta_{IJ}.
    \end{split} \label{first equation for k}
    \end{align}
    Also, from \eqref{gauss equation for modified curvature},
    \begin{align*}
        \widehat R_{0I0J} = -\widehat R_{0IJ0} = \widehat R_{IJ} - \widehat R_{KIJK} = \widehat R_{IJ} - \shatr_{IJ} - k_{KK} k_{IJ} + k_{KJ}k_{IK}.
    \end{align*}
    Combining this with \eqref{eq: k equation propagation of constraints}, we see that
    \begin{equation} \label{second equation for k}
        e_0k_{IJ} = \widehat R_{(IJ)} - \shatr_{(IJ)} - \theta k_{IJ} + e_{(I}e_{J)}(\ln N) - e_K(\ln N)\Gamma_{(IJ)K} + e_I(\ln N)e_J(\ln N).
    \end{equation}
    Putting \eqref{first equation for k} and \eqref{second equation for k} together yields
    \begin{equation} \label{eq: prelim eq for symm E}
    \begin{split}
        E_{(IJ)} &= - e_{(I}\cn_{J)} + \cn_K\Gamma_{(IJ)K} - \cn_Ie_J(\ln N) - N_I\cn_J\\
        &\quad + \tfrac{1}{3}\beta\big( \shatr + \theta^2 - k_{KL}k_{KL} - e_0(\s)^2 - \s_K\s_K - 2V\circ\s \big)\delta_{IJ}
    \end{split}   
    \end{equation}
    On the other hand, by contracting \eqref{gauss equation for modified curvature},
    \begin{align*}
        \widehat R + 2\widehat R_{00} = \shatr + \theta^2 - k_{IJ}k_{IJ}.
    \end{align*}
    Also, by taking the trace of \eqref{eq: prelim eq for symm E},
    \[
    \begin{split}
        \widehat R + \widehat R_{00} = \widehat R_{II} &= - e_I\cn_I + \cn_K\Gamma_{IIK} - \cn_Ie_I(\ln N) - N_I\cn_I + e_I(\s)e_I(\s) + 3V\circ\s\\
        &\quad + \beta\big( \widehat R + 2\widehat R_{00} - e_0(\s)^2 - \s_K\s_K - 2V\circ\s \big),
    \end{split}
    \]
    implying
    \[
    \tfrac{1}{6\rho}(\widehat R + \widehat R_{00}) = e_I\cn_I - \cn_K\Gamma_{IIK} + \cn_I e_I(\ln N) + N_I\cn_I + \tfrac{1}{6\rho}e_I(\s)e_I(\s) + \tfrac{1}{2\rho}V\circ\s - \beta E_{00}.
    \]
    Finally, by putting these observations together,
    \[
    \begin{split}
    &\tfrac{1}{3}\beta\big( \shatr + \theta^2 - k_{KL}k_{KL} - e_0(\s)^2 - \s_K\s_K - 2V\circ\s \big)\delta_{IJ}\\
    &\hspace{3cm} = \tfrac{1}{3}\beta\big( \widehat R + 2\widehat R_{00} - e_0(\s)^2 - \s_K\s_K - 2V\circ\s \big)\delta_{IJ}\\
    &\hspace{3cm} = \tfrac{1}{3}\beta\big( 6\rho e_L\cn_L - 6\rho\cn_K\Gamma_{LLK} + 6\rho\cn_Le_L(\ln N) + 6\rho N_L\cn_L\\
    &\hspace{3cm}\quad - 6\rho\beta E_{00} + \widehat R_{00} - e_0(\s)^2 + V \circ \s \big)\delta_{IJ}\\
    &\hspace{3cm} = \tfrac{1}{3}(1 + 6\rho)\big( e_L\cn_L - \cn_K\Gamma_{LLK} + \cn_L e_L(\ln N) + N_L\cn_L - E_{00} \big)\delta_{IJ}.
    \end{split}
    \]
    This equality together with \eqref{eq: prelim eq for symm E} yields the desired equality for $E_{(IJ)}$. Note that the equation for $E_{00}$ follows from the equation above. The equalities for $E_{I0}$ and $E_{[IJ]}$ follow directly from \eqref{trace of codazzi for modified curvature}, \eqref{antisymmetric part of modified ricci} and \eqref{torsion identity}.  
    
    Now we move on to the Bianchi type identity for $E_{\alpha\beta}$. By \eqref{phi equation before constraints}, we have
    \begin{align*}
        D^\alpha( d\s \otimes d\s )_{\alpha\beta} &= \big( -D_{e_0}d\s(e_0) + D_{e_I}d\s(e_I) \big)d\s_\beta + d\s^\alpha D_\alpha d\s_\beta\\
        &= \big( -e_0e_0\s + d\s(D_{e_0}e_0) + e_Ie_I\s - d\s(D_{e_I}e_I) \big)d\s_\beta + d\s^\alpha D_\alpha d\s_\beta\\
        &= d\s_I\cn_Id\s_\beta + (V' \circ \s)d\s_\beta + d\s^\alpha D_\alpha d\s_\beta.
    \end{align*}
    By contracting \eqref{second bianchi identity for modified curvature} and by Lemma~\ref{cyclic sum lemma}, it follows that
    \begin{align*}
        D^\alpha \widehat R_{\beta\alpha} = \tfrac{1}{2}e_\beta \widehat R + L,
    \end{align*}
    with $L$ as in the statement of the lemma. Hence
    \begin{align*}
        D^\alpha E_{\beta\alpha} = \tfrac{1}{2} e_\beta \widehat R - d\s_I\cn_Id\s_\beta - 2(V' \circ \s)d\s_\beta - d\s^\alpha D_\alpha d\s_\beta + L.
    \end{align*}
    Moreover,
    \begin{align*}
        e_\beta E_\alpha{}^\alpha = e_\beta( \widehat R - d\s_\alpha d\s^\alpha - 4V\circ\s ) = e_\beta \widehat R - 2d\s^\alpha D_\beta d\s_\alpha - 4(V' \circ \s)d\s_\beta.
    \end{align*}
    Finally, note that
    \begin{align*}
        D_\beta d\s_\alpha = e_\beta e_\alpha\s - \Gamma_{\beta\alpha}{}^\gamma e_\gamma\s = e_\alpha e_\beta\s + [e_\beta,e_\alpha]\s - \Gamma_{\beta\alpha}{}^\gamma e_\gamma\s = D_\alpha d\s_\beta + C_{\beta\alpha}{}^\gamma e_\gamma\s.
    \end{align*}
    Putting the above obsevations together yields the result.
\end{proof}

In order to obtain a symmetric hyperbolic system for the constraint quantities, it turns out to be necessary to consider as variables, not only $\cn_I$, but their first derivatives $e_I\cn_J$, and the trace of their first derivatives $e_I\cn_I$.

\begin{proposition} \label{evolution equations for constrants}
    The variables $(C_{IJK},\cn_I,e_I\cn_J,e_I\cn_I,E_{00},E_{I0},E_{[IJ]})$ satisfy a system of equations of the form:
    \begin{subequations} \label{propagation of constraints system}
    \begin{align}
        e_0 C_{IJK} &= L,\\
        e_0\cn_I &= L,\label{eq: error for lapse derivatives}\\
        e_0e_I\cn_J &= -\alpha e_IE_{J0} + L,\label{eq: error for derivatives of lapse derivatives}\\
        e_0e_I\cn_I &= -\alpha e_IE_{I0} + L,\label{eq: error for trace of lapse derivatives}\\
        e_0E_{I0} &= \big( \tfrac{2}{3} + \rho \big)e_IE_{00} - \tfrac{1}{2} e_Ke_K \cn_I - \tfrac{1}{6}(1+6\rho)e_Ie_K\cn_K + e_KE_{[IK]} + L,\label{evolution equation for momentum constraint}\\
        e_0E_{00} &= \big( \tfrac{1}{3\rho} - \alpha \big)e_KE_{K0} + L,\label{evolution equation for hamiltonian constraint}\\
        e_0E_{[IJ]} &= \tfrac{1}{2}( e_JE_{I0} - e_IE_{J0} ) + L,\label{eq: evolution equation for anisymmetric error}
    \end{align}
    \end{subequations}
    where $L$ denotes sums of terms that are linear and homogeneous in the $E_{00}$, $E_{I0}$, $E_{[IJ]}$, $C_{IJK}$, $\cn_I$ and $e_I\cn_J$, with coefficients depending only on the solution to \eqref{equations for short time existence}, $\lpar$ and $\rho$. In particular, if ${0 < \rho < \frac{2}{9\lpar+6}}$, the variables $(C_{IJK},\cn_I,e_I\cn_J,e_I\cn_I,E_{00},E_{I0},E_{[IJ]})$ satisfy a linear and homogeneous symmetrizable hyperbolic system.
\end{proposition}

\begin{proof}
    The equation for $C_{IJK}$ follows directly from \eqref{evolution equation for torsion}. Now we deduce a relation between $E_{I0}$ and $E_{0I}$. By \eqref{eq: gamma equation prop of const},
    \begin{align*}
        \widehat R_{0I} = -\widehat R_{0KIK} &= -g( D_{e_0}D_{e_K}e_I - D_{e_K}D_{e_0}e_I - D_{D_{e_0}e_K - D_{e_K}e_0}e_I, e_K )\\
        &= -k_{KI}e_K(\ln N) - e_0\Gamma_{KIK} + e_I(\ln N)\theta - k_{KL}\Gamma_{LIK}\\
        &= -k_{IK}\cn_K + \theta \cn_I + \widehat R_{I0} - 2\big(\widehat R_{I0} - e_0(\s)e_I(\s) \big).
    \end{align*}
    Hence,
    \begin{equation} \label{eq: antisymmetry of momentum constraint}
        E_{0I} = -E_{I0} - k_{IK}\cn_K + \theta \cn_I.
    \end{equation}
    We split the rest of the proof into several steps.\vspace{0.2cm}\\  
    \textbf{Equations for $\cn_I$:} As a consequence of \eqref{N equation short time}, \eqref{lapse derivatives equation short time}, \eqref{eq: commutation formula} and the fact that $\Sigma_{II} = 0$, it follows that
    \[
    e_0\cn_I = (\lpar+1)\theta\cn_I - k_{IJ}\cn_J - \alpha E_{I0}.
    \]
    \eqref{eq: error for derivatives of lapse derivatives} and \eqref{eq: error for trace of lapse derivatives} are direct consequences of this equation and \eqref{eq: commutation formula}.
    \vspace{0.2cm}\\ 
    \textbf{Equation for $E_{I0}$:} By Lemma~\ref{second bianchi identity for error tensor} and \eqref{eq: antisymmetry of momentum constraint},
    \[
    e_0E_{I0} = e_KE_{IK} + \tfrac{1}{2} e_IE_{00} - \tfrac{1}{2} e_IE_{KK} + L.
    \]
    Also,
    \[
    \begin{split}
        e_KE_{IK} = e_K E_{[IK]} + e_K E_{(IK)} &= e_K E_{[IK]} - e_K e_{(I} \cn_{K)} - \tfrac{1}{3}(1 + 6\rho)e_IE_{00}\\
        &\quad + \tfrac{1}{3}(1 + 6\rho)e_Ie_K\cn_K + L,
    \end{split}
    \]
    and
    \[
    e_IE_{KK} = -e_Ie_K\cn_K - (1 + 6\rho) e_IE_{00} + (1+6\rho)e_Ie_K\cn_K + L.
    \]
    Putting these observations together yields \eqref{evolution equation for momentum constraint}.
    \vspace{0.2cm}\\
    \textbf{Equation for $E_{00}$:} From Lemma~\ref{second bianchi identity for error tensor} and \eqref{eq: antisymmetry of momentum constraint},
    \[
    e_0E_{00} = e_IE_{0I} + \tfrac{1}{2} e_0E_{00} - \tfrac{1}{2} e_0E_{KK} + L.
    \]
    We have
    \[
    e_0E_{KK} = 6\rho e_0e_K\cn_K - (1+6\rho)e_0E_{00} + L = -6\rho\alpha e_KE_{K0} - (1+6\rho)e_0E_{00} + L,
    \]
    where we have used \eqref{eq: error for lapse derivatives} and \eqref{eq: error for trace of lapse derivatives}. By \eqref{eq: antisymmetry of momentum constraint}, we conclude that \eqref{evolution equation for hamiltonian constraint} holds.
    \vspace{0.2cm}\\
    \textbf{Equation for $E_{[IJ]}$:} By Lemma~\ref{second bianchi identity for error tensor}, \eqref{eq: commutation formula} and \eqref{evolution equation for torsion},
    \begin{align*}
        -2e_0E_{[IJ]} &= e_Ke_0C_{IJK} + e_Ie_0C_{JKK} + e_Je_0C_{KIK}\\
        &\quad - k_{KL}e_LC_{IJK} - k_{IL}e_LC_{JKK} - k_{JL}e_LC_{KIK} + L\\
        &= e_K( k_{KL}C_{IJL} - k_{IL}C_{LJK} - k_{JL}C_{ILK} - \delta_{IK}E_{J0} + \delta_{JK}E_{I0} )\\
        &\quad + e_I( k_{KL}C_{JKL} - k_{JL}C_{LKK} - k_{KL}C_{JLK} -\delta_{JK}E_{K0} + 3E_{J0} )\\
        &\quad + e_J( k_{KL}C_{KIL} - k_{KL}C_{LIK} - k_{IL}C_{KLK} - 3E_{I0} + \delta_{IK}E_{K0} )\\
        &\quad - k_{KL}e_LC_{IJK} - k_{IL}e_LC_{JKK} - k_{JL}e_LC_{KIK} + L\\
        &= e_IE_{J0} - e_JE_{I0} + k_{KL}( e_KC_{IJL} + e_IC_{JKL} + e_JC_{KIL} )\\
        &\quad - k_{IL}( e_KC_{LJK} + e_LC_{JKK} + e_JC_{KLK} )\\
        &\quad - k_{JL}( e_KC_{ILK} + e_IC_{LKK} + e_LC_{KIK} )\\
        &\quad - k_{KL}( e_IC_{JLK} + e_JC_{LIK} + e_LC_{IJK} ) + L.
    \end{align*}
    We obtain \eqref{eq: evolution equation for anisymmetric error} by using \eqref{first bianchi identity for modified curvature} and Lemma~\ref{cyclic sum lemma} to substitute the terms with first derivatives of $C$ for terms that are linear in the constraint variables. 
    \vspace{0.2cm}\\
    \textbf{The system is symmetrizable:} Finally, to prove that the system is indeed hyperbolic, note that the principal symbol is given by
    \begin{align*}
        \sigma_\xi e_I\cn_J &= -\alpha\xi_IE_{J0},\\
        \sigma_\xi e_I\cn_I &= -\alpha\xi_IE_{I0},\\
        \sigma_\xi E_{I0} &= \big( \tfrac{2}{3} + \rho \big) \xi_IE_{00} - \tfrac{1}{2}\xi_Ke_K\cn_I - \tfrac{1}{6}(1+6\rho)\xi_Ie_K\cn_K + \xi_KE_{[IK]},\\
        \sigma_\xi E_{00} &= \big( \tfrac{1}{3\rho} - \alpha \big) \xi_KE_{K0},\\
        \sigma_\xi E_{[IJ]} &= \tfrac{1}{2}(\xi_JE_{I0} - \xi_IE_{J0}),\\
        \sigma_\xi C_{IJK} &= \sigma_\xi \cn_I = 0
    \end{align*}
    for $\xi = (\xi_1,\xi_2,\xi_3) \in \R^3$. Then
    \begin{align*}
        &\frac{1}{2\alpha}\widetilde{e_I\cn_J}\cdot\sigma_\xi e_I\cn_J + \frac{1+6\rho}{6\alpha}\widetilde{e_K\cn_K}\cdot \sigma_\xi e_K\cn_K\\
        & + \widetilde E_{I0}\cdot\sigma_\xi E_{I0} + \frac{\frac{2}{3} + \rho}{\frac{1}{3\rho} - \alpha}\widetilde E_{00} \cdot\sigma_\xi E_{00} + \widetilde E_{[IJ]}\cdot\sigma_\xi E_{[IJ]}\\
        &\hspace{2cm}= -\tfrac{1}{2}\xi_I( \widetilde{e_I\cn_J}E_{J0} + \widetilde E_{J0} e_I\cn_J ) - \tfrac{1}{6}(1 + 6\rho)\xi_I( \widetilde{e_K\cn_K}E_{I0} + \widetilde E_{I0} e_K\cn_K )\\
        &\hspace{2cm}\quad + \big( \tfrac{2}{3} + \rho \big)\xi_I( \widetilde E_{I0}E_{00} + \widetilde E_{00}E_{I0} ) + \xi_K( \widetilde E_{I0}E_{[IK]} + \widetilde E_{[IK]} E_{I0} ).
    \end{align*}
    Hence, similarly as in the proof of Proposition~\ref{prop: the system is hyperbolic}, the system \eqref{propagation of constraints system} is symmetrizable, as long as we can ensure that $\alpha < \frac{1}{3\rho}$. But this is equivalent to $\rho < \frac{2}{9\lpar+6}$.
\end{proof}

\subsection{Short-time existence in the localized setting} \label{short time existence localized}

Now we are ready to prove short-time existence along with a continuation principle for Einstein's equations in a domain of the form $\Omega_{(s,t_0]}$, as in Definition~\ref{def: spacetime domain}. We refer the reader to Subsection~\ref{ssec: conventions appendix} for our conventions regarding norms of functions.  

\begin{theorem} \label{thm: short time existence}
    Let $(\overline\Omega_{t_0},\bar h, \bar k, \bar \s, \bar\psi)$ be initial data for the Einstein--nonlinear scalar field equations with potential $V$, for some $t_0 > 0$. Moreover, let $\{\bar e_I\}$ be an orthonormal frame for $\bar h$, and $\bar N \in C^\infty(\overline\Omega_{t_0})$. Then there is a constant $\mathfrak{b}$ depending only on the $\lpar$ and $\rho$ chosen for \eqref{equations for short time existence}, such that the following holds. Suppose that $t_0^{1-3\delta}|\bar e_I(x)| < \mathfrak{b}$ and $\frac{1}{2} < \bar N(x) < \frac{3}{2}$ for all $x \in \overline\Omega_{t_0}$. Then there is an $s < t_0$ and a unique solution $(e_I^i,\omega^I_i,N,k_{IJ},\Gamma_{IJK},\s)$ to the system of equations of Proposition~\ref{reduced equations} (where we take $E_i = \p_i$) defined on $\Omega_{(s,t_0]}$, smooth on $\Omega_{(s,t_0]}$ and $C^\infty(\overline\Omega_t)$ for each $t \in (s,t_0]$, such that the following holds. If $e_I := e_I^i\p_i$ and $\omega^I := \omega^I_idx^i$, then $e_I(t_0) = \bar e_I$, and $\{e_I\}$ is a frame for $\Omega_t$ with dual frame $\{\omega^I\}$ for all $t \in (s,t_0]$. $N(t_0) = \bar N$ and $N > 0$. If we define
    \[
    g := -N^2dt \otimes dt + \omega^I \otimes \omega^I,
    \]
    then $g$ is a Lorentzian metric; $\{e_I\}$ is a Fermi-Walker frame; $k_{IJ} = k(e_I,e_J)$, where $k$ is the second fundamental form of the $\Omega_t$ hypersurfaces with respect to $g$; and, if $\n$ is the Levi-Civita connection of $g$, then $\Gamma_{IJK} = g(\n_{e_I}e_J,e_K)$. Moreover, $(g,\s)$ is a solution to the Einstein--nonlinear scalar field equations with potential $V$, and the induced initial data on $\overline\Omega_{t_0}$ by $(g,\s)$ coincides with $(\bar h, \bar k, \bar \s, \bar \psi)$. Finally, define
    \[
    \mfe(t) := \|e\|_{H^3(\Omega_t)} + \|\omega\|_{H^3(\Omega_t)} + \|N\|_{H^4(\Omega_t)} + \|k\|_{H^3(\Omega_t)} + \|\Gamma\|_{H^3(\Omega_t)} + \|\s\|_{H^4(\Omega_t)}.
    \]
    Assuming that for all $(t,x) \in \Omega_{(s,t_0]}$, we have that $t^{1-3\delta}|e_I(t,x)|$ is contained in a compact subset of $[0,\mathfrak{b})$ and $N(t,x)$ is contained in a compact subset of $(\frac{1}{2},\frac{3}{2})$, then the following continuation principle holds. Either $\limsup_{t \to s^+}\mfe(t) = \infty$, or there is an $s' < s$ such that $(e_I^i,\omega^I_i,N,k_{IJ},\Gamma_{IJK},\s)$ extends to $\Omega_{(s',t_0]}$ as a solution to the equations of Proposition~\ref{reduced equations} as described above. 
\end{theorem}

\begin{remark}
    The bounds for $\bar N$ are chosen only for convenience. The result would still hold if $a < \bar N < b$. for any positive $a$ and $b$. The only difference is that the constant $\mathfrak{b}$ would also depend on $a$ and $b$.
\end{remark}

\begin{proof}
    To obtain the solution, we solve the system \eqref{equations for short time existence}. First, we obtain initial data 
    \[
    \ci = (\bar e_I^i,\bar\omega_i^I, \bar N, \bar\Sigma_{IJ}, \bar\Gamma_{IJK},\bar N_I, \bar\theta, \bar\s_t,\bar\s_I,\bar\s)
    \]
    for \eqref{equations for short time existence} from $(\bar h, \bar k, \bar \s, \bar\psi)$ in the way described in Subsection~\ref{constructing a solution to einsteins equations}. The idea is to apply the general results of \cite{schochet_euler_1986} to \eqref{equations for short time existence}. Let $B_1(0) \subset \R^3$ denote the open ball with radius $1$ and centered at the origin. Define the map $\Phi : [0,t_0) \times B_1(0) \to \Omega_{(0,t_0]}$ as follows:
    \[
    \Phi(\tau,y) := \big(t_0-\tau,\Phi_\tau(y)\big) := \bigg(t_0-\tau, \Big(\frac{(t_0-\tau)^{3\delta}}{3\delta} + \varepsilon\Big)y\bigg).
    \]
    Then
    \[
    d\Phi(\p_\tau) = -\p_t - (t_0-\tau)^{-1+3\delta}y^i\p_{x^i}, \qquad d\Phi(\p_{y^i}) = \Big( \frac{(t_0-\tau)^{3\delta}}{3\delta} + \varepsilon \Big)\p_{x^i}.
    \]
    If $u = (e_I^i,\omega_i^I,N,\Sigma_{IJ},\Gamma_{IJK},N_I,\theta,\s_t,\s_I,\s)$, by Proposition~\ref{prop: the system is hyperbolic}, the system \eqref{equations for short time existence} can be written in the form
    \[
    A^0\p_tu + A^Ie_Iu + Bu = 0,
    \]
    where the matrices $A^\mu$ are symmetric and depend only on $\lpar$, $\rho$ and $N$ (we write the equations in this form to conform with the conventions of \cite{schochet_euler_1986}). Now we pull-back this equation to $[0,t_0) \times B_1(0)$. Define $\widetilde u$ by
    \[
    \widetilde u := (t^{1-3\delta}e_I^i,\omega_i^I,N,\Sigma_{IJ},\Gamma_{IJK},N_I,\theta,\s_t,\s_I,\s) \circ \Phi.
    \]
    Then $\widetilde u$ satisfies a system of the form
    \begin{equation} \label{pulled back system}
        A^0\p_\tau\widetilde u + \widetilde A^i\p_{y^i}\widetilde u + \widetilde B\widetilde u = 0, 
    \end{equation}
    where 
    \[
    \widetilde A^i = \frac{(t_0-\tau)^{-1+3\delta}}{\frac{(t_0-\tau)^{3\delta}}{3\delta} + \varepsilon}( -A^I\widetilde e_I^{\;i} + A^0y^i ),
    \]
    $\widetilde e_I^{\;i} = (t_0-\tau)^{1-3\delta}e_I^i \circ \Phi$, and we think of $A^\mu$ as functions of $\widetilde N = N \circ \Phi$. If we let $\nu = y^i\p_{y^i}$ denote the outward pointing unit normal of $\p B_1(0)$, then the boundary matrix of the new system is given by
    \[
    A^\nu := \textstyle\sum_i \widetilde A^i\nu^i = \dfrac{(t_0-\tau)^{-1+3\delta}}{\frac{(t_0-\tau)^{3\delta}}{3\delta} + \varepsilon} \big( -\textstyle\sum_i A^I\widetilde e_I^{\;i} y^i + A^0 \big).
    \]
    At this point we restrict the domain of $A^0$, $\widetilde A^i$ and $\widetilde B$ to the open set where $|\widetilde e_I| < \mathfrak{b}$ and $\frac{1}{2} < \widetilde N < \frac{3}{2}$, where $\mathfrak{b}$ is unspecified for now. Take a vector $v$ of the appropriate size. Then we have
    \[
    \big|\big\langle \tsum_i \widetilde e_I^{\;i}y^iA^Iv,v \big\rangle\big| \leq \mathfrak{b}\sum_I\|A^I\| |v|^2 \leq \mathfrak{b}c_1|v|^2,
    \]
    where $\|A^I\|$ denotes the operator norm of $A^I$, and $c_1$ is a constant depending only on $\lpar$ and $\rho$. By Proposition~\ref{prop: the system is hyperbolic} there is a constant $c_2$, depending only on $\lpar$ and $\rho$, such that $A^0 \geq c_2I$. Going back to $A^\nu$, we see that
    \[
    \langle A^\nu v,v \rangle \geq \frac{(t_0-\tau)^{-1+3\delta}}{\frac{(t_0-\tau)^{3\delta}}{3\delta} + \varepsilon}( -\mathfrak{b}c_1 + c_2 )|v|^2.
    \]
    If $\mathfrak{b} < c_2/c_1$, then on every closed interval $I \subset [0,t_0)$, the matrix $A^\nu$ is positive definite with a uniform positive lower bound. We conclude that \cite[Theorems~A2 and A6]{schochet_euler_1986} are applicable to \eqref{pulled back system}, with trivial boundary conditions (that is, with $M(x) = 0$ in the notation of \cite{schochet_euler_1986}). 
    
    Define
    \[
    \widetilde\ci := (t_0^{1-3\delta}\bar e_I^i,\bar\omega_i^I, \bar N, \bar\Sigma_{IJ}, \bar\Gamma_{IJK},\bar N_I, \bar\theta, \bar\s_t,\bar\s_I,\bar\s) \circ \Phi_0.
    \]
    By our assumptions on the initial data, we can then solve \eqref{pulled back system} to obtain a smooth solution $\widetilde u$ of \eqref{pulled back system}, defined on $[0,t_0-s)$ for some $s < t_0$, with $\widetilde u(0) = \widetilde\ci$. But then $\widetilde u$ corresponds to a solution $u$ of \eqref{equations for short time existence}, defined on $\Omega_{(s,t_0]}$, with $u(t_0) = \ci$. Now we need to show that $u$ is indeed a solution to the equations of Proposition~\ref{reduced equations}. Note that the proof of Lemma~\ref{lemma: solution before constraints} applies with no changes in this setting as well. Thus, it only remains to consider the system \eqref{propagation of constraints system}. We can apply the same reasoning as above to \eqref{propagation of constraints system}. Therefore, by taking $\mathfrak{b}$ smaller if necessary, but still depending only on $\lpar$ and $\rho$, the results of \cite{schochet_euler_1986} are applicable to the corresponding pulled-back system. Since \eqref{propagation of constraints system} is linear and homogeneous in the constraint quantities, and the constraint quantities vanish at $t = t_0$, we conclude that they vanish in all of $\Omega_{(s,t_0]}$. Therefore, similarly as in Subsection~\ref{constructing a solution to einsteins equations}, our solution $u$ is in fact a solution to the equations of Propostion~\ref{reduced equations} as claimed.

    Finally, the continuation principle. Given an increasing sequence $\tau_n \to T-s$, we may consider solutions $\widetilde u_n$ to \eqref{pulled back system} with initial data at $\tau_n$ given by $\widetilde u_n(\tau_n) = \widetilde u(\tau_n)$. Now assume that $\limsup_{t \to s^+} \mfe(t) < \infty$. Then there is a constant $C$, depending only on $\varepsilon$, such that $\|\widetilde u(\tau)\|_{H^3(B_1(0))} \leq C$ for all $\tau \in [0,T-s)$. Moreover, by assumption, $|\widetilde e_I(\tau,y)|$ is contained in a compact subset of $[0,\mathfrak{b})$, and $\widetilde N(\tau,y)$ is contained in a compact subset of $(\frac{1}{2},\frac{3}{2})$ for all $(\tau,y) \in [0,T-s) \times B_1(0)$. Hence, the constant $\varepsilon_0$ in \cite[Theorem~A2]{schochet_euler_1986}, may be chosen independently of $\tau_n$. It follows that the solutions $\widetilde u_n$ exist on a time interval of fixed length. Thus, for $n$ large enough, the solution $\widetilde u_n$ exists past $T-s$. By uniqueness, it then follows that $\widetilde u$ extends past $T-s$. But this implies that $u$ extends past $s$ as well.  
\end{proof}

\section{Global existence} \label{sec: global existence}

The main result of this section is Theorem~\ref{thm: global existence} below, which gives us global existence in $\Omega_{(0,T]}$ and some basic bounds on solutions to Einstein's equations, under the assumptions of Theorem~\ref{thm: main theorem}. For the proof, we introduce the quantity
\[
\Theta := N\theta
\]
as an auxiliary variable. $\Theta$ plays a crucial role in the analysis, as it allows us to synchronize the singularity at $t = 0$. As a consequence of \eqref{k equation global} and \eqref{lapse equation global}, $\Theta$ satisfies 
\begin{equation}
    \begin{split}
        \p_t(t\Theta - 1) &= \frac{\lpar}{t}(t\Theta - 1) + \frac{\lpar}{t}(t\Theta - 1)^2 + tN^2\big( -2e_I\gamma_{IJJ} + \tfrac{1}{4}\gamma_{IJK}\gamma_{IJK}\\
        &\quad + \tfrac{1}{2}\gamma_{IJK}\gamma_{IKJ} + \gamma_{IJJ}\gamma_{IKK} + e_Ie_I(\ln N) + e_I(\ln N)e_I(\ln N)\\
        &\quad - \gamma_{IKK}e_I(\ln N) + e_I(\s)e_I(\s) + 3V\circ\s \big).
    \end{split}\label{mean curvature low order equation}
\end{equation}
For the energy estimates, it is also necessary to include the frame derivatives of $\ln N$ and $\s$ as variables. Therefore, in addition to \eqref{frame equations reduced}--\eqref{momentum constraint with structure coeffs} and \eqref{mean curvature low order equation}, we shall also use the following equations:
\begin{align}
    \p_te_I(\ln N) &= (\lpar+1)e_I\Theta - Nk_{IJ}e_J(\ln N),\label{lapse derivatives low order equation}\\
    e_0e_I\s &= e_Ie_0\s + e_I(\ln N)e_0\s - k_{IJ}e_J\s.\label{eq: phi derivatives equation global}
\end{align}
We also need the following alternative versions of \eqref{mean curvature low order equation} and \eqref{lapse derivatives low order equation}:
\begin{align}
    \begin{split}
        \p_t(t\Theta - 1) &= tN^2\big( e_Ie_I(\ln N) + e_I(\ln N)e_I(\ln N) - \gamma_{IKK}e_I(\ln N) - k_{IJ}k_{IJ}\\
        &\quad - e_0(\s)^2 + V\circ\s \big) + \frac{1}{t} + \frac{\lpar+2}{t}(t\Theta-1) + \frac{\lpar+1}{t}(t\Theta-1)^2,
    \end{split}\label{mean curvature high order equation}\\
    \begin{split}
        \p_te_I(\ln N) &= \lpar e_I\Theta + Ne_Jk_{JI} - N\gamma_{JKK}k_{JI} - N\gamma_{JIK}k_{JK}\\
        &\quad - Ne_0(\s)e_I(\s) + \Theta e_I(\ln N) - Nk_{IJ}e_J(\ln N),
    \end{split}\label{lapse derivatives high order equation}
\end{align}
which are obtained by using the constraint equations \eqref{hamiltonian constraint with structure coeffs} and \eqref{momentum constraint with structure coeffs}.

Before getting started, we point out one difference between our argument and \cite{oude_groeniger_formation_2023}. We do not introduce a full \emph{scaffold}. That is, an approximate solution to Einstein's equations constructed from the initial data. Instead, we essentially control the deviation of $k_{IJ}$, $e_0\s$ and $N$ from the initial data, while the remaining variables we control directly. This is possible since for $\Theta$ we can control the asymptotics sharply from the beginning, and the remaining variables are asymptotically negligible in the end. This can be seen as a manifestation of the BKL heuristics, where spatial derivatives are less important than time derivatives.

\subsection{Basic consequences of the assumptions} \label{ssec: bounds on initial data}

We work under the assumptions of Theorem~\ref{thm: main theorem}. There is no loss of generality in assuming that $T < 1$ and that $k_1 \geq 2$, and we do so from now on. Since the eigenvalues of $\bar\mck$ are distinct everywhere in $\overline{\Omega}_T$, then there are smooth vector fields $\bar e_I$ and smooth functions $\bar p_I$ on $\overline{\Omega}_T$, such that $\bar\mck(\bar e_I) = \bar p_I \bar e_I$; see \cite[Theorem~5.3]{serre_matrices_2010} We may normalize the $\bar e_I$ so that $\bar h(\bar e_I,\bar e_I) = 1$. In that case, since $\bar\mck$ is symmetric with respect to $\bar h$, then $\{\bar e_I\}$ is an orthonormal frame for $\bar h$. The corresponding dual frame is denoted by $\{\bar\omega^I\}$. Moreover, we denote by $\bar\gamma_{IJK} = \bar h([\bar e_I,\bar e_J],\bar e_K)$ the corresponding structure coefficients. We impose the condition that $T$ is small enough such that $\frac{1}{3\delta}T^{3\delta} \leq 1$. This is to make sure that the constants arising from Sobolev embedding, interpolation and Moser estimates can be taken to be independent of $T$. We refer the reader to Subsection~\ref{ssec: conventions appendix} for our conventions regarding norms and multiindices. Now we derive some consequences of \eqref{expansion normalized bounds}. 

\begin{remark}
    Below we will use $C$ to denote a general constant whose value may change from line to line. We will use the notation $C = C(a_1,\ldots,a_n)$ to indicate that $C$ depends only on the parameters $a_1,\ldots,a_n$. 
\end{remark}

\begin{lemma} \label{lemma: estimates on initial mean curvature}
    The inequalities
    \[
    \|T\bar\theta-1\|_{C^0(\overline{\Omega}_T)} \leq CT^{4\delta}, \qquad \|\ln\bar\theta\|_{H^{k_1+1}(\Omega_T)} \leq C\langle \ln T \rangle
    \]
    hold for $C = C(\zeta_0,k_1,\varepsilon)$. In particular, if $T$ is small enough, then $\|T\bar\theta-1\|_{C^0(\overline{\Omega}_T)} \leq \frac{1}{2}$.
\end{lemma}

\begin{proof}
    For $x \in \overline\Omega_T$, we have
    \[
    \bar\theta(x) - \bar\theta(0) = \int_0^1 \frac{d}{ds}\bar\theta(sx)\,ds
    = \int_0^1 d\bar\theta_{sx}(x^i\p_i)\,ds.
    \]
    Hence,
    \[
    |T\bar\theta(x) - 1| = T|\bar\theta(x) - \bar\theta(0)| \leq T\|d\bar\theta\|_{C^0(\overline\Omega_T)}|x| \leq C T^{4\delta},
    \]
    where we have used Sobolev embedding and $C = C(\zeta_0,\varepsilon)$. This proves the first inequality. Next, the result for $\ln\bar\theta$ is a consequence of the following observations. If there are no derivatives, just note that $\ln\bar\theta = \ln T\bar\theta - \ln T$. If $|\I| \geq 1$, then
    \[
    \p_{\hspace{1pt}\I}(\ln\bar\theta) = \sum \frac{(-1)^rr!}{(T\bar\theta)^{r+1}} T^{r+1} (\p_{\hspace{1pt}\I_1}\bar\theta) \cdots (\p_{\hspace{1pt}\I_r}\bar\theta) (\p_\J\bar\theta), \qquad \I_1 \cup \cdots \cup \I_r \cup \J = \I,
    \]
    where $|\J| \geq 1$. Then we can estimate the fraction in $C^0$, apply Moser estimates, Lemma~\ref{lemma: moser}, to what remains and use the first inequality.
\end{proof}

\begin{remark}
    In the global setting we need a different proof for the bound for $T\bar\theta-1$. Let $\gamma:[0,1] \to \Sigma$ be a minimizing geodesic with respect to $\refmetric$, such that $\gamma(0) = x$ and $\gamma(1) = y$. Then
    \[
    \bar\theta(y) - \bar\theta(x) = \int_0^1 \frac{d}{ds}\bar\theta\big( \gamma(s) \big) \,ds = \int_0^1 d\bar\theta_{\gamma(s)}\big( \gamma'(s) \big)\,ds,
    \]
    implying
    \[
    |\bar\theta(y) - \bar\theta(x)| \leq \ck[(\Sigma)]{0}{d\bar\theta} d_{\mathrm{ref}}(x,y)
    \]
    for all $x,y \in \Sigma$, where $d_{\mathrm{ref}}$ is the Riemannian distance associated with $\refmetric$. Moreover, from the definition of $T$,
    \[
    T\bar\theta(x) - 1 = \frac{T}{\vol(\Sigma)} \int_\Sigma \bar\theta(x) - \bar\theta(y)\,\muref(y) 
    \]
    for all $x \in \Sigma$. The desired inequality is a consequence of these observations, with $C$ depending additionally on $(\Sigma,\refmetric)$.
\end{remark}

\begin{lemma} \label{lemma: initial C^0 bounds on p_I}
    We have
    \[
    |\bar p_I| < 1 - 6\delta
    \]
    for all $I$.
\end{lemma}

\begin{proof}
    This follows exactly as in \cite[Lemma 77]{oude_groeniger_formation_2023}.
\end{proof}

\begin{lemma} \label{lemma: initial bounds on p_I and frame}
    The inequalities
    \begin{align*}
        \|\bar p_I\|_{H^{k_1+1}(\Omega_T)} + \|\bar\theta^{-\bar p_I}\bar e_I\|_{H^{k_1+1}(\Omega_T)} + \|\bar\theta^{\bar p_I} \bar\omega^I\|_{H^{k_1+1}(\Omega_T)} &\leq C,\\
        \|\bar\theta^{\bar p_K - \bar p_I - \bar p_J} \bar\gamma_{IJK}\|_{H^{k_1}(\Omega_T)} &\leq C\langle \ln T \rangle
    \end{align*}
    hold for $C = C(\zeta_0,k_1,\varepsilon)$.
\end{lemma}

\begin{proof}
    The first inequality follows by the same argument as in the proof of \cite[Proposition~74]{oude_groeniger_formation_2023}. The only difference being that, since we do not assume the initial data to be CMC, $t_0$ in the notation of \cite{oude_groeniger_formation_2023} is substituted by $\bar\theta^{-1}$. Turning our attention to the structure coefficients,
    \[
    \begin{split}
        \bar\gamma_{IJK} &= \bar\omega^K\big( [\bar e_I,\bar e_J] \big) = \bar\omega^K\big( [\bar\theta^{\bar p_I}\bar\theta^{-\bar p_I}\bar e_I,\bar\theta^{\bar p_J}\bar\theta^{-\bar p_J}\bar e_J] \big)\\
        &= \bar\theta^{\bar p_I + \bar p_J - \bar p_K} \Big( \bar\theta^{\bar p_K} \bar\omega^K\big( [\bar\theta^{-\bar p_I}\bar e_I, \bar\theta^{-\bar p_J} \bar e_J] \big) + \bar\theta^{-\bar p_I}\bar e_I(\bar p_J\ln\bar\theta)\delta_J^K - \bar\theta^{-\bar p_J} \bar e_J(\bar p_I \ln\bar\theta)\delta_I^K  \Big).
    \end{split}
    \]
    Hence, the second inequality follows from the first one and Lemma~\ref{lemma: estimates on initial mean curvature}, by using Moser estimates and Sobolev embedding.
\end{proof}

It is convenient to formulate the estimates on the initial data in terms of rescalings using powers of $T$ instead of powers of $\bar\theta$.

\begin{lemma} \label{lemma: initial bounds on frame with powers of T}
    The inequalities
    \begin{align*}
        T\|\bar e_I\|_{H^{k_1+1}(\Omega_T)} + T\|\bar\omega^I\|_{H^{k_1+1}(\Omega_T)} + T\|\bar\gamma_{IJK}\|_{H^{k_1}(\Omega_T)} &\leq CT^{5\delta},\\
        T\|\bar e_I\bar\s\|_{H^{k_1}(\Omega_T)} &\leq CT^{4\delta}
    \end{align*}
    hold for $C = C(\zeta_0,k_1,\delta,\varepsilon)$.
\end{lemma}

\begin{proof}
    By writing $T \bar e_I = T e^{\bar p_I \ln\bar\theta} \bar\theta^{-\bar p_I} \bar e_I$, we have
    \[
    T\p_{\hspace{1pt}\I} \bar e_I^i = T^{1-\bar p_I} e^{\bar p_I \ln T\bar\theta} \sum \p_{\hspace{1pt}\I_1}( \bar p_I \ln\bar\theta ) \cdots \p_{\hspace{1pt}\I_r}( \bar p_I \ln\bar\theta ) \p_\J( \bar\theta^{-\bar p_I} \bar e_I^i ), \quad \I_1 \cup \cdots \cup \I_r \cup \J = \I
    \]
    for $|\I| \geq 1$. Similarly for $\bar\omega^I$ and $\bar\gamma_{IJK}$. The first inequality follows by Lemmas~\ref{lemma: estimates on initial mean curvature}, \ref{lemma: initial C^0 bounds on p_I}, \ref{lemma: initial bounds on p_I and frame}, Moser estimates and Sobolev embedding, since $\langle \ln T \rangle^{k_1+1}T^\delta$ is bounded by a constant depending only on $k_1$ and $\delta$. The second inequality is obtained by noting that
    \[
    \bar e_I\bar\s = \bar e_I^i\big( \p_i\bar\Phi - (\p_i\bar\Psi)\ln\bar\theta - \bar\Psi(\p_i\ln\bar\theta) \big),
    \]
    and using Moser estimates, Sobolev embedding and the first inequality.
\end{proof}

The following result is important for dealing with the terms that involve the potential $V$.

\begin{lemma} \label{lemma: initial estimate on norm der of phi}
    If $T$ is small enough, depending only on $\zeta_0, k_1, \delta$ and $\varepsilon$, then
    \[
    |\bar\Psi| \leq 1 - \frac{1}{6}.
    \]
\end{lemma}

\begin{proof}
    After multiplying by $\bar\theta^{-2}$, the Hamiltonian constraint \eqref{hamiltonian constraint with structure coeffs} for the initial data reads
    \[
    \begin{split}
        \tsum_I \bar p_I^2 + \bar\Psi^2 - 1 + 2\bar\theta^{-2}V \circ\bar\s &= \bar\theta^{-2} \big( 2\bar e_I\bar\gamma_{IJJ} - \tfrac{1}{4} \bar\gamma_{IJK}\bar\gamma_{IJK} - \tfrac{1}{2}\bar\gamma_{IJK}\bar\gamma_{IKJ}\\
        &\quad - \bar\gamma_{IJJ}\bar\gamma_{IKK} - \bar e_I(\bar\s)\bar e_I(\bar\s) \big).
    \end{split}
    \]
    By Lemmas~\ref{lemma: estimates on initial mean curvature} and \ref{lemma: initial bounds on frame with powers of T}, this implies that
    \[
    \big| \tsum_I \bar p_I^2 + \bar\Psi^2 - 1 + 2\bar\theta^{-2}V\circ\bar\s \big| \leq CT^{8\delta},
    \]
    where $C = C(\zeta_0,k_1,\delta,\varepsilon)$ and we have used that $(T\bar\theta)^{-2}$ is bounded. Moreover, note that
    \[
    1 = \tsum_I \bar p_I \leq \sqrt{3}\sqrt{\tsum_I \bar p_I^2}.
    \]
    That is, $\sum_I \bar p_I^2 \geq \frac{1}{3}$. Using these two observations and the nonnegativity of the potential, we obtain
    \[
    \bar\Psi^2 \leq 1 - \tsum_I \bar p_I^2 + CT^{8\delta} \leq 1 - \displaystyle\frac{1}{3} + CT^{8\delta} \leq 1 - \frac{1}{3} + \frac{1}{36} = \Big( 1 - \frac{1}{6} \Big)^2,
    \]
    provided that $T$ is small enough.
\end{proof}

\subsection{Bootstrap assumptions and global existence} \label{ssec: bootstrap assumptions}

For the remainder of this section, we work under the assumption that we have a solution to the Einstein--nonlinear scalar field equations with potential $V$, given by the functions $(e_I^i,\omega_i^I,N,k_{IJ},\gamma_{IJK},\s)$ as described in Proposition~\ref{prop: reduced equations struct coeffs}, defined on $\overline{\Omega}_{[t_b,T]}$, where the fixed reference frame is $E_i = \p_i$, and which corresponds to the initial data $\mathfrak{I}$ in the sense that
\begin{gather*}
    e_I^i(T) = \bar e_I^i, \qquad \omega_i^I(T) = \bar \omega_i^I, \qquad N(T) = 1, \qquad k_{IJ}(T) = \bar\theta \bar p_I\delta_{IJ},\\
    \gamma_{IJK}(T) = \bar\gamma_{IJK}, \qquad \s(T) = \bar\s, \qquad \p_t\s(T) = \bar\psi
\end{gather*}
on $\overline{\Omega}_T$.

In order to control $k_{IJ}$ and $e_0\s$, we introduce quantities
\[
\delta k_{IJ} := k_{IJ} - \frac{\bar p_I}{t}\delta_{IJ}, \qquad \delta\s_0 := e_0\s - \frac{1}{t} \bar\Psi.
\]
These are the appropriate objects to consider, since they can be expected to start and remain small throughout the evolution.

\begin{lemma} \label{lemma: initial bounds for delta objects}
    We have
    \begin{equation*} 
        T\|\delta k_{IJ}(T)\|_{H^{k_1+1}(\Omega_T)} + T\|\delta\s_0(T)\|_{H^{k_1+1}(\Omega_T)} \leq CT^{4\delta},
    \end{equation*}
    where $C = C(\zeta_0,k_1,\varepsilon)$.
\end{lemma}

\begin{proof}
    From the definition, we have
    \[
    T\delta k_{IJ}(T) = (T\bar\theta - 1)\bar p_I\delta_{IJ}, \qquad T\delta\s_0(T) = (T\bar\theta - 1)\bar\Psi.
    \]
    The result now follows as a consequence of Lemmas~\ref{lemma: estimates on initial mean curvature}, \ref{lemma: initial bounds on p_I and frame}, Moser estimates and Sobolev embedding.
\end{proof}

Similarly as in \cite{oude_groeniger_formation_2023}, the idea is to control a high order energy, given in terms of $H^{k_1}$ norms, and a low order energy, given in terms of $C^k$ norms, for $k = k_0$ or $k_0+1$. Since we work in the $3+1$-dimensional setting, we may fix the value of $k_0$ to be $2$; this is the lower bound for $k_0$ that is required by Lemma~\ref{lemma: complicated lemma}. However, we still write our arguments in terms of $k_0$, since in general its value would depend on the dimension. Let $k_0 \geq 2$ and assume that $k_1 \geq k_0 + 2$. The low order energy is defined by
\[
\begin{split}
    \bbl(t) &:= t^{1-4\delta}\|e\|_{C^{k_0+1}(\overline\Omega_t)} + t^{1-4\delta}\|\omega\|_{C^{k_0+1}(\overline\Omega_t)}\\
    &\quad + t^{-3\delta} \|t\Theta - 1\|_{C^{k_0+1}(\overline\Omega_t)} + \|N-1\|_{C^{k_0+1}(\overline\Omega_t)} + t^{1-4\delta}\|\ve(\ln N)\|_{C^{k_0}(\overline\Omega_t)}\\
    &\quad + t\|\delta k\|_{C^{k_0+1}(\overline\Omega_t)} + t^{1-3\delta}\|\gamma\|_{C^{k_0}(\overline\Omega_t)} + t\|\delta\s_0\|_{C^{k_0+1}(\overline\Omega_t)} + t^{1-3\delta}\|\ve\s\|_{C^{k_0}(\overline\Omega_t)}.
\end{split}
\]
Define the energy densities by
\begin{align*}
    \densho{e}&:= \sum_{I,i} \sum_{|\I| \leq k_1} t^{2(1-3\delta)} (\p_{\hspace{1pt}\I} e_I^i)^2, &  \densho{\omega} &:= \sum_{I,i} \sum_{|\I| \leq k_1} t^{2(1-3\delta)} (\p_{\hspace{1pt}\I} \omega_i^I)^2,\\
    \densho{\delta k} &:= \sum_{I,J} \sum_{|\I| \leq k_1} t^2(\p_{\hspace{1pt}\I} \delta k_{IJ})^2, & \densho{\gamma} &:= \sum_{I,J,K} \sum_{|\I| \leq k_1} \tfrac{1}{2}t^2(\p_{\hspace{1pt}\I}\gamma_{IJK})^2,\\
    \densho{N} &:= \sum_{|\I| \leq k_1} \big( \p_{\hspace{1pt}\I}(N-1) \big)^2, & \densho{\ve(\ln N)} &:= \sum_I \sum_{|\I| \leq k_1} t^2\big(\p_{\hspace{1pt}\I} e_I(\ln N)\big)^2,\\
    \densho{\Theta} &:= \sum_{|\I| \leq k_1} \frac{\lpar}{N^2} \big(\p_{\hspace{1pt}\I}(t\Theta - 1)\big)^2,\\
    \densho{\delta\s_0} &:= \sum_{|\I| \leq k_1} t^2(\p_{\hspace{1pt}\I} \delta\s_0)^2, & \densho{\ve\s} &:= \sum_I \sum_{|\I| \leq k_1} t^2( \p_{\hspace{1pt}\I} e_I\s )^2;
\end{align*}
and the total energy density
\[
\rho := \densho{e} + \densho{\omega} + \densho{\delta k} + \densho{\gamma} + \densho{N} + \densho{\ve(\ln N)} + \densho{\Theta} + \densho{\delta\s_0} + \densho{\ve\s}.
\]
Then the high order energy is defined by
\[
\bbh(t)^2 := t^{2A} \int_{\Omega_t} \rho(t) \,dx,
\]
where $A \geq 1$ is a constant. In \cite{oude_groeniger_formation_2023}, the value of $A$ is specified explicitly. In principle, it would be possible to do so here as well. However, we choose not to do so for the sake of simplicity. 

\begin{remark}
    Recall that in the global setting, the reference frame $\{E_i\}$ is extended to $[t_b,T] \times \Sigma$ by requiring that $[\p_t,E_i] = 0$ for all $i$. Then the derivatives $\pI$ that we use in the localized setting are substituted by $E_\I$. 
\end{remark}

Now we make the bootstrap assumption that our solution $(e_I^i,\omega_i^I,N,k_{IJ},\gamma_{IJK},\s)$ satisfies the inequality
\begin{equation} \label{bootstrap inequality}
    \bbl(t) + \bbh(t) \leq r,
\end{equation}
for some $r \in (0, \frac{1}{18} ]$, and for all $t \in [t_b,T]$. As a consequence of our assumptions on the initial data, we can ensure that \eqref{bootstrap inequality} holds in a neighborhood of $T$, provided that $T$ is small enough. The key observation for the proof of global existence is that if \eqref{bootstrap inequality} holds and $r$ and $T$ are small enough, then, in fact, \eqref{bootstrap inequality} holds with a strictly smaller constant on the right-hand side.   

\begin{proposition}[Bootstrap improvement] \label{prop: bootstrap improvement}
    Under the assumptions described above. If $\lpar \geq 3\delta$; $A$ is large enough, depending only on $\delta$ and $\lpar$; $k_1$ is large enough, depending only on $\delta, k_0$ and $A$; and $r$ and $T$ are small enough, depending only on $\zeta_0,k_0,k_1,\delta,\varepsilon$ and $\lpar$; then there is a constant $C = C(\zeta_0,k_0,k_1,\delta,\varepsilon,\lpar)$ such that
    \[
    \bbh(t) + \bbl(t) \leq CT^{\delta/2}
    \]
    for all $t \in [t_b,T]$. In particular, if $T$ is small enough, the bootstrap assumption \eqref{bootstrap inequality} is improved.
\end{proposition}

\begin{proof}
    The result is a direct consequence of Propositions~\ref{prop: high order energy estimate} and \ref{prop: low order energy estimate} below, Lemmas~\ref{lemma: estimates on initial mean curvature}, \ref{lemma: initial bounds on frame with powers of T} and \ref{lemma: initial bounds for delta objects}, and the fact that $N(T)=1$.
\end{proof}

As a consequence of this proposition and the continuation principle of Theorem~\ref{thm: short time existence}, one can prove global existence of solutions.

\begin{theorem}[Global existence] \label{thm: global existence}
    Under the assumptions of Theorem~\ref{thm: main theorem}. If $k_0 = 2$; $\lpar \geq 3\delta$; $A$ and $k_1$ are large enough, depending only on $\delta$ and $\lpar$; and $T$ is small enough, depending only on $\zeta_0,k_1,\delta,\varepsilon$ and $\lpar$; then there is a solution $(\Omega_{(0,T]}, g, \s)$ to the Einstein--nonlinear scalar field equations with potential $V$, as described in Proposition~\ref{prop: reduced equations struct coeffs}, such that the following holds. The initial data induced on $\Omega_T$ by $(g,\s)$ coincide with $(\bar h, \bar k, \bar\s,\bar\psi)$. There is a constant $C = C(\zeta_0,k_1,\delta,\varepsilon,\lpar)$ such that the inequality
    \begin{equation} \label{eq: estimates on the global solution}
        \bbh(t) + \bbl(t) \leq CT^{\delta/2}
    \end{equation}
    holds for all $t \in (0,T]$. In particular, $CT^{\delta/2} < \frac{1}{18}$.
\end{theorem}

\begin{proof}
    Let $A$, $k_1$, $T$ and $r$ be such that Proposition~\ref{prop: bootstrap improvement} is applicable. By Lemma~\ref{lemma: initial bounds on frame with powers of T}, let $T$ be small enough such that $T^{1-3\delta}\|\bar e\|_{C^0(\overline{\Omega}_T)} < \mathfrak{b}$, where $\mathfrak{b}$ is the constant appearing in the statement of Theorem~\ref{thm: short time existence}; and set the initial value for the lapse to be $\bar N = 1$ on $\overline{\Omega}_T$. Then, by Theorem~\ref{thm: short time existence}, we get a solution $(\Omega_{(s,T]},g,\s)$ to Einstein's equations, as described in Proposition~\ref{prop: reduced equations struct coeffs}, that induces the correct initial data on $\Omega_T$, for some $s < T$. Now consider the following set
    \[
    \ca := \{ \tau \in (0,T] : (g,\s) \ \text{extends to} \ \Omega_{[\tau,T]}, \ \text{and} \ \bbh(t) + \bbl(t) \leq r \ \text{for all} \ t \in [\tau,T] \}.
    \]
    By Lemmas~\ref{lemma: estimates on initial mean curvature}, \ref{lemma: initial bounds on frame with powers of T} and \ref{lemma: initial bounds for delta objects}, the fact that $N(T) = 1$, and the continuity of the solution, $\ca$ is nonempty, as long as $k_1 \geq 4$ and $T$ is small enough. It is open by Theorem~\ref{thm: short time existence} and Proposition~\ref{prop: bootstrap improvement}, since $r < \frac{1}{2}$, provided that $rT^\delta < \mathfrak{b}$ and that $T$ is small enough. To prove that it is closed, let $\tau_n \to \tau$ be a decreasing sequence in $\ca$. Then $(g,\s)$ is defined on $\Omega_{(\tau,T]}$ and $\bbh(t) + \bbl(t) \leq r$ for all $t \in (\tau,T]$. Note that we can obtain bounds for the connection coefficients by using \eqref{eq: relation between structure and connection coeffs}. Therefore, the continuation principle of Theorem~\ref{thm: short time existence} implies that $(g,\s)$ extends to $\Omega_{[\tau,T]}$, and $\bbh(\tau) + \bbl(\tau) \leq r$ by continuity. Thus, $\tau \in \ca$. Since $(0,T]$ is connected, we conclude that $\ca = (0,T]$. Finally, \eqref{eq: estimates on the global solution} follows from another application of Proposition~\ref{prop: bootstrap improvement}. 
\end{proof}

\section{Energy estimates} \label{sec: energy estimates}

For this section, we work in the setting described in Subsection~\ref{ssec: bootstrap assumptions}. In particular, we always assume that the solution satisfies the bootstrap assumption \eqref{bootstrap inequality}. In addition, we assume that $k_1$ is large enough such that Lemma~\ref{interpolation lemma} below is applicable. Our goal now is to deduce the estimates required to prove Proposition~\ref{prop: bootstrap improvement}. The main results of this section are Propositions~\ref{prop: high order energy estimate} and \ref{prop: low order energy estimate} below.

\begin{remark}
    We use $C$ and $C_*$ to denote constants with the following dependence:
    \[
    C = C(\zeta_0,k_0,k_1,\delta,\varepsilon,\lpar), \qquad C_* = C_*(\delta,\lpar).
    \]
    Also, when we say that $T$ and $r$ should be small enough, it is understood that the required smallness is only allowed to depend on $\zeta_0,k_0,k_1,\delta,\varepsilon$ and $\lpar$. For the remainder of this section, we use this convention without further comment. 
\end{remark}

\begin{remark}
    In the global setting, we additionally allow $C$ to depend on $(\Sigma,\refmetric)$ and the reference frame $\{E_i\}$.
\end{remark}

\subsection{Basic consequences of the bootstrap assumption}

Since the estimates for $\bbl$ are of ODE type, we need a way to avoid a loss of derivatives. This is achieved by the following lemma, which also gives the largeness condition on $k_1$.

\begin{lemma} \label{interpolation lemma}
    If 
    \[
    k_1 \geq \frac{(A+4\delta)(k_0+4)}{\delta},
    \]
    then
    \begin{align*}
        t^{1-3\delta}\|\gamma\|_{C^{k_0+2}(\overline{\Omega}_t)} + t^{1-3\delta}\|\ve\s\|_{C^{k_0+2}(\overline{\Omega}_t)} + t^{1-4\delta}\|\ve(\ln N)\|_{C^{k_0+2}(\overline{\Omega}_t)} \leq Ct^{-\delta}
    \end{align*}
    for all $t \in [t_b,T]$.
\end{lemma}

\begin{proof}
    We begin with $\gamma_{IJK}$. Sobolev embedding and interpolation, Lemma~\ref{lemma: interpolation}, imply
    \[
    \begin{split}
        t^{1-3\delta}\|\gamma\|_{C^{k_0+2}(\overline{\Omega}_t)} &\leq Ct^{1-3\delta}\|\gamma\|_{H^{k_0+4}(\Omega_t)}\\
    &\leq C t^{1-3\delta}\|\gamma\|_{L^2(\Omega_t)}^{1-\frac{k_0+4}{k_1}} \|\gamma\|_{H^{k_1}(\Omega_t)}^{\frac{k_0+4}{k_1}}\\
    &= Ct^{-(A+3\delta)\frac{k_0+4}{k_1}} \big( t^{1-3\delta}\|\gamma\|_{L^2(\Omega_t)} \big)^{1-\frac{k_0+4}{k_1}} \big( t^{A+1}\|\gamma\|_{H^{k_1}(\Omega_t)} \big)^{\frac{k_0+4}{k_1}}.
    \end{split}
    \]
    The result for the $\gamma_{IJK}$ now follows by our assumption and \eqref{bootstrap inequality}. The argument for $e_I\s$ and $e_I(\ln N)$ is similar.
\end{proof}

\begin{lemma} \label{lemma: potential estimates}
    We have
    \begin{align*}
        t^2\|V\circ\s\|_{C^{k_0+1}(\overline{\Omega}_t)} + t^2\|V'\circ\s\|_{C^{k_0+1}(\overline{\Omega}_t)} &\leq Ct^{5\delta},\\
        t^{A+2}\|V\circ\s\|_{H^{k_1}(\Omega_t)} + t^{A+2}\|V'\circ\s\|_{H^{k_1}(\Omega_t)} &\leq Ct^{5\delta}
    \end{align*}
    for all $t \in [t_b,T]$.
\end{lemma}

\begin{proof}
    We begin by estimating $\s$. Define $\check{\s}$ by
    \[
    \check{\s} := \bar\Psi\ln\Big(\frac{t}{T}\Big) - \bar\Psi \ln\bar\theta + \bar\Phi.
    \]
    By Sobolev embedding, Moser estimates and Lemma~\ref{lemma: estimates on initial mean curvature}, we have
    \[
    \|\check{\s}\|_{C^{k_0+1}(\overline{\Omega}_t)} + \|\check{\s}\|_{H^{k_1}(\Omega_t)} \leq C\langle \ln t \rangle.
    \]
    Also, Lemma~\ref{lemma: initial estimate on norm der of phi} implies that
    \[
    |\check{\s}(t)| \leq -|\bar\Psi|\ln t + C \leq - \Big(1 - \frac{1}{6}\Big)\ln t + C.
    \]
    Furthermore,
    \[
    \s(t) - \check{\s}(t) = -\int_t^T \p_t(\s - \check{\s})(s)\,ds = -\int_t^T \frac{1}{s} \big( Ns\delta\s_0 + (N-1)\bar\Psi \big)(s)\,ds.
    \]
    Hence,
    \[
    \|\s-\check{\s}\|_{H^{k_1}(\Omega_t)} \leq \int_t^T \frac{1}{s}\|Ns\delta\s_0 + (N-1)\bar\Psi\|_{H^{k_1}(\Omega_s)} \,ds \leq C t^{-A},
    \]
    where we have used Moser estimates and the fact that $A \geq 1$. Similarly
    \[
    \|\s-\check{\s}\|_{C^{k_0+1}(\overline{\Omega}_t)} \leq C\langle \ln t \rangle.
    \]
    Moreover, we also have
    \[
    |\s(t) - \check{\s}(t)| \leq (r^2 + 2r) \int_t^T \frac{1}{s}\,ds \leq 3r\int_t^T\frac{1}{s}\,ds \leq -\frac{1}{6}\ln t,
    \]
    as $r \leq \frac{1}{18}$. Putting these observations together yields
    \[
    |\s(t)| \leq -\ln t + C, \qquad \|\s\|_{C^{k_0+1}(\overline{\Omega}_t)} + t^A\|\s\|_{H^{k_1}(\Omega_t)} \leq C\langle \ln t \rangle.
    \]
    Turning our attention to the potential, we see that for $k \leq k_1$,
    \[
    |V^{(k)} \circ \s| \leq Ce^{2(1-3\delta)|\s|} \leq Ct^{-2+6\delta}.
    \]
    Now we write
    \[
    \p_{\hspace{1pt}\I} V\circ\s = \sum (V^{(k)}\circ\s) (\p_{\hspace{1pt}\I_1}\s) \cdots (\p_{\hspace{1pt}\I_k}\s), \qquad \I_1 \cup \cdots \cup \I_k = \I,
    \]
    from which the result follows for $V\circ\s$ by Moser estimates. Similarly for $V'\circ\s$.
\end{proof}

\subsection{Setup for the low order energy estimate}

In order to be able to estimate the low order energy, it is necessary to commute spatial derivatives with the $\bar p_I$. For that purpose, the following lemma is crucial. We use the same ideas as in \cite[Subsection~5.1]{oude_groeniger_formation_2023}.

\begin{lemma} \label{lemma: commutator estimate with crazy sequence}
    Let $\mfp$ denote any of $\bar p_I$, $-\bar p_I$, or $\bar p_I + \bar p_J - \bar p_K$, for $I \neq J$. Then there are constants $\sigma_m \geq 1$, for $m = 0,\ldots,k_0+1$, such that the inequality
    \[
    \sum_{m=1}^\ell 2\sigma_m \Bigg| \sum_{|\I| = m} [\p_{\hspace{1pt}\I}, \mfp](u) \p_{\hspace{1pt}\I} u \Bigg| \leq 2\delta \sum_{m=0}^\ell \sigma_m \sum_{|\I|=m} (\p_{\hspace{1pt}\I} u)^2
    \]
    holds for all $u \in C^\infty(\overline{\Omega}_t)$ and $\ell = k_0, k_0+1$. Moreover, the $\sigma_m$ are bounded above by a constant depending only on $\zeta_0$, $k_0$, $k_1$, $\delta$ and $\varepsilon$.
\end{lemma}

\begin{proof}
    Let
    \[
    B_{k_0} := 4^{k_0+1} \max_{I \neq J} \|\bar p_I + \bar p_J - \bar p_K\|_{C^{k_0+1}(\overline{\Omega}_T)}^2.
    \]
    Now define $\sigma_{k_0+1} := 1$ and 
    \[
    \sigma_m := \Big( \frac{B_{k_0}}{\delta^2} + 1 \Big)^{k_0+1-m}, \qquad 0 \leq m \leq k_0.
    \]
    Then, the inequality
    \begin{equation} \label{crazy sequence inequality}
        \sigma_m \geq \frac{B_{k_0}}{\delta^2} \sum_{i = m+1}^{k_0+1} \sigma_i
    \end{equation}
    holds for all $m \leq k_0$. Indeed, proceeding by induction, note that
    \[
    \sigma_{k_0} = \frac{B_{k_0}}{\delta^2} + 1 \geq \frac{B_{k_0}}{\delta^2} = \frac{B_{k_0}}{\delta^2} \sigma_{k_0+1}.
    \]
    Now suppose \eqref{crazy sequence inequality} holds for $m$. Then
    \[
    \begin{split}
        \sigma_{m-1} = \Big( \frac{B_{k_0}}{\delta^2} + 1 \Big)^{k_0+2-m} = \Big( \frac{B_{k_0}}{\delta^2} + 1 \Big) \sigma_m \geq \frac{B_{k_0}}{\delta^2} \sigma_m + \frac{B_{k_0}}{\delta^2} \sum_{i = m+1}^{k_0+1} \sigma_i = \frac{B_{k_0}}{\delta^2} \sum_{i=m}^{k_0+1} \sigma_i.
    \end{split}
    \]
    That is, \eqref{crazy sequence inequality} holds for $m-1$. Thus, \eqref{crazy sequence inequality} holds.

    Now we move on to the commutator estimate. By Young's inequality and Lemma~\ref{commutator estimate}, we have
    \[
    2\sigma_m \Bigg| \sum_{|\I| = m} [\p_{\hspace{1pt}\I}, \mfp](u) \p_{\hspace{1pt}\I} u \Bigg| \leq \sigma_m \frac{B_{k_0}}{\delta} \sum_{|\J| \leq m-1} (\p_{\J}u)^2 + \delta\sigma_m \sum_{|\I|=m} (\p_{\hspace{1pt}\I}u)^2,
    \]
    for $|\I| \leq k_0+1$. Therefore,
    \[
    \begin{split}
        \sum_{m=1}^\ell 2\sigma_m \Bigg| \sum_{|\I| = m} [\p_{\hspace{1pt}\I}, \mfp](u) \p_{\hspace{1pt}\I} u \Bigg| &\leq  \frac{B_{k_0}}{\delta} \sum_{m=1}^\ell \sigma_m \sum_{|\J| \leq m-1} (\p_\J u)^2 + \delta \sum_{m=1}^\ell \sigma_m \sum_{|\I|=m} (\p_{\hspace{1pt}\I} u)^2\\
        &= \delta \frac{B_{k_0}}{\delta^2} \sum_{m=0}^{\ell-1} \bigg( \sum_{i=m+1}^\ell \sigma_i \bigg) \sum_{|\J|=m} (\p_\J u)^2 + \delta \sum_{m=1}^\ell \sigma_m \sum_{|\I|=m} (\p_{\hspace{1pt}\I} u)^2\\
        &\leq \delta \sum_{m=0}^{\ell-1} \sigma_m \sum_{|\J|=m} (\p_\J u)^2 + \delta \sum_{m=1}^\ell \sigma_m \sum_{|\I|=m} (\p_{\hspace{1pt}\I} u)^2\\
        &\leq 2\delta \sum_{m=0}^\ell \sigma_m \sum_{|\I|=m} (\p_{\hspace{1pt}\I} u)^2,
    \end{split}
    \]
    as desired.
\end{proof}

Now we introduce the objects that will be used to perform the low order energy estimates:
\begin{align*}
    \lowen{e} &:= \sum_{I,i} \sum_{m=0}^{k_0+1} \sigma_m \sum_{|\I|=m} t^{2(1-4\delta)} (\p_{\hspace{1pt}\I} e_I^i)^2, \quad \lowen{\omega} := \sum_{I,i} \sum_{m=0}^{k_0+1} \sigma_m \sum_{|\I|=m} t^{2(1-4\delta)} (\p_{\hspace{1pt}\I} \omega_i^I)^2.\\
    \lowen{\delta k} &:= \sum_{I,J} \sum_{|\I| \leq k_0+1} t^2 (\p_{\hspace{1pt}\I} \delta k_{IJ})^2, \quad \lowen{\Theta} := \sum_{|\I| \leq k_0+1} t^{-6\delta} \big( \p_{\hspace{1pt}\I}(t\Theta-1) \big)^2,\\
    \lowen{N} &:= \sum_{|\I| \leq k_0+1} \big( \p_{\hspace{1pt}\I} (N-1) \big)^2, \quad \lowen{\gamma} := \sum_{I,J,K} \sum_{m=0}^{k_0} \sigma_m \sum_{|\I|=m} t^{2(1-3\delta)}(\p_{\hspace{1pt}\I}\gamma_{IJK})^2,\\
    \lowen{\ve(\ln N)} &:= \sum_{I} \sum_{m=0}^{k_0} \sigma_m \sum_{|\I|=m} t^{2(1-4\delta)} \big(\p_{\hspace{1pt}\I} e_I(\ln N)\big)^2,\\
    \lowen{\delta\s_0} &:= \sum_{|\I| \leq k_0+1} t^2( \p_{\hspace{1pt}\I} \delta\s_0 )^2, \quad \lowen{\ve\s} := \sum_I \sum_{m=0}^{k_0} \sigma_m \sum_{|\I|=m} t^{2(1-3\delta)} (\p_{\hspace{1pt}\I} e_I\s)^2,
\end{align*}
where the $\sigma_m$ are the constants whose existence is asserted in Lemma~\ref{lemma: commutator estimate with crazy sequence}.

\subsection{Estimates for subcritical and critical terms}

Here we explain how to estimate the terms that we refer to as subcritical and critical in the equations. But first, a comment about notation. 

\begin{remark} \label{rmk: schematic notation}
    In order to simplify the proofs, we will routinely use schematic notation to denote some of the terms appearing in the equations. The notation works as follows: we write $\ct_1 * \cdots * \ct_n$ to denote a product of factors where the exact constant in front of the term, and the exact contractions and values of the indices are not important, and therefore are omitted. Following this convention, we write $\p$ for an arbitrary spatial partial derivative, and often write $e$ instead of $e_I^i$, $\ve\s$ instead of $e_I\s$, etc.
\end{remark}

Let $\ct$ be a term consisting of factors that are any of the objects listed in Table~\ref{table 1} (possibly with $\p$ applied to them), some power of $t$ or a constant.

\begin{table}[ht]
\centering
\captionsetup{width=13cm}
\begin{tabular}{ | c | c | }
 \hline
 Object & Contribution\\ 
 \hline\hline
 $t\Theta - 1$ & $t^{3\delta}$\\
 \hline
 $N, N-1, \bar p, \bar\Psi$ & $1$\\
 \hline
 $e, \omega, \ve(\ln N)$ & $t^{-1+4\delta}$\\
 \hline
 $\gamma, \ve\s$ & $t^{-1+3\delta}$\\
 \hline
 $\delta k, \delta\s_0, k, e_0\s$ & $t^{-1}$\\
 \hline
 $V\circ\s, V'\circ\s$ & $t^{-2+5\delta}$\\
 \hline\hline
 Application of Lemma~\ref{interpolation lemma} & $t^{-\delta}$\\
 \hline
\end{tabular}
\caption{This table summarizes the time dependent contribution of the bounds that we have for the $C^k$ norms of the listed objects, for $k = k_0$ or $k = k_0+1$ depending on the object.}
\label{table 1}
\end{table}

\paragraph{The $C^k$ estimates:} The relevant information for the $C^k$ estimates is summarized in Table~\ref{table 1}. The table lists the time dependent contribution of each type of factor that may appear. Note that the bootstrap assumption gives us control of $e, \omega, t\Theta-1, N-1, \delta k$ and $\delta\s_0$ in $C^{k_0+1}$, but gives control of $\ve(\ln N), \gamma$ and $\ve\s$ only in $C^{k_0}$. Sometimes we need control of the latter objects in $C^{k_0+2}$. This is achieved through an application of Lemma~\ref{interpolation lemma}. We say that $\ct$ is \emph{subcritical} if it satisfies an estimate of the form
\begin{equation}
    t^\alpha\|\ct\|_{C^k(\overline{\Omega}_t)} \leq Ct^{-1+\delta},
\end{equation}
for an appropriate value of $\alpha$, and $k = k_0$ or $k_0+1$. What value should $\alpha$ have depends on the context, so its exact value will be specified on a case by case basis; and similarly for $k$. Moreover, we say that $\ct$ is \emph{critical} if it satisfies an estimate of the form
\begin{equation}
    \sum_{|\I| \leq k} t^\alpha |\p_{\hspace{1pt}\I}\ct| \leq \frac{Cr}{t} \lowen{\xi}^{1/2},
\end{equation}
with $\alpha$ and $k$ as above, where $\xi$ corresponds to one of the factors in $\ct$. What $\xi$ should be will also be indicated in each case. In the critical terms, there will always be one factor among $N-1$, $\delta k$ and $\delta\s_0$, which introduces $r$ to the estimate. These estimates are obtained by repeated applications of Lemma~\ref{commutator estimate}. 

\paragraph{The $H^{k_1}$ estimates:} In this case, in addition to Table~\ref{table 1}, we also need the information in Table~\ref{table 2}. Similarly to the $C^k$ case, the term $\ct$ is \emph{subcritical} if it satisfies an estimate of the form
\begin{equation}
    t^{A+\alpha} \|\ct\|_{H^{k_1}(\Omega_t)} \leq Ct^{-1+\delta},
\end{equation}
for an appropriate value of $\alpha$, to be specified on a case by case basis. Moreover, $\ct$ is called \emph{critical} if it satisfies an estimate of the form
\begin{equation}
    t^{A+\alpha}\|\ct\|_{H^{k_1}(\Omega_t)} \leq \frac{Cr}{t} \bbh(t) + Ct^{-1+\delta},
\end{equation}
with $\alpha$ as above. As in the $C^k$ case, the critical terms always present a factor of $N-1$, $\delta k$ or $\delta\s_0$, which introduces $r$ to the estimate. The idea for these estimates is to first apply the Moser estimates of Lemma~\ref{lemma: moser} to $\ct$, then add the contribution of each factor, according to Tables~\ref{table 1} and \ref{table 2}. 

\begin{table}[ht]
\centering
\captionsetup{width=13cm}
\begin{tabular}{ | c | c | }
 \hline
 Object & Contribution in $H^{k_1}$ \\ 
 \hline\hline
 $\bar p, \bar\Psi$ & $1$\\
 \hline
 $t\Theta - 1, N, N-1$ & $t^{-A}$\\
 \hline
 $e, \omega$ & $t^{-A-1+3\delta}$\\
 \hline
 $\delta k, \delta\s_0, k, e_0\s, \gamma, \ve\s, \ve(\ln N)$ & $t^{-A-1}$\\
 \hline
 $V\circ\s, V'\circ\s$ & $t^{-A-2+5\delta}$ \\
 \hline
\end{tabular}
\caption{This table summarizes the time dependent contribution of the bounds that we have for the $H^{k_1}$ norms of the listed objects.}
\label{table 2}
\end{table}

\begin{remark}
    In order to avoid the notation getting too cumbersome, for the remainder of this section, we omit mention of the domain in the $C^k$ and $H^k$ norms, whenever appropriate. So, unless otherwise stated, these norms are understood to correspond to $\overline{\Omega}_t$ or $\Omega_t$ respectively.
\end{remark}

Now we give some examples to illustrate how the estimates go. Starting with the $C^k$ estimates, we first consider the term $N * e * \p\gamma$. This is subcritical with $\alpha = 1$ and $k = k_0+1$. This type of term arises from the low order estimates for the second fundamental form. Note that for this term, we need control of $\gamma$ in $C^{k_0+2}$. Thus, one application of Lemma~\ref{interpolation lemma} is required. By Lemma~\ref{commutator estimate}, and taking into account Table~\ref{table 1}, we see that 
\[
t\|N * e * \p\gamma\|_{C^{k_0+1}} \leq Ct \cdot t^{-1+4\delta} \cdot t^{-1+3\delta} \cdot t^{-\delta} \leq Ct^{-1+\delta}.
\]
Now consider $(N-1) * k * e$. This term is critical with $\alpha = 1-4\delta$, $k = k_0+1$ and $\xi = e$. It arises from the low order estimates for the frame. By Lemma~\ref{commutator estimate} and the Cauchy-Schwarz inequality,
\[
\sum_{|\I| \leq k_0+1} t^{1-4\delta}\big|\p_{\hspace{1pt}\I}\big((N-1) * k * e\big)\big| \leq C\|N-1\|_{C^{k_0+1}} \|k\|_{C^{k_0+1}} \lowen{e}^{1/2} \leq \frac{Cr}{t} \lowen{e}^{1/2}.
\]

Moving on to the $H^{k_1}$ estimates, consider $N * \ve(\ln N) * \gamma$. This term is subcritical with $\alpha = 1$. It arises in the high order estimates for the second fundamental form. By Moser estimates and the information in Tables~\ref{table 1} and \ref{table 2}, we see that
\[
\begin{split}
    t^{A+1}\|N * \ve(\ln N) * \gamma\|_{H^{k_1}} &\leq Ct^{A+1} \big( \|N\|_{C^0} \|\ve(\ln N)\|_{C^0} \|\gamma\|_{H^{k_1}}\\
    &\quad + \|N\|_{C^0} \|\ve(\ln N)\|_{H^{k_1}} \|\gamma\|_{C^0}\\
    &\quad + \|N\|_{H^{k_1}} \|\ve(\ln N)\|_{C^0} \|\gamma\|_{C^0} \big)\\
    &\leq Ct^{A+1}( t^{-1+4\delta} \cdot t^{-A-1} + t^{-A-1} \cdot t^{-1+3\delta} + t^{-A} \cdot t^{-1+4\delta} \cdot t^{-1+3\delta} )\\
    &\leq Ct^{-1+\delta}.
\end{split}
\]
Next, we consider the term $\delta k * e$. This is critical with $\alpha = 1-3\delta$. It arises in the high order estimates for the frame. Again, we use Moser estimates and the information in Tables~\ref{table 1} and \ref{table 2} to obtain
\[
\begin{split}
    t^{A+1-3\delta}\|\delta k * e\|_{H^{k_1}} &\leq Ct^{A+1-3\delta}\big( \|\delta k\|_{C^0} \|e\|_{H^{k_1}} + \|\delta k\|_{H^{k_1}} \|e\|_{C^0} \big)\\
    &\leq Ct^{A+1-3\delta}( rt^{-1} \cdot t^{-A-1+3\delta} \bbh(t) + t^{-A-1} \cdot t^{-1+4\delta} )\\
    &\leq \frac{Cr}{t}\bbh(t) + Ct^{-1+\delta}.
\end{split}
\]
Note that after the application of Moser estimates, the total contribution of the first term adds up to $t^{-1}$. For this reason, we need to include the $r$ coming from estimating $\|\delta k\|_{C^0}$, and we need to estimate $\|e\|_{H^{k_1}}$ in terms of the high order energy $\bbh(t)$. However, the total contribution of the second term adds up to $t^{-1+\delta}$, so we can treat this term in the same way as the subcritical terms. This is what happens in general when dealing with the critical terms. We finish this subsection with some observations that will be useful in what follows.

\begin{remark}
    In the following proofs, we will use the notation $\subc$ to denote sums of subcritical terms, and $\crit$ to denote sums of critical terms, either in the $C^k$ or the $H^{k_1}$ estimates.
\end{remark}


\begin{lemma} \label{estimates for powers of theta}
  Let $\alpha \in \R$, then
  \begin{equation}\label{eq:basic theta est}
    \big|\p_{\hspace{1pt}\I}\big( \theta^\alpha\big )\big| \leq C\langle \ln t \rangle^{|\I|} t^{-\alpha}, \qquad
    \|\ln\theta\|_{C^{k_0+1}} \leq C\langle \ln t \rangle
  \end{equation}
      for $|\I| \leq k_0 + 1$ and for all $t\in [t_b,T]$.
\end{lemma}

\begin{proof}
  We begin with the case without derivatives. Write
  \[
  \theta^\alpha = N^{-\alpha} (t\Theta)^\alpha t^{-\alpha}.
  \]
  It is clear from the bootstrap assumption that $\frac{1}{N} \leq 2$ and $\frac{1}{t\Theta}\leq 2$. Due to this observation and the assumptions,
  it is clear that $N$, $\frac{1}{N}$, $t\Theta$ and $\frac{1}{t\Theta}$ are bounded. Combining these observations, it is clear that
  \[
  |\theta^\alpha| \leq Ct^{-\alpha}.
  \]
  Next,
  \[
  |\ln \theta| = |\ln(t\Theta) - \ln N - \ln t|\leq C\langle \ln t \rangle.
  \]
  Turning to the derivatives, if $1 \leq |\I| \leq k_0 + 1$,
  \[
  \p_{\hspace{1pt}\I} \ln\theta = \p_{\hspace{1pt}\I}(\ln(t\Theta) - \ln N-\ln t)= \p_{\hspace{1pt}\I}(\ln(t\Theta) - \ln N).
  \]
  It follows that $|\p_{\hspace{1pt}\I} \ln \theta| \leq C$. Hence,
  \[
  \|\ln\theta\|_{C^{k_0+1}} \leq C\langle \ln t \rangle.
  \]
  Finally,
  \[
  \p_{\hspace{1pt}\I} (\theta^\alpha)
  = \theta^\alpha \sum \p_{\hspace{1pt}\I_1} ( \alpha\ln\theta )
  \cdots \p_{\hspace{1pt}\I_k}( \alpha\ln\theta ), \qquad \I_1 \cup \cdots \cup \I_k = \I,
  \]
  from which the result follows.
\end{proof} 

\begin{lemma} \label{estimate on the kasner relations}
    We have 
    \[
    \|\theta^{-2}k_{IJ}k_{IJ} + \theta^{-2}e_0(\s)^2 - 1\|_{C^{k_0}} \leq Ct^{4\delta}
    \]
    for all $t \in [t_b,T]$. In particular, 
    \[
    \|\tsum_I \bar p_I^2 + \bar\Psi^2 - 1\|_{C^{k_0}} \leq CT^{4\delta}.
    \]
\end{lemma}

\begin{proof}
    By the Hamiltonian constraint \eqref{hamiltonian constraint with structure coeffs}, taking into account the conventions introduced in Remark~\ref{rmk: schematic notation},
    \[
    \theta^{-2}k_{IJ}k_{IJ} + \theta^{-2}e_0(\s)^2 - 1 = \theta^{-2} \big( \gamma * \gamma - e_I(\s)e_I(\s) - 2V\circ\s \big). 
    \]
    The result follows by estimating the right-hand side of this equation using Lemma~\ref{estimates for powers of theta}.
\end{proof}

\begin{lemma} \label{commutator with N^2}
    Let $u \in C^\infty(\overline{\Omega}_t)$ and $|\I| \leq k_1$, then
    \[
    \|[\p_{\hspace{1pt}\I}, N^2] u\|_{L^2} \leq C\|u\|_{H^{k_1-1}} + Ct^{-A}\|u\|_{C^0} \bbh(t).
    \]
\end{lemma}

\begin{proof}
    We have
    \[
    [\p_{\hspace{1pt}\I},N^2] u = \sum (\p_{\J_1} N)(\p_{\J_2} N) (\p_{\J_3} u), \qquad \J_1 \cup \J_2 \cup \J_3 = \I, \quad |\J_1| + |\J_2| \geq 1.  
    \]
    There are two types of terms that arise in the sum above:
    \[
    N(\p_{\J_1} \p_i N)(\p_{\J_2} u), \qquad (\p_{\J_1} \p_i N)( \p_{\J_2} \p_j N )( \p_{\J_3} u ).
    \]
    Note that $\p_i N$ may be replaced by $\p_i(N-1)$. For the first type, we first estimate $N$ in $C^0$ and then apply Moser estimates, leaving the $\p_i$ derivative fixed. For the second type, we apply Moser estimates leaving both the $\p_i$ and $\p_j$ derivatives fixed. This yields the result.
\end{proof}

\begin{lemma}
    We have
    \[
    \tfrac{2}{3}\lpar^{1/2}\|t\Theta-1\|_{H^{k_1}} \leq \Big( \int_{\Omega_t} \densho{\Theta}(t) \,dx \Big)^{1/2} \leq 2\lpar^{1/2} \|t\Theta-1\|_{H^{k_1}}
    \]
    for all $t \in [t_b,T]$.
\end{lemma}

\begin{proof}
    This follows from the bootstrap assumption, since $\frac{1}{2} < N < \frac{3}{2}$, as $r < \frac{1}{2}$.
\end{proof}

\begin{lemma} \label{lemma: divergence of the frame}
    With the conventions introduced in Subsection~\ref{app: the spacetime domain}, we have
    \[
    \|\diver_\eta e_I\|_{C^0} + \|\diver_\eta (Ne_I)\|_{C^0} + t^A\|Ne_I\|_{H^{k_1}} \leq Ct^{-1+\delta}
    \]
    for all $t \in [t_b,T]$, and 
    \[
    |df(e_I)|(t,x) \leq Ct^{-1+4\delta}
    \]
    for all $(t,x) \in \side_{[t_b,T]}$. 
\end{lemma}

\begin{proof}
    The result follows from the following observations:
    \[
    \diver_\eta e_I = \p_ie_I^i, \qquad \diver_\eta(Ne_I) = e_IN + N\diver_\eta e_I, \qquad |df(e_I)| = \bigg| \sum_i \frac{x^i}{|x|}e_I^i \bigg| \leq |e_I|;
    \]
    in addition to applying Moser estimates to $\|Ne_I\|_{H^{k_1}}$ and the bootstrap assumption.
\end{proof}

\subsection{The frame and the dual frame}

Equations~\eqref{frame equation global} and \eqref{dual frame equation global} can be written as
\begin{subequations}
    \begin{align}
    \p_t e_I^i &= -\frac{\bar p_I}{t} e_I^i - \delta k_{IJ} e_J^i - (N-1)k_{IJ}e_J^i,\label{frame equation global 2}\\
    \p_t \omega_i^I &= \frac{\bar p_I}{t}\omega_i^J + \delta k_{IJ}\omega_i^J + (N-1)k_{IJ}\omega_i^J.
    \end{align}
\end{subequations}

\begin{lemma} \label{lemma: low order estimate frame and dual frame}
    The inequalities
    \begin{align*}
        \lowen{e}(s) &\leq \lowen{e}(T) + \int_s^T \frac{Cr-2\delta}{t} \lowen{e}(t) \,dt,\\
        \lowen{\omega}(s) &\leq \lowen{\omega}(T) + \int_s^T \frac{Cr-2\delta}{t} \lowen{\omega}(t) \,dt
    \end{align*}
    hold for all $s \in [t_b,T]$.
\end{lemma}

\begin{proof}
    We only prove the result for $e$, since the result for $\omega$ is very similar. We begin by estimating some terms on the right-hand side of \eqref{frame equation global 2}. There are no subcritical terms in this case. The critical terms are
    \[
    \delta k * e, \quad (N-1) * k * e, 
    \]
    with $\alpha = 1-4\delta$, $k = k_0+1$ and $\xi = e$. Moving on, 
    \[
    \begin{split}
        \lowen{e}(s) - \lowen{e}(T) &= \int_s^T -\p_t \lowen{e}(t) \,dt\\
        &= \int_s^T -\frac{2(1-4\delta)}{t} \lowen{e}(t) + t^{2(1-4\delta)} \sum_{I,i} \sum_{m=0}^{k_0+1} \sigma_m \sum_{|\I|=m} 2\p_{\hspace{1pt}\I}(-\p_te_I^i) \p_{\hspace{1pt}\I} e_I^i \,dt.
    \end{split}
    \]
    Moreover,
    \[
    \begin{split}
        &t^{2(1-4\delta)}\Bigg| \sum_{I,i} \sum_{m=0}^{k_0+1} \sigma_m \sum_{|\I|=m} 2\p_{\hspace{1pt}\I}(-\p_te_I^i) \p_{\hspace{1pt}\I} e_I^i \Bigg|\\
        &\hspace{1cm}= t^{2(1-4\delta)}\Bigg| \sum_{I,i} \sum_{m=0}^{k_0+1} 2\sigma_m \sum_{|\I|=m} \bigg( \frac{\bar p_I}{t} \p_{\hspace{1pt}\I}e_I^i + \frac{1}{t} [\p_{\hspace{1pt}\I}, \bar p_I ]e_I^i + \p_{\hspace{1pt}\I} \Big( -\p_te_I^i - \frac{\bar p_I}{t}e_I^i \Big) \bigg) \p_{\hspace{1pt}\I}e_I^i \Bigg|\\
        &\hspace{1cm}\leq \frac{2(1-6\delta)}{t} \lowen{e} + \frac{2\delta}{t} \lowen{e} + \frac{Cr}{t} \lowen{e},
    \end{split}
    \]
    where we have used Lemmas~\ref{lemma: initial C^0 bounds on p_I}, \ref{lemma: commutator estimate with crazy sequence} and \eqref{frame equation global 2}. Putting these observations together, we obtain the result.
\end{proof}

\begin{lemma} \label{lemma: high order estimate frame and dual frame}
    We have
    \begin{align*}
        t^{2A}\int_{\Omega_t} -\p_t \densho{e} \,dx &\leq \frac{Cr}{t} \bbh(t)^2 + Ct^{-1+\delta},\\
        t^{2A}\int_{\Omega_t} -\p_t \densho{\omega} \,dx &\leq \frac{Cr}{t}\bbh(t)^2 + Ct^{-1+\delta}
    \end{align*}
    for all $t \in [t_b,T]$.
\end{lemma}

\begin{proof}
    Similarly to the proof of Lemma~\ref{lemma: low order estimate frame and dual frame}, we only prove it for $e$. We estimate the terms on the right-hand side of \eqref{frame equation global 2}. There are no subcritical terms. The critical terms are
    \[
    \delta k * e, \quad  (N-1) * k * e,
    \]    
    with $\alpha = 1-3\delta$. Next, we apply \eqref{complicated estimate 2} with $\eta = \delta$ (to the sum over $I$ only) to obtain
    \[
    \begin{split}
        2t^{2(A+1-3\delta)}\bigg| \sum_{I,i} \Big\langle \frac{\bar p_I}{t} e_I^i, e_I^i \Big\rangle_{H^{k_1}} \bigg| &\leq \frac{2}{t} \big( \delta + \textstyle\max_I \|\bar p_I\|_{C^0} \big) t^{2(A+1-3\delta)}\|e\|_{H^{k_1}}^2\\
        &\quad + \displaystyle\frac{C}{t} \langle 1/\delta \rangle^{k_1-1} \big\langle \tsum_I \|\bar p_I\|_{H^{k_1}} \big\rangle^{k_1} t^{2(A+1-3\delta)} \|e\|_{C^{k_0}} \|e\|_{H^{k_1}}\\
        &\leq \frac{2(1-5\delta)}{t} t^{2(A+1-3\delta)} \|e\|_{H^{k_1}}^2 + Ct^{-1+\delta}.
    \end{split}
    \]
    Finally, we write
    \[
    \begin{split}
        \int_{\Omega_t} -\p_t \densho{e} \,dx = -\frac{2(1-3\delta)}{t} t^{2(1-3\delta)} \|e\|_{H^{k_1}}^2 + \sum_{I,i} 2t^{2(1-3\delta)} \big\langle -\p_te_I^i, e_I^i \big\rangle_{H^{k_1}}.
    \end{split}
    \]
    The result follows from the previous estimates, \eqref{frame equation global 2} and the Cauchy-Schwarz inequality for the $H^{k_1}$ inner product. 
\end{proof}

\subsection{The second fundamental form}

\begin{lemma} \label{lemma: low order estimate second fundamental form}
    We have
    \[
    \lowen{\delta k}(s) \leq \lowen{\delta k}(T) + \int_s^T Ct^{-1+\delta}\,dt,
    \]
    for all $s \in [t_b,T]$.
\end{lemma}

\begin{proof}
    As a consequence of \eqref{k equation global}, $\delta k_{IJ}$ satisfies
    \begin{equation} \label{eq: k equation global low order}
        -\p_t \delta k_{IJ} = \frac{1}{t} \delta k_{IJ} + \subc,
    \end{equation}
    where
    \begin{align*}
        \subc &= N * e * \p\gamma + N * e * \p \ve(\ln N) + \frac{1}{t}(t\Theta-1)k_{IJ} + N * \gamma * \gamma\\
        &\quad - Ne_I(\ln N)e_J(\ln N) + N * \gamma * \ve(\ln N) - Ne_I(\s)e_J(\s) - N(V\circ\s)\delta_{IJ},
    \end{align*}
    with $\alpha = 1$ and $k = k_0+1$. Now we can write $\lowen{\delta k}(s) - \lowen{\delta k}(T) = \int_s^T -\p_t\lowen{\delta k}(t) \,dt$ and use \eqref{eq: k equation global low order} to obtain the result.
\end{proof}

\begin{lemma} \label{lemma: high order estimate second fundamental form}
    we have
    \[
    \begin{split}
        t^{2A} \int_{\Omega_t} -\p_t\densho{\delta k} \,dx &\leq 2t^{2(A+1)} \big\langle Ne_K \gamma_{KIJ} + Ne_{I} \gamma_{JKK} - Ne_Ie_J(\ln N), \delta k_{IJ} \big\rangle_{H^{k_1}(\Omega_t)}\\
        &\quad + \frac{C_* + Cr}{t} \bbh(t)^2 + Ct^{-1+\delta} 
    \end{split}
    \]
    for all $t \in [t_b,T]$.
\end{lemma}

\begin{proof}
    As a consequence of \eqref{k equation global}, $\delta k_{IJ}$ satisfies 
    \begin{equation} \label{k equation global 2}
    \begin{split}
        -\p_t \delta k_{IJ}  &= Ne_K \gamma_{K(IJ)} + Ne_{(I} \gamma_{J)KK} - N e_{(I}e_{J)}(\ln N)\\
        &\quad + \frac{1}{t} \delta k_{IJ} + \frac{1}{t^2} (t\Theta-1) \bar p_I \delta_{IJ} + \subc + \crit, 
    \end{split}
    \end{equation}
    where
    \begin{align*}
    \subc &= N * \gamma * \gamma - Ne_I(\ln N)e_J(\ln N) + N * \gamma * \ve(\ln N) - Ne_I(\s)e_J(\s) - N(V\circ\s)\delta_{IJ},\\
    \crit &= \frac{1}{t} (t\Theta-1) \delta k_{IJ},   
    \end{align*}
    with $\alpha = 1$. Next,
    \[
    \begin{split}
        t^{2(A+1)} \big|\tsum_I\big\langle \tfrac{1}{t^2}(t\Theta-1) \bar p_I\delta_{IJ}, \delta k_{IJ} \big\rangle_{H^{k_1}} \big| &= t^{2A} \big| \tsum_I \big\langle (t\Theta-1)\bar p_I, \delta k_{II} \big\rangle_{H^{k_1}} \big|\\
        &\leq t^{2A} ( \delta + \textstyle\max_I\|\bar p_I\|_{C^0} ) \sqrt{3} \|t\Theta-1\|_{H^{k_1}} \|\delta k\|_{H^{k_1}}\\
        &\quad + Ct^{2A}\|t\Theta-1\|_{C^{k_0}}\|\delta k\|_{H^{k_1}}\\
        &\leq \frac{C_*}{t} \bbh(t)^2 + Ct^{-1+\delta},
    \end{split}
    \]
    where we have used \eqref{complicated estimate 2}. The result follows by applying $-\p_t$ to $\densho{\delta k}$, using \eqref{k equation global 2}, and the above estimates. Note that we use the symmetry of $\delta k_{IJ}$ to remove the symmetrizations from the first three terms on the right-hand side of \eqref{k equation global 2}.
\end{proof}

\subsection{The mean curvature}

\begin{lemma} \label{lemma: low order estimate mean curvature}
    We have
    \[
    \lowen{\Theta}(s) \leq \lowen{\Theta}(T) + \int_s^T \frac{6\delta-2\lpar}{t} \lowen{\Theta}(t) + Ct^{-1+\delta} \,dt
    \]
    for all $s \in [t_b,T]$.
\end{lemma}

\begin{proof}
    Equation~\eqref{mean curvature low order equation} can be written as
    \begin{equation} \label{mean curvature low order equation 2}
        -\p_t(t\Theta-1) = -\frac{\lpar}{t}(t\Theta-1) + \subc,
    \end{equation}
    where
    \[
    \begin{split}
        \subc &= -\frac{\lpar}{t}(t\Theta-1)^2 + tN^2 * e * \p\gamma + tN^2 *\gamma*\gamma + tN^2 * e * \p \ve(\ln N)\\
        &\quad - tN^2 e_I(\ln N) e_I(\ln N) + tN^2 * \gamma * \ve(\ln N) - tN^2e_I(\s)e_I(\s) - 3tN^2 V\circ\s,
    \end{split}
    \]
    with $\alpha = -3\delta$ and $k = k_0+1$. Now we can write $\lowen{\Theta}(s) - \lowen{\Theta}(T) = \int_s^T -\p_t\lowen{\Theta}(t) \,dt$ and use \eqref{mean curvature low order equation 2} to obtain the result.
\end{proof}

\begin{lemma} \label{lemma: high order estimate mean curvature}
    If $T$ is small enough, then
    \[
    \begin{split}
        t^{2A} \int_{\Omega_t} -\p_t\densho{\Theta} \,dx &\leq 2\lpar t^{2A+1} \big\langle -e_Ie_I(\ln N), t\Theta-1 \big\rangle_{H^{k_1}(\Omega_t)} + \frac{C_* + Cr}{t} \bbh(t)^2 + Ct^{-1+\delta}
    \end{split}
    \]
    for all $t \in [t_b,T]$.
\end{lemma}

\begin{proof}
    We can rewrite \eqref{mean curvature high order equation} as
    \begin{equation*} 
    \begin{split}
        \p_t(t\Theta-1) &= tN^2 \Big( e_Ie_I(\ln N) + e_I(\ln N) e_I(\ln N) + \gamma * \ve(\ln N) - \delta k_{IJ} \delta k_{IJ}\\
        &\quad - \frac{2}{t}\tsum_I \bar p_I \delta k_{II} - \displaystyle\frac{1}{t^2} \tsum_I \bar p_I^2 - \delta\s_0^2 - \displaystyle\frac{2}{t}\bar\Psi \delta\s_0 - \frac{1}{t^2} \bar\Psi^2 + V\circ\s \Big)\\
        &\quad + \frac{1}{t} + \frac{\lpar+2}{t}(t\Theta-1) + \frac{\lpar+1}{t}(t\Theta-1)^2.
    \end{split}
    \end{equation*}
    We further rewrite the equation by noting that $N^2 = (N-1)^2 + 2(N-1) + 1$. Thus,
    \begin{equation} \label{eq: mean curvature high order equation 3}
    \begin{split}
        -\p_t(t\Theta-1) &= -tN^2e_Ie_I(\ln N) + \frac{2}{t}(N-1)\big( \tsum_I \bar p_I^2 + \bar\Psi^2 \big)\\
        &\quad + 2\big( \tsum_I \bar p_I \delta k_{II} + \bar\Psi \delta\s_0 \big) - \displaystyle\frac{\lpar+2}{t}(t\Theta-1) + \subc + \crit, 
    \end{split}
    \end{equation}
    where
    \begin{align*}
    \begin{split}
        \subc &= -tN^2 e_I(\ln N)e_I(\ln N) + tN^2 * \gamma * \ve(\ln N) - tN^2 V\circ\s\\
        &\quad - \frac{\lpar+1}{t}(t\Theta-1)^2 + \frac{1}{t} \big(\tsum_I \bar p_I^2 + \bar\Psi^2 -1 \big).
    \end{split}\\
    \begin{split}
        \crit &= t\big( (N-1)^2 + 2(N-1) + 1 \big) (\delta k_{IJ} \delta k_{IJ} + \delta\s_0^2) + \frac{1}{t} (N-1)^2 \big( \tsum_I \bar p_I^2 + \bar\Psi^2 \big)\\
        &\quad + 2\big( (N-1)^2 + 2(N-1)\big)\big( \tsum_I \bar p_I \delta k_{II} + \bar\Psi \delta\s_0\big),
    \end{split}
    \end{align*}
    with $\alpha = 0$. Next, by Lemma~\ref{commutator with N^2}, 
    \[
    \begin{split}
        t^{A+1}\big\|[\p_{\hspace{1pt}\I},N^2] \big( e_Ie_I(\ln N) \big)\big\|_{L^2} &\leq Ct^{A+1}\|e_Ie_I(\ln N)\|_{H^{k_1-1}} + Ct\|e_Ie_I(\ln N)\|_{C^0}\bbh(t)\\
        &\leq Ct^{-1+\delta},
    \end{split}
    \]
    for $|\I| \leq k_1$, where we have used Moser estimates. Moving on, since $r < \frac{1}{2}$, then $1/N^{2} < 4$. Therefore,
    \[
    \begin{split}
        &\bigg|t^{2A-1} \int_{\Omega_t} \sum_{|\I| \leq k_1} \frac{2\lpar}{N^2} \p_{\hspace{1pt}\I}\Big( 2(N-1)\big( \tsum_I \bar p_I^2 + \bar\Psi^2 \big) \Big) \p_{\hspace{1pt}\I}(t\Theta-1) \,dx\bigg|\\
        &\hspace{3cm}\leq 16\lpar t^{2A-1}  \bigsob{k_1}{(N-1)\big( \tsum_I \bar p_I^2 + \bar\Psi^2 \big)} \sob{k_1}{t\Theta-1}\\
        &\hspace{3cm}\leq 16\lpar t^{2A-1}\big( \delta + \|\tsum_I \bar p_I^2 + \bar\Psi^2\|_{C^0} \big) \|N-1\|_{H^{k_1}} \|t\Theta-1\|_{H^{k_1}}\\
        &\hspace{3cm}\quad + Ct^{2A-1} \|N-1\|_{C^{k_0}} \|t\Theta-1\|_{H^{k_1}}\\
        &\hspace{3cm}\leq \frac{C_*}{t} \bbh(t)^2 + Ct^{-1+\delta},
    \end{split}
    \]
    by Lemma~\ref{estimate on the kasner relations} (if $CT^{4\delta} \leq \delta$) and \eqref{complicated estimate 1}. Finally,
    \[
    \begin{split}
        &\Bigg| t^{2A} \int_{\Omega_t} \sum_{|\I| \leq k_1} \frac{2\lpar}{N^2} \p_{\hspace{1pt}\I}\Big(2\big( \tsum_I \bar p_I \delta k_{II} + \bar\Psi \delta\s_0 \big)\Big) \p_{\hspace{1pt}\I}(t\Theta-1) \,dx \Bigg|\\
        &\hspace{2cm}\leq 16\lpar t^{2A} \bigsob{k_1}{\tsum_I \bar p_I \delta k_{II} + \bar\Psi \delta\s_0} \sob{k_1}{t\Theta-1}\\
        &\hspace{2cm}\leq 16\lpar t^{2A}\big( \delta + \|\tsum_I \bar p_I^2 + \bar\Psi^2\|_{C^0}^{1/2} \big)\big( \|\delta k\|_{H^{k_1}}^2 + \|\delta\s_0\|_{H^{k_1}}^2 \big)^{1/2} \|t\Theta-1\|_{H^{k_1}}\\
        &\hspace{2cm}\quad + Ct^{2A} \big( \|\delta k\|_{C^{k_0}} + \|\delta\s_0\|_{C^{k_0}} \big)\|t\Theta-1\|_{H^{k_1}}\\
        &\hspace{2cm}\leq \frac{C_*}{t}\bbh(t)^2 + Ct^{-1+\delta},
    \end{split}
    \]
    again by Lemma~\ref{estimate on the kasner relations} (if $CT^{2\delta} \leq \delta$) and \eqref{complicated estimate 1}. Taking into account \eqref{lapse equation global}, we can write
    \[
    \begin{split}
        t^{2A} \int_{\Omega_t} -\p_t \densho{\Theta} \,dx &= t^{2A} \int_{\Omega_t} \frac{2(\lpar+1)}{t}(t\Theta-1) \sum_{|\I| \leq k_1} \frac{\lpar}{N^2} \big( \p_{\hspace{1pt}\I}(t\Theta-1) \big)^2 \,dx\\
        &\quad + t^{2A} \int_{\Omega_t} \sum_{|\I| \leq k_1} \frac{2\lpar}{N^2} \p_{\hspace{1pt}\I}\big( -\p_t(t\Theta-1) \big) \p_{\hspace{1pt}\I}(t\Theta-1) \;dx. 
    \end{split}
    \]
    The result now follows by using Equation~\eqref{eq: mean curvature high order equation 3} and the above estimates.
\end{proof}

\subsection{The lapse}

Equation \eqref{lapse equation global} may be written as
\begin{equation} \label{lapse equation global 2}
    -\p_t(N-1) = -\frac{\lpar+1}{t} (N-1)(t\Theta-1) - \frac{\lpar+1}{t}(t\Theta-1).
\end{equation}

\begin{lemma} \label{lemma: low order estimate lapse} 
    We have
    \[
    \lowen{N}(s) \leq \int_s^T Ct^{-1+\delta}\,dt,
    \]
    for all $s \in [t_b,T]$.
\end{lemma}

\begin{proof}
    The right-hand side of \eqref{lapse equation global 2} can be written as just one subcritical term,
    \[
    -\frac{\lpar+1}{t} N(t\Theta-1),
    \]
    with $\alpha = 0$ and $k = k_0+1$. Since $N(T) = 1$, then $\lowen{N}(s) = \int_s^T -\p_t \lowen{N}(t)\,dt$, and the result follows by using \eqref{lapse equation global 2}.
\end{proof}

\begin{lemma} \label{lemma: high order estimate lapse}
    We have
    \[
    \begin{split}
        t^{2A} \int_{\Omega_t} -\p_t \densho{N} \,dx \leq \frac{C_* + Cr}{t} \bbh(t)^2 + Ct^{-1+\delta} 
    \end{split}
    \]
    for all $t \in [t_b,T]$.
\end{lemma}

\begin{proof}
    There is one critical term,
    \[
    -\frac{\lpar+1}{t}(N-1)(t\Theta-1),
    \]
    with $\alpha = 0$. The result follows by applying $-\p_t$ to $\densho{N}$ and using \eqref{lapse equation global 2}.
\end{proof}

\begin{lemma} \label{lemma: low order estimate lapse derivatives}
    We have
    \[
    \lowen{\ve(\ln N)}(s) \leq \int_s^T \frac{Cr-2\delta}{t} \lowen{\ve(\ln N)}(t) + Ct^{-1+\delta} \;dt,
    \]
    for all $s \in [t_b,T]$.
\end{lemma}

\begin{proof}
    Equation~\eqref{lapse derivatives low order equation} can be written as
    \begin{equation} \label{eq: lapse derivatives low order 2}
        -\p_te_I(\ln N) = \frac{\bar p_I}{t} e_I(\ln N) - \frac{\lpar+1}{t} e_I(t\Theta-1) + \crit,
    \end{equation}
    where the only subcritical term is $-\frac{\lpar+1}{t}e_I(t\Theta-1)$, and
    \[
    \crit = (N-1) * k * \ve(\ln N) + \delta k * \ve(\ln N), 
    \]
    with $\alpha = 1-4\delta$, $k = k_0$ and $\xi = \ve(\ln N)$. Note that $\lowen{\ve(\ln N)}(s) = \int_s^T -\p_t \lowen{\ve(\ln N)}(t) \,dt$. The result follows similarly as in the proof of Lemma~\ref{lemma: low order estimate frame and dual frame}, by using \eqref{eq: lapse derivatives low order 2} and commuting $\p_{\hspace{1pt}\I}$ with $\bar p_I$ with Lemma~\ref{lemma: commutator estimate with crazy sequence}.
\end{proof}

\begin{lemma} \label{lemma: high order estimate lapse derivatives}
    If $T$ is small enough, then
    \[
    \begin{split}
        t^{2A} \int_{\Omega_t} -\p_t \densho{\ve(\ln N)} \,dx &\leq 2t^{2(A+1)} \big\langle -\lpar e_I\Theta - Ne_J\delta k_{JI}, e_I(\ln N)\big\rangle_{H^{k_1}(\Omega_t)}\\
        &\quad + \frac{C_* + Cr}{t}\bbh(t)^2 + Ct^{-1+\delta} 
    \end{split}
    \]
    for all $t \in [t_b,T]$. 
\end{lemma}

\begin{proof}
    Equation~\eqref{lapse derivatives high order equation} can be written as
    \begin{equation} \label{eq: lapse derivatives high order 2}
    \begin{split}
        -\p_t e_I(\ln N) &= -\lpar e_I\Theta - Ne_J\delta k_{JI} + \frac{1}{t}\gamma_{IKK} \bar p_I + \frac{1}{t} \tsum_J \gamma_{JIJ} \bar p_J + \displaystyle\frac{1}{t} \bar\Psi e_I\s\\
        &\quad - \frac{1}{t}e_I(\ln N) + \frac{\bar p_I}{t} e_I(\ln N) + \subc + \crit,
    \end{split}
    \end{equation}
    where
    \begin{align*}
        \subc &= -\frac{1}{t} Ne_I \bar p_I - \frac{1}{t}(t\Theta-1) e_I(\ln N),\\
        \begin{split}
            \crit &= (N-1) * \gamma * k + \gamma * \delta k + (N-1)e_0(\s)e_I(\s)\\
            &\quad + \delta\s_0 e_I(\s) + (N-1) * k * \ve(\ln N) + \delta k * \ve(\ln N),
    \end{split} 
    \end{align*}
    with $\alpha = 1$. Next, by \eqref{complicated estimate 1},
    \begin{equation} \label{eq: momentum constraint terms estimate}
    \begin{split}
        &t^{A+1}\Big(\tsum_I \big\| \tsum_J \frac{1}{t}\gamma_{JIJ}\bar p_J + \frac{1}{t}\bar\Psi e_I\s\big\|_{H^{k_1}}^2\Big)^{1/2}\\
        &\hspace{2cm}\leq t^A\big( \delta + \|\tsum_J \bar p_J^2 + \bar\Psi^2\|_{C^0}^{1/2} \big)\big( \tsum_{I,J}\|\gamma_{JIJ}\|_{H^{k_1}}^2 + \|\ve\s\|_{H^{k_1}}^2 \big)^{1/2}\\
        &\hspace{2cm}\quad+ Ct^A\big( \|\gamma\|_{C^{k_0}} + \|\ve\s\|_{C^{k_0}} \big)\\
        &\hspace{2cm}\leq \frac{C_*}{t}\bbh(t) + Ct^{-1+\delta},
    \end{split}
    \end{equation}
    where we have used the triangle inequality and Lemma~\ref{estimate on the kasner relations} (with $CT^{2\delta} \leq \delta$). Moving on, by \eqref{complicated estimate 2},
    \[
    \begin{split}
        &t^{2(A+1)}\big| \tsum_I \big\langle \tfrac{1}{t} \gamma_{IKK} \bar p_I, e_I(\ln N) \big\rangle_{H^{k_1}}\big|\\
        &\hspace{2cm}\leq t^{2A+1}\big( \delta + \textstyle\max_I\|\bar p_I\|_{C^0} \big)\big( \tsum_I \|\gamma_{IKK}\|_{H^{k_1}}^2 \big)^{1/2}\|\ve(\ln N)\|_{H^{k_1}}\\
        &\hspace{2cm}\quad + Ct^{2A+1}\|\gamma\|_{C^{k_0}}\|\ve(\ln N)\|_{H^{k_1}}\\
        &\hspace{2cm}\leq \frac{C_*}{t} \bbh(t)^2 + Ct^{-1+\delta}.
    \end{split}
    \]
    Finally, again by \eqref{complicated estimate 2},
    \[
    \begin{split}
        t^{2(A+1)} \Big| \tsum_I \Big\langle \displaystyle\frac{\bar p_I}{t}e_I(\ln N), e_I(\ln N) \Big\rangle_{H^{k_1}} \Big| &\leq t^{2A+1}\big( \delta + \textstyle\max_I\|\bar p_I\|_{C^0} \big) \|\ve(\ln N)\|_{H^{k_1}}^2\\
        &\quad + Ct^{2A+1}\|\ve(\ln N)\|_{C^{k_0}}\|\ve(\ln N)\|_{H^{k_1}}\\
        &\leq t^{2A+1}(1-5\delta)\|\ve(\ln N)\|_{H^{k_1}}^2 + Ct^{-1+\delta}.
    \end{split}
    \]
    The result follows by applying $-\p_t$ to $\densho{\ve(\ln N)}$, using \eqref{eq: lapse derivatives high order 2} and the above inequalities.
\end{proof}

\subsection{The structure coefficients}

\begin{lemma} \label{lemma: low order estimate gammas}
    We have
    \[
    \lowen{\gamma}(s) \leq \lowen{\gamma}(T) + \int_s^T \frac{Cr - 4\delta}{t} \lowen{\gamma}(t) + Ct^{-1+\delta} \,dt
    \]
    for all $s \in [t_b,T]$.
\end{lemma}

\begin{proof}
    We can write \eqref{gamma equation global} as
    \begin{equation} \label{eq: gamma equation global low order}
    -\p_t\gamma_{IJK} = \frac{\bar p_I + \bar p_J - \bar p_K}{t}\gamma_{IJK} + \subc + \crit,
    \end{equation}
    where
    \begin{align*}
        \subc &= N * e * \p\delta k + \frac{1}{t} N * e * \p\bar p + N * \ve(\ln N) * k,\\
        \crit &= (N-1) * k * \gamma + \delta k * \gamma,
    \end{align*}
    with $\alpha = 1-3\delta$, $k = k_0$ and $\xi = \gamma$. Now we can write $\lowen{\gamma}(s) - \lowen{\gamma}(t) = \int_s^T -\p_t\lowen{\gamma}(t) \,dt$, and the result follows similarly as in the proof of Lemma~\ref{lemma: low order estimate frame and dual frame}, by using \eqref{eq: gamma equation global low order} and commuting $\p_{\hspace{1pt}\I}$ with $\bar p_I + \bar p_J - \bar p_K$ with Lemma~\ref{lemma: commutator estimate with crazy sequence}.
\end{proof}

\begin{lemma} \label{lemma: high order estimate gammas}
    We have
    \[
    \begin{split}
        t^{2A} \int_{\Omega_t} -\p_t \densho{\gamma} \,dx &\leq 2t^{2(A+1)} \big\langle Ne_I\delta k_{JK}, \gamma_{IJK} \big\rangle_{H^{k_1}(\Omega_t)} + \frac{C_* + Cr}{t} \bbh(t)^2 + Ct^{-1+\delta}
    \end{split}
    \]
    for all $t \in [t_b,T]$.
\end{lemma}

\begin{proof}
    We can write \eqref{gamma equation global} as
    \begin{equation} \label{eq: gamma equation global high order}
    \begin{split}
        -\p_t\gamma_{IJK} &= - Ne_J\delta k_{IK} + Ne_I\delta k_{JK} + \frac{\bar p_I + \bar p_J - \bar p_K}{t}\gamma_{IJK}\\
        &\quad - \frac{1}{t} e_J(\ln N)\bar p_I \delta_{IK} + \frac{1}{t} e_I(\ln N)\bar p_J\delta_{JK} + \subc + \crit,
    \end{split}
    \end{equation}
    where,
    \begin{align*}
        \subc &= \frac{1}{t} N * e * \p\bar p\\
        \crit &= (N-1) * k * \gamma + \delta k * \gamma + (N-1) * \ve(\ln N) * k + \ve(\ln N) * \delta k,
    \end{align*}
    with $\alpha = 1$. Next,
    \[
    \begin{split}
        &\sum_{I,J,K} t^{2(A+1)} \big\langle \tfrac{1}{t} (\bar p_I + \bar p_J - \bar p_K)\gamma_{IJK}, \gamma_{IJK} \big\rangle_{H^{k_1}}\\
        &\hspace{3cm}\leq t^{2A+1}\Big( \delta + \max_{I \neq J}\sup_{x \in \overline{\Omega}_t}\big(\bar p_I(x) + \bar p_J(x) - \bar p_K(x)\big) \Big) \|\gamma\|_{H^{k_1}}^2\\
        &\hspace{3cm}\quad + Ct^{2A+1} \|\gamma\|_{C^{k_0}} \|\gamma\|_{H^{k_1}}\\
        &\hspace{3cm}\leq \frac{1-5\delta}{t} t^{2(A+1)}\|\gamma\|_{H^{k_1}}^2 + Ct^{-1+\delta},
    \end{split}
    \]
    where we have used \eqref{complicated estimate 3}. Also, by \eqref{complicated estimate 4},
    \[
    \begin{split}
        &t^{2(A+1)} \big| \tsum_{I,J} \big\langle -\frac{1}{t}e_J(\ln N) \bar p_I\delta_{IK} + \frac{1}{t}e_I(\ln N) \bar p_J\delta_{JK}, \gamma_{IJK} \big\rangle_{H^{k_1}} \big|\\
        &\hspace{2cm}= 2t^{2A+1}\big|\tsum_I \big\langle e_J(\ln N) \bar p_I, \gamma_{JII} \big\rangle_{H^{k_1}} \big|\\
        &\hspace{2cm}\leq 2t^{2A+1}\big( \delta + \|\tsum_I \bar p_I^2\|_{C^0}^{1/2} \big)\|\ve(\ln N)\|_{H^{k_1}} \|\gamma\|_{H^{k_1}} + Ct^{2A+1}\|\ve(\ln N)\|_{C^{k_0}}\|\gamma\|_{H^{k_1}}\\
        &\hspace{2cm}\leq \frac{C_*}{t} \bbh(t)^2 + Ct^{-1+\delta}.
    \end{split}
    \]
    The result follows by applying $-\p_t$ to $\densho{\gamma}$ and using \eqref{eq: gamma equation global high order}. Note that we rewrite the terms that arise from the first two terms on the right-hand side of \eqref{eq: gamma equation global high order} by using the antisymmetry of $\gamma_{IJK}$.
\end{proof}

\subsection{The scalar field}

\begin{lemma} \label{lemma: low order estimate phi time derivative}
    We have
    \[
    \lowen{\delta\s_0}(s) \leq \lowen{\delta\s_0}(T) + \int_s^T Ct^{-1+\delta} \,dt
    \]
    for all $s \in [t_b,T]$.
\end{lemma}

\begin{proof}
    We can write \eqref{eq: phi equation global} as
    \begin{equation} \label{eq: phi equation global low order}
    -\p_t \delta\s_0 = \frac{1}{t} \delta\s_0 + \subc,
    \end{equation}
    where
    \[
    \subc = N * e * \p \ve\s + N * \gamma * \ve\s + N * \ve\s * \ve(\ln N) + NV'\circ\s + \frac{1}{t}(t\Theta-1)e_0\s,
    \]
    with $\alpha = 1$ and $k = k_0+1$. Now we can write $\lowen{\delta\s_0}(s) - \lowen{\delta\s_0}(T) = \int_s^T -\p_t\lowen{\delta\s_0}(t)\,dt$ and use \eqref{eq: phi equation global low order} to obtain the result.
\end{proof}

\begin{lemma} \label{lemma: low order estimate phi spatial derivatives}
    We have
    \[
    \lowen{\ve\s}(s) \leq \lowen{\ve\s}(T) + \int_s^T \frac{Cr-4\delta}{t} \lowen{\ve\s}(t) + Ct^{-1+\delta} \,dt
    \]
    for all $s \in [t_b,T]$.
\end{lemma}

\begin{proof}
    We can write \eqref{eq: phi derivatives equation global} as
    \begin{equation} \label{eq: phi derivatives equation global low order}
    -\p_te_I\s = \frac{\bar p_I}{t} e_I\s + \subc + \crit,
    \end{equation}
    where
    \begin{align*}
        \subc &= N * e * \p e_0\s - Ne_I(\ln N)e_0\s,\\
        \crit &= (N-1) * k * \ve\s + \delta k * \ve\s, 
    \end{align*}
    with $\alpha = 1-3\delta$, $k = k_0$ and $\xi = \ve\s$. The result follows similarly as in the proof of Lemma~\ref{lemma: low order estimate frame and dual frame}, by writing $\lowen{\ve\s}(s) - \lowen{\ve\s}(T) = \int_s^T -\p_t\lowen{\ve\s}(t)\,dt$, using \eqref{eq: phi derivatives equation global low order} and commuting $\p_{\hspace{1pt}\I}$ and $\bar p_I$ with Lemma~\ref{lemma: commutator estimate with crazy sequence}.
\end{proof}

\begin{lemma} \label{lemma: high order estimate phi time derivative}
    We have
    \[
    \begin{split}
        t^{2A} \int_{\Omega_t} -\p_t \densho{\delta\s_0} \,dx \leq 2t^{2(A+1)} \big\langle -Ne_Ie_I\s, \delta\s_0 \big\rangle_{H^{k_1}(\Omega_t)} + \frac{C_* + Cr}{t} \bbh(t)^2 + Ct^{-1+\delta}
    \end{split}
    \]
    for all $t \in [t_b,T]$.
\end{lemma}

\begin{proof}
    As a consequence of \eqref{eq: phi equation global}, $\delta\s_0$ satisfies
    \begin{equation} \label{eq: phi equation global high order}
    \begin{split}
        -\p_t \delta\s_0 = -Ne_Ie_I\s + \frac{1}{t} \delta\s_0 + \frac{1}{t^2}(t\Theta-1)\bar\Psi + \subc + \crit,
    \end{split}
    \end{equation}
    where
    \begin{align*}
        \subc &= N * \gamma * \ve\s - Ne_I(\s)e_I(\ln N) + NV'\circ\s,\\
        \crit &= \frac{1}{t}(t\Theta-1)\delta\s_0,
    \end{align*}
    with $\alpha = 1$. Next, by \eqref{complicated estimate 1} and Lemma~\ref{lemma: initial estimate on norm der of phi},
    \[
    \begin{split}
        t^{A+1}\big\|\tfrac{1}{t^2}(t\Theta-1)\bar\Psi\big\|_{H^{k_1}} &\leq t^{A-1}\big( \delta + \|\bar\Psi\|_{C^0} \big)\|t\Theta-1\|_{H^{k_1}} + Ct^{A-1}\|t\Theta-1\|_{C^{k_0}}\\
        &\leq \frac{C_*}{t}\bbh(t) + Ct^{-1+\delta}.
    \end{split}
    \]
    The result follows by applying $-\p_t$ to $\densho{\delta\s_0}$ and using \eqref{eq: phi equation global high order}.
\end{proof}

\begin{lemma} \label{lemma: high order estimate phi spatial derivatives}
    We have
    \[
    \begin{split}
        t^{2A} \int_{\Omega_t} -\p_t \densho{\ve\s} \,dx \leq 2t^{2(A+1)} \big\langle -Ne_I\delta\s_0, e_I\s \big\rangle_{H^{k_1}(\Omega_t)} + \frac{C_* + Cr}{t}\bbh(t)^2 + Ct^{-1+\delta}
    \end{split}
    \]
    for all $t \in [t_b,T]$.
\end{lemma}

\begin{proof}
    We can write \eqref{eq: phi derivatives equation global} as
    \begin{equation} \label{eq: phi derivatives equation global high order}
    \begin{split}
        -\p_te_I\s = -Ne_I\delta\s_0 + \frac{\bar p_I}{t} e_I\s - \frac{1}{t}e_I(\ln N)\bar\Psi + \subc + \crit,
    \end{split}
    \end{equation}
    where 
    \begin{align*}
        \subc &= \frac{1}{t}N * e * \p\bar\Psi\\
        \crit &= (N-1) * k * \ve\s + \delta k * \ve\s - (N-1)e_I(\ln N)e_0\s - e_I(\ln N) \delta\s_0,
    \end{align*}
    with $\alpha = 1$. Moreover, by \eqref{complicated estimate 2},
    \[
    \begin{split}
        t^{2A+1} \big| \tsum_I \big\langle \bar p_I e_I\s, e_I\s \big\rangle_{H^{k_1}}\big| &\leq \big( \delta + \textstyle\max_I \|\bar p_I\|_{C^0} \big) t^{2A+1}\|\ve\s\|_{H^{k_1}}^2\\
        &\quad + Ct^{2A+1} \|\ve\s\|_{C^{k_0}} \|\ve\s\|_{H^{k_1}}\\
        &\leq (1-5\delta)t^{2A+1}\|\ve\s\|_{H^{k_1}}^2 + Ct^{-1+\delta}.
    \end{split}
    \]
    Finally, by \eqref{complicated estimate 2},
    \[
    \begin{split}
        t^{2(A+1)}\big| \big\langle \tfrac{1}{t} \bar\Psi e_I(\ln N), e_I\s \big\rangle_{H^{k_1}} \big| &\leq t^{2A+1} \big( \delta + \|\bar\Psi\|_{C^0} \big)\|\ve(\ln N)\|_{H^{k_1}} \|\ve\s\|_{H^{k_1}}\\
        &\quad + Ct^{2A+1} \|\ve(\ln N)\|_{C^{k_0}} \|\ve\s\|_{H^{k_1}}\\
        &\leq \frac{C_*}{t}\bbh(t)^2 + Ct^{-1+\delta}.
    \end{split}
    \]
    The result follows by applying $-\p_t$ to $\densho{\ve\s}$ and using \eqref{eq: phi derivatives equation global high order}.
\end{proof}

\subsection{The main energy estimates}

Below we will use the conventions introduced in Subsection~\ref{app: the spacetime domain} regarding the spacetime domain $\Omega_{[t_b,T]}$. Before proceeding with the main energy estimates, we need to deal with one of the terms appearing in Lemma~\ref{lemma: high order estimate second fundamental form}. This can be done by using the constraint equations.

\begin{lemma} \label{lemma: constraint equations trick}
    If $T$ is small enough, then
    \[
    \begin{split}
        \Big|\int_s^T t^{2(A+1)} \big\langle Ne_I\gamma_{JKK}, \delta k_{IJ} \big\rangle_{H^{k_1}(\Omega_t)} \,dt\Big| &\leq \int_s^T \frac{C_* + Cr}{t}\bbh(t)^2 + Ct^{-1+\delta} \,dt\\
        &\quad + C\int_{\side_{[s,T]}} t^{2A}\rho(t)\, \frac{t^{-1+4\delta}}{|df|} \,\mu_{\side_{[s,T]}}
    \end{split}
    \]
    for all $s \in [t_b,T]$.
\end{lemma}

\begin{remark}
    In the global setting there is no side boundary. Hence the integral over $\side_{[s,T]}$ does not appear. 
\end{remark}

\begin{proof}
    We begin by estimating
    \[
    t^{2(A+1)} \big\langle \gamma_{JKK}, Ne_I\delta k_{IJ} - e_J\Theta \big\rangle_{H^{k_1}}.
    \]
    Note that, by the momentum constraint,
    \[
    \begin{split}
        Ne_I\delta k_{IJ} - e_J\Theta &= \frac{1}{t}\gamma_{JLL} \bar p_J + \frac{1}{t}\tsum_I\gamma_{IJI} \bar p_I - \displaystyle\frac{1}{t}e_J(\ln N) + \frac{1}{t}\bar\Psi e_J\s + \subc + \crit,
    \end{split}
    \]
    where
    \begin{align*}
        \subc &= \frac{1}{t}N * e * \p\bar p - \frac{1}{t}(t\Theta-1)e_J(\ln N),\\
        \crit &= (N-1) * \gamma * k + \gamma * \delta k + (N-1) e_0(\s)e_J(\s) + \delta\s_0 e_J\s,
    \end{align*}
    with $\alpha = 1$. Moreover, note that
    \[
    \begin{split}
        &t^{2(A+1)}\big|\tsum_J \big\langle \gamma_{JKK}, \tfrac{1}{t}\gamma_{JLL} \bar p_J \big\rangle_{H^{k_1}}\big|\\
        &\hspace{2cm}\leq \big( \delta + \textstyle\max_I\|\bar p_I\|_{C^0} \big)t^{2A+1}\tsum_J \|\gamma_{JKK}\|_{H^{k_1}}^2 + Ct^{2A+1}\|\gamma\|_{C^{k_0}}\|\gamma\|_{H^{k_1}}\\
        &\hspace{2cm}\leq \frac{C_*}{t} \bbh(t)^2 + Ct^{-1+\delta},
    \end{split}
    \]
    where we have used \eqref{complicated estimate 2}. So, taking into account \eqref{eq: momentum constraint terms estimate}, we obtain
    \begin{equation} \label{eq: momentum constraint trick estimate}
    t^{2(A+1)} \big|\big\langle \gamma_{JKK},Ne_I\delta k_{IJ} - e_J\Theta \big\rangle_{H^{k_1}}\big| \leq  \frac{C_* + Cr}{t}\bbh(t)^2 + Ct^{-1+\delta}.
    \end{equation}
    Next, we estimate
    \[
    t^{2A+1} \big\langle e_J\gamma_{JKK}, t\Theta-1 \big\rangle_{H^{k_1}}.
    \]
    By the Hamiltonian constraint, we have
    \[
    \begin{split}
        2e_J\gamma_{JKK} &= \frac{2}{t} \tsum_I \bar p_I\delta k_{II} + \displaystyle\frac{2}{t}\bar\Psi\delta\s_0 - \frac{2}{t}\delta k_{II} + \subc + \crit,
    \end{split}
    \]
    where
    \begin{align*}
        \subc &= \gamma * \gamma + e_I(\s)e_I(\s) + \frac{1}{t^2} \tsum_I  \bar p_I^2 + \displaystyle\frac{1}{t^2} \bar\Psi^2 - \frac{1}{t^2} + 2V\circ\s,\\
        \crit &= \delta k_{IJ}\delta k_{IJ} + \delta\s_0^2 - (\delta k_{II})^2,
    \end{align*}
    with $\alpha = 1$. Next, by \eqref{complicated estimate 1},
    \[
    \begin{split}
        t^{A+1}\big\|\tfrac{2}{t} \tsum_I  \bar p_I \delta k_{II} + \frac{2}{t} \bar\Psi\delta\s_0\big\|_{H^{k_1}} &\leq 2\big( \delta + \|\tsum_I \bar p_I^2 + \bar\Psi^2\|_{C^0}^{1/2} \big) t^A\big( \|\delta k\|_{H^{k_1}}^2 + \|\delta\s_0\|_{H^{k_1}}^2 \big)^{1/2}\\
        &\quad + Ct^A\big(\|\delta k\|_{C^{k_0}} + \|\delta\s_0\|_{C^{k_0}}\big)\\
        &\leq \frac{C_*}{t}\bbh(t) + Ct^{-1+\delta},
    \end{split}
    \]
    where we have used Lemma~\ref{estimate on the kasner relations} (with $CT^{2\delta} \leq \delta$). Therefore,
    \begin{equation} \label{eq: hamiltonian constraint trick estimate}
    t^{2A+1} \big|\big\langle e_J\gamma_{JKK}, t\Theta-1 \big\rangle_{H^{k_1}}\big| \leq \frac{C_* + Cr}{t}\bbh(t)^2 + Ct^{-1+\delta}.
    \end{equation}
    Now we turn our attention to the term of interest. Note that
    \[
    \begin{split}
        &\int_s^T t^{2(A+1)} \big\langle Ne_I\gamma_{JKK}, \delta k_{IJ} \big\rangle_{H^{k_1}} \,dt\\
        &\hspace{2cm}= \int_s^T t^{2(A+1)} \Big( \big\langle Ne_I\gamma_{JKK}, \delta k_{IJ} \big\rangle_{H^{k_1}} + \big\langle \gamma_{JKK},  Ne_I\delta k_{IJ} \big\rangle_{H^{k_1}} \Big)\,dt\\
        &\hspace{2cm}\quad - \int_s^T t^{2(A+1)} \big\langle \gamma_{JKK}, Ne_I\delta k_{IJ} - e_J\Theta \big\rangle_{H^{k_1}} \,dt\\
        &\hspace{2cm}\quad - \int_s^T t^{2A+1} \Big( \big\langle \gamma_{JKK}, e_J(t\Theta-1) \big\rangle_{H^{k_1}} + \big\langle e_J\gamma_{JKK}, t\Theta-1 \big\rangle_{H^{k_1}} \Big) \,dt\\
        &\hspace{2cm}\quad + \int_s^T t^{2A+1} \big\langle e_J\gamma_{JKK}, t\Theta-1 \big\rangle_{H^{k_1}} \,dt.
    \end{split}
    \]
    As a consequence of Lemmas~\ref{lemma: symmetric hyperbolic estimate} and \ref{lemma: divergence of the frame}, we see that the first and third integrals on the right-hand side can be estimated in absolute value by
    \[
    \int_s^T Ct^{-1+\delta} \,dt + C\int_{\side_{[s,T]}} t^{2A}\rho(t)\, \frac{t^{-1+4\delta}}{|df|} \,\mu_{\side_{[s,T]}}.
    \]
    This information, together with \eqref{eq: momentum constraint trick estimate} and \eqref{eq: hamiltonian constraint trick estimate}, yields the result.
\end{proof}

\begin{proposition} \label{prop: high order energy estimate}
    If $A$ is large enough, depending only on $\delta$ and $\lpar$, and $T$ and $r$ are small enough, then 
    \[
    \bbh(s) \leq \bbh(T) + CT^{\delta/2}
    \]
    for all $s \in [t_b,T]$.
\end{proposition}

\begin{proof}
    Applying Lemma~\ref{divergence theorem lemma} to $t^{2A}\rho(t)$, and putting Lemmas~\ref{lemma: high order estimate frame and dual frame}, \ref{lemma: high order estimate second fundamental form}, \ref{lemma: high order estimate mean curvature}, \ref{lemma: high order estimate lapse}, \ref{lemma: high order estimate lapse derivatives}, \ref{lemma: high order estimate gammas}, \ref{lemma: high order estimate phi time derivative} and \ref{lemma: high order estimate phi spatial derivatives} together, yields
    \[
    \begin{split}
        &\bbh(s)^2 + \int_{\side_{[s,T]}} t^{2A}\rho(t)\, \frac{t^{-1+3\delta}}{|df|} \,\mu_{\side_{[s,T]}}\\
        &\leq \int_s^T 2t^{2(A+1)} \Big( \big\langle Ne_K\gamma_{KIJ}, \delta k_{IJ} \big\rangle_{H^{k_1}(\Omega_t)} + \big\langle Ne_I\delta k_{JK}, \gamma_{IJK} \big\rangle_{H^{k_1}(\Omega_t)} \Big) \,dt\\
        &\quad + \int_s^T 2t^{2(A+1)} \Big( \big\langle -Ne_Ie_J(\ln N), \delta k_{IJ} \big\rangle_{H^{k_1}(\Omega_t)} + \big\langle -Ne_J\delta k_{JI}, e_I(\ln N) \big\rangle_{H^{k_1}(\Omega_t)} \Big) \,dt\\
        &\quad + \int_s^T 2t^{2A+1} \Big( \big\langle -\lpar e_Ie_I(\ln N), t\Theta-1 \big\rangle_{H^{k_1}(\Omega_t)} + \big\langle -\lpar e_I(t\Theta-1), e_I(\ln N) \big\rangle_{H^{k_1}(\Omega_t)} \Big) \,dt\\
        &\quad + \int_s^T 2t^{2(A+1)} \Big( \big\langle -Ne_Ie_I\s, \delta\s_0 \big\rangle_{H^{k_1}(\Omega_t)} + \big\langle -Ne_I\delta\s_0, e_I\s \big\rangle_{H^{k_1}(\Omega_t)} \Big) \,dt\\
        &\quad + \bbh(T)^2 + \int_s^T 2t^{2(A+1)} \big\langle Ne_I\gamma_{JKK}, \delta k_{IJ} \big\rangle_{H^{k_1}(\Omega_t)} + \frac{C_* + Cr - 2A}{t}\bbh(t)^2 + Ct^{-1+\delta} \,dt.
    \end{split}
    \]
    Estimating the first four integrals on the right-hand side with Lemmas~\ref{lemma: symmetric hyperbolic estimate} and \ref{lemma: divergence of the frame}, and using Lemma~\ref{lemma: constraint equations trick}, yields
    \[
    \begin{split}
        \bbh(s)^2 + \int_{\side_{[s,T]}} t^{2A}\rho(t)\, \frac{t^{-1+3\delta}}{|df|} \,\mu_{\side_{[s,T]}} &\leq \bbh(T)^2 + \int_s^T \frac{C_* + Cr - 2A}{t} \bbh(t)^2 + Ct^{-1+\delta} \,dt\\
        &\quad + C\int_{\side_{[s,T]}} t^{2A}\rho(t)\, \frac{t^{-1+4\delta}}{|df|} \,\mu_{\side_{[s,T]}}
    \end{split}
    \]
    (note that the integrals over $\side_{[s,T]}$ do not appear in the global setting). At this point, we choose $r$ small enough such that $Cr \leq 1$, and $T$ small enough such that $CT^\delta \leq 1$. Then the boundary integral on the right-hand side can be absorbed on the left-hand side of the inequality. Hence, the desired estimate holds if $2A \geq C_* + 1$. 
\end{proof}

\begin{proposition} \label{prop: low order energy estimate}
    If $\lpar \geq 3\delta$ and $r$ is small enough, then
    \[
    \bbl(s) \leq \bbl(T) + CT^{\delta/2}
    \]
    for all $s \in [t_b,T]$.
\end{proposition}

\begin{proof}
    Define
    \[
    \mfe := \lowen{e} + \lowen{\omega} + \lowen{\delta k} + \lowen{\Theta} + \lowen{N} + \lowen{\ve(\ln N)} + \lowen{\gamma} + \lowen{\delta\s_0} + \lowen{\ve\s}.
    \]
    Putting Lemmas~\ref{lemma: low order estimate frame and dual frame}, \ref{lemma: low order estimate second fundamental form}, \ref{lemma: low order estimate mean curvature}, \ref{lemma: low order estimate lapse}, \ref{lemma: low order estimate lapse derivatives}, \ref{lemma: low order estimate gammas}, \ref{lemma: low order estimate phi time derivative} and \ref{lemma: low order estimate phi spatial derivatives} together yields
    \[
    \begin{split}
        \mfe(s) &\leq \mfe(T) + \int_s^T \frac{6\delta-2\lpar}{t} \lowen{\Theta}(t) + Ct^{-1+\delta} \,dt\\
        &\quad + \int_s^T \frac{Cr-2\delta}{t}\Big( \lowen{e}(t) + \lowen{\omega}(t) + \lowen{\ve(\ln N)}(t) + \lowen{\gamma}(t) + \lowen{\ve\s}(t) \Big) \,dt.
    \end{split}
    \]
    Hence, if $\lpar \geq 3\delta$ and $r$ is small enough such that $Cr \leq 2\delta$, 
    \[
    \mfe(s) \leq \mfe(T) + CT^\delta.
    \]
    This proves the result, since there is a constant $C$ such that
    \[
    \tfrac{1}{C}\bbl(s) \leq \|\mfe(s)^{1/2}\|_{C^0(\overline{\Omega}_s)} \leq C\bbl(s)
    \]
    for all $s \in [t_b,T]$.
\end{proof}

\section{Higher order energy estimates} \label{sec: higher order estimates}

The purpose of this section is to upgrade the $C^{k_0}$ estimates obtained for the solutions of Theorem~\ref{thm: global existence} to estimates for derivatives of all orders. The strategy is similar to that of \cite[Section~2]{franco_complete_asymptotics_2026}. We start from assumptions that are weaker than the results of Theorem~\ref{thm: global existence}. For some $\delta, \varepsilon > 0$ and $T < 1$, suppose that we have a solution $(\Omega_{(0,T]},g,\s)$ to Einstein's equations, as described in Proposition~\ref{prop: reduced equations struct coeffs}, where $V$ is a $3\delta$-admissible potential and the mean curvature $\theta$ is always positive, such that the following estimates hold:
\begin{subequations} \label{eq: C1 assumptions}
    \begin{align}
    t^{1-4\delta}\ck[(\overline{\Omega}_t)]{1}{e} + t^{1-4\delta}\ck[(\overline{\Omega}_t)]{1}{\omega} + \ck[(\overline{\Omega}_t)]{1}{N} + t\ck[(\overline{\Omega}_t)]{1}{k} + t^{1-3\delta}\ck[(\overline{\Omega}_t)]{1}{\gamma} &\leq C,\\
    t^{-3\delta}\ck[(\overline{\Omega}_t)]{1}{t\Theta-1} + t^{1-4\delta}\ck[(\overline{\Omega}_t)]{1}{\ve(\ln N)} + \ck[(\overline{\Omega}_t)]{0}{1/(t\Theta)} + \ck[(\overline{\Omega}_t)]{0}{1/N} &\leq C,\\
    t\ck[(\overline{\Omega}_t)]{1}{e_0\s} + t^{1-3\delta}\ck[(\overline{\Omega}_t)]{1}{\ve\s} &\leq C
    \end{align}
\end{subequations}
and
\begin{equation} \label{eq: kasner relations estimate}
    \ck[(\overline{\Omega}_t)]{0}{t^2N^2( k_{IJ}k_{IJ} + e_0(\s)^2 ) - 1} \leq Ct^{2\delta},
\end{equation}
for all $t \in (0,T]$, where $C$ is a constant. The following is the main result of this section.

\begin{theorem} \label{thm: smooth estimates}
    Under the assumptions just described, if $2^{-1/4} \leq N \leq 2^{1/4}$ and $1 \leq \lpar \leq \frac{4}{3}$, then there are constants $C_\ell$ such that
    \begin{align*}
        t^{1-3\delta}\ck[(\overline{\Omega}_t)]{\ell}{e} + t^{1-3\delta}\ck[(\overline{\Omega}_t)]{\ell}{\omega} + \ck[(\overline{\Omega}_t)]{\ell}{N} + t\ck[(\overline{\Omega}_t)]{\ell}{k} + t^{1-2\delta}\ck[(\overline{\Omega}_t)]{\ell}{\gamma} &\leq C_\ell,\\
        t^{-2\delta}\ck[(\overline{\Omega}_t)]{\ell}{t\Theta-1} + t^{1-3\delta}\ck[(\overline{\Omega}_t)]{\ell}{\ve(\ln N)} + t\ck[(\overline{\Omega}_t)]{\ell}{e_0\s} + t^{1-2\delta}\ck[(\overline{\Omega}_t)]{\ell}{\ve\s} &\leq C_\ell
    \end{align*}
    for all $\ell$ and all $t \in (0,T]$. Moreover, 
    \begin{equation} \label{eq: estimates on time derivatives}
        \ck[(\overline{\Omega}_t)]{\ell}{\p_t(tk_{IJ})} + \ck[(\overline{\Omega}_t)]{\ell}{\p_t(te_0\s)} + \ck[(\overline{\Omega}_t)]{\ell}{\p_tN} \leq C_\ell t^{-1+\delta}.
    \end{equation}
    In particular, there are functions $\mfkr_{IJ}, \nr, \phir, \psir \in C^\infty(\overline{\Omega}_0)$ such that
    \begin{align*}
        \ck[(\overline{\Omega}_0)]{\ell}{tk_{IJ}(t) - \mfkr_{IJ}} + \ck[(\overline{\Omega}_0)]{\ell}{N(t) - \nr} + \ck[(\overline{\Omega}_0)]{\ell}{\Psi(t) - \psir} &\leq C_\ell t^{\delta},\\
        \ck[(\overline{\Omega}_0)]{\ell}{\Phi(t) - \phir} &\leq C_\ell t^{\delta/2}
    \end{align*}
    for all $\ell$ and all $t \in (0,T]$.
\end{theorem}

\begin{remark}
    The explicit bounds on $N$ and $\lpar$ in the statement of the theorem are chosen for convenience; see Lemma~\ref{lemma: differential energy estimate} below. The result would still hold for any upper bounds on $N$ and $1/N$, and any choice of $\lpar > 0$. 
\end{remark}

The proof of the theorem is to be found at the end of this section. Similarly as in \cite{franco_complete_asymptotics_2026}, it is convenient to introduce
\[
a_I := \gamma_{IKK}.
\]
As a consequence of \eqref{gamma equation global}, 
\begin{equation*} 
    e_0a_I = e_Kk_{IK} - e_I\theta - k_{IL}a_L + e_K(\ln N)k_{IK} - e_I(\ln N)\theta.
\end{equation*}
Hence, by using the momentum constraint \eqref{momentum constraint with structure coeffs},
\begin{equation} \label{eq: a equation}
    e_0a_I = \gamma_{JIK} k_{JK} + e_0(\s)e_I(\s) + e_K(\ln N)k_{IK} - e_I(\ln N)\theta.
\end{equation}
The reason for introducing $a_I$ is that it allows for the power of $t$ in Proposition~\ref{prop: higher order energy estimate} below to be independent of the dimension (recall that the only part of the proof where the restriction to $3+1$ dimensions is relevant is Lemma~\ref{cyclic sum lemma}). 

Now we introduce the energy densities
\begin{align*}
    \dens{e} &:= \sum_{I,i} \sum_{|\I| \leq \ell} t^{-2\delta} (\pI e_I^i)^2, & \dens{\omega} &:= \sum_{I,i} \sum_{|\I| \leq \ell} t^{-2\delta} (\pI \omega_i^I)^2,\\
    \dens{k} &:= \sum_{I,J} \sum_{|\I| \leq \ell} (\pI k_{IJ})^2, & \dens{\gamma} &:= \sum_{I,J,K} \sum_{|\I| \leq \ell} \tfrac{1}{2} (\pI \gamma_{IJK})^2,\\
    \dens{N} &:= \sum_{|\I| \leq \ell} t^{-2} (\pI N)^2, & \dens{\ve(\ln N)} &:= \sum_I \sum_{|\I| \leq \ell} \big(\pI e_I(\ln N)\big)^2\\
    \dens{\Theta} &:= \sum_{|\I| \leq \ell} \frac{\lpar}{N^2} t^{-2} \big(\pI(t\Theta-1)\big)^2,\\
    \dens{e_0\s} &:= \sum_{|\I| \leq \ell} (\pI e_0\s)^2, & \dens{\ve\s} &:= \sum_I \sum_{|\I| \leq \ell} (\pI e_I\s)^2,\\
    \dens{\s} &:= \sum_{|\I| \leq \ell} t^{-2+2\delta} (\pI \s)^2, & \dens{a} &:= \sum_I \sum_{|\I| \leq \ell} (\pI a_I)^2,
\end{align*}
for $\ell$ a positive integer, the total energy density
\[
\begin{split}
    \rho_\ell &:= \dens{e} + \dens{\omega} + \dens{k} + \dens{\gamma} + \dens{N} + \dens{\ve(\ln N)}\\
    &\quad + \dens{\Theta} + \dens{e_0\s} + \dens{\ve\s} + \dens{\s} + \dens{a},
\end{split}
\]
and the energy
\[
\ce_\ell(t) := \int_{\Omega_t} \rho_\ell(t) \,dx.
\]
For the remainder of this section, we work under the assumptions of Theorem~\ref{thm: smooth estimates}. Now we derive some consequences of \eqref{eq: C1 assumptions} and \eqref{eq: kasner relations estimate}.

\begin{remark}
    Similarly as in Section~\ref{sec: energy estimates}, we usually omit mention of the domain for the $C^\ell$ and $H^\ell$ norms whenever there is no danger of confusion.
\end{remark}

\begin{lemma} \label{lemma: basic blow up of energy}
    There is a $t_0 \in (0,T]$ and a constant $C$ such that
    \[
    Ct^{-2} \leq \ce_\ell(t)
    \]
    for all $t \leq t_0$.
\end{lemma}

\begin{remark}
    From now on, whenever we write $t_0$, we mean the number that appears in this lemma.
\end{remark}

\begin{proof}
    As a consequence of \eqref{eq: kasner relations estimate}, there is $t_0 \in (0,T]$ such that
    \[
    \frac{1}{2} \leq t^2N^2(k_{IJ}k_{IJ} + e_0(\s)^2) \leq Ct^2(k_{IJ}k_{IJ} + e_0(\s)^2).
    \]
    Integrating this inequality over $\Omega_t$, we obtain the result.
\end{proof}

\begin{lemma} \label{lemma: potential estimates 2}
    We have
    \[
    \begin{split}
        \sob{\ell}{V\circ\s} + \sob{\ell}{V'\circ\s} &\leq C_\ell t^{-2+5\delta}( 1 + \sob{\ell}{\s} ),\\
        \ck{1}{V\circ\s} + \ck{1}{V'\circ\s} &\leq C_\ell t^{-2+5\delta}
    \end{split}
    \]
    for all $t \leq T$. In particular,
    \[
    \sob{\ell}{V\circ\s} + \sob{\ell}{V'\circ\s} \leq C_\ell t^{-1+4\delta} \ce_\ell^{1/2}
    \]
    for all $t \leq t_0$.
\end{lemma}

\begin{proof}
    As a consequence of \eqref{eq: kasner relations estimate}, 
    \[
    t\ck{0}{\p_t\s} \leq 1 + Ct^\delta.
    \]
    This together with \eqref{eq: C1 assumptions} implies that
    \[
    \ck{0}{\s} \leq -\ln t + C, \quad \ck{1}{\s} \leq C\langle \ln t \rangle.
    \]
    Proceeding similarly as in the proof of Lemma~\ref{lemma: potential estimates}, we obtain the first and second inequalities. The third inequality now follows from the first one, the definition of $\ce_\ell$ and Lemma~\ref{lemma: basic blow up of energy}.
\end{proof}

\subsection{The subcritical and critical terms}

We use schematic notation in a similar way as in Section~\ref{sec: energy estimates}; see Remark~\ref{rmk: schematic notation}. This time there is also a notion of subcritical and critical terms. For $\alpha \in \R$, let $\ct$ be a term consisting of products of $t^\alpha, e, \omega, k, \gamma, N, 1/N, \ve(\ln N), \Theta, t\Theta-1, e_0\s, \ve\s, V\circ\s, V'\circ\s$ or $a$, potentially with spatial derivatives applied to them. Then $\ct$ is called \emph{subcritical} if it satisfies an estimate of the form
\begin{equation}
    \sob[(\Omega_t)]{k}{\ct} \leq C_\ell t^{-1+\delta} \ce_\ell^{1/2}(t)
\end{equation}
for an appropriate value of $k$, to be specified in each case. The general idea for these terms is to apply Moser estimates, then use \eqref{eq: C1 assumptions}, and estimate the $H^k$ norms in terms of the energy. As before, we use the notation $\subc$ to denote sums of subcritical terms. Moreover, $\ct$ is called \emph{critical} if it satisfies an estimate of the form
\begin{equation}
    \|\ct\|_{L^2(\Omega_t)} \leq \frac{C_\ell}{t} \ce_{\ell-1}^{1/2}(t) + C_\ell t^{-1+\delta} \ce_\ell^{1/2}(t),
\end{equation}
We now give an example to illustrate the idea behind the critical terms. Consider the term $ t^{-\delta} \pI[1]\p N * \pI[2]k * \pI[3]e$ with $|\I_1| + |\I_2| + |\I_3| \leq \ell-1$. By Moser estimates, leaving the $\p$ derivative on $N$ fixed, we obtain
\[
\begin{split}
    t^{-\delta}\|\pI[1]\p N * \pI[2]k * \pI[3]e\|_{L^2(\Omega_t)} &\leq C_\ell t^{-\delta} \big( \ck[(\overline{\Omega}_t)]{1}{N} \ck[(\overline{\Omega}_t)]{0}{k} \sob[(\Omega_t)]{\ell-1}{e}\\
    &\quad + \ck[(\overline{\Omega}_t)]{1}{N} \sob[(\Omega_t)]{\ell-1}{k} \ck[(\overline{\Omega}_t)]{0}{e}\\
    &\quad + \sob[(\Omega_t)]{\ell}{N} \ck[(\overline{\Omega}_t)]{0}{k} \ck[(\overline{\Omega}_t)]{0}{e} \big)\\
    &\leq \frac{C_\ell}{t} \ce_{\ell-1}^{1/2}(t) + C_\ell t^{-1+\delta} \ce_\ell^{1/2}(t).
\end{split}
\]
The idea in general is the same for these types of terms. As above, whenever we say a term is critical, and there are single $\p$ derivatives applied to some of the factors, we mean to imply that those derivatives should be left fixed when applying Moser estimates.

\subsection{The frame and the dual frame}

\begin{lemma} \label{lemma: frame higher order estimate}
    We have
    \begin{align*}
        \int_{\Omega_t} -\p_t \dens{e} \,dx &\leq \frac{2(1+\delta)}{t} \int_{\Omega_t} \dens{e} \,dx + \frac{C_\ell}{t} \ce_{\ell-1}^{1/2} \ce_\ell^{1/2} + C_\ell t^{-1+\delta} \ce_\ell,\\
        \int_{\Omega_t} -\p_t \dens{\omega} \,dx &\leq \frac{2(1+\delta)}{t} \int_{\Omega_t} \dens{\omega} \,dx + \frac{C_\ell}{t} \ce_{\ell-1}^{1/2} \ce_\ell^{1/2} + C_\ell t^{-1+\delta} \ce_\ell
    \end{align*}
    for all $t \leq t_0$.
\end{lemma}

\begin{proof}
    We only prove it for $e$. Write \eqref{frame equation global} as
    \begin{equation} \label{eq: frame equation higher order}
        -\p_te_I^i = Nk_{IJ}e_J^i.
    \end{equation}
    If $|\I| \leq \ell$, note that $t^{-\delta}\pI(Nk_{IJ}e_J^i)$ can be written as a sum of terms of the form
    \[
    t^{-\delta}Nk_{IJ}\pI e_J^i, \qquad t^{-\delta}\pI[1]\p N * \pI[2]k * \pI[3]e, \qquad t^{-\delta}\pI[1] N * \pI[2]\p k * \pI[3]e,  
    \]
    with $|\I_1| + |\I_2| + |\I_3| \leq \ell-1$. The second and third types of terms are critical. For the first term, by applying the Cauchy-Schwarz inequality to the sum over $I$ and $J$ only,
    \[
    \sum_{i} \sum_{|\I| \leq \ell} t^{-2\delta} Nk_{IJ} (\pI e_J^i) (\pI e_I^i) \leq \sum_{i} \sum_{|\I| \leq \ell} t^{-2\delta} (  N^2k_{IJ}k_{IJ} )^{1/2} \sum_I (\pI e_I^i)^2 \leq \frac{1+Ct^\delta}{t} \dens{e},
    \]
    where we have used \eqref{eq: kasner relations estimate}. Now we can write
    \[
    -\p_t \dens{e} = \frac{2\delta}{t} \dens{e} + 2\sum_{i} \sum_{|\I| \leq \ell} t^{-\delta} \pI( -t^{-\delta}\p_te_I^i ) \pI e_I^i,
    \]
    and we obtain the result by using \eqref{eq: frame equation higher order} and integrating over $\Omega_t$.
\end{proof}

\subsection{The second fundamental form and the mean curvature}

\begin{lemma} \label{lemma: k higher order estimate}
    We have
    \[
    \begin{split}
        \int_{\Omega_t} -\p_t \dens{k} \,dx &\leq 2\sobprod[(\Omega_t)]{\ell}{Ne_K\gamma_{KIJ} + Ne_Ia_J - Ne_Ie_J(\ln N)}{k_{IJ}} + \frac{2}{t} \int_{\Omega_t} \dens{k} \,dx\\
        &\quad + \int_{\Omega_t} 2k_{IJ} \sum_{1 \leq |\I| \leq \ell} (\pI\Theta)(\pI k_{IJ}) \,dx + \frac{C_\ell}{t} \ce_{\ell-1}^{1/2} \ce_\ell^{1/2} + C_\ell t^{-1+\delta} \ce_\ell
    \end{split}
    \]
    for all $t \leq t_0$.
\end{lemma}

\begin{proof}
    Equation~\eqref{k equation global} can be written as
    \begin{equation} \label{eq: k equation higher order}
        -\p_tk_{IJ} = Ne_K\gamma_{K(IJ)} + Ne_{(I}a_{J)} - Ne_{(I}e_{J)}(\ln N) + \Theta k_{IJ} + \subc,
    \end{equation}
    where
    \[
    \subc = N * \gamma * \gamma - Ne_I(\ln N)e_J(\ln N) + N * \gamma * \ve(\ln N) - Ne_I(\s)e_J(\s) - N(V\circ\s)\delta_{IJ}, 
    \]
    with $k = \ell$. If $|\I| \leq \ell$, note that $\pI(\Theta k_{IJ})$ can be written as a sum of terms of the form
    \[
    \frac{1}{t}\pI k_{IJ}, \qquad \frac{1}{t} (t\Theta-1)\pI k_{IJ}, \qquad k_{IJ} \pI\Theta, \qquad \pI[1] \p\Theta \cdot \pI[2]\p k_{IJ},
    \]
    where $|\I_1| + |\I_2| \leq \ell-2$ and the third term can be omitted if $\I = 0$. The first and third terms we keep as is. The second term can be estimated in $L^2$ as the subcritical terms, if we estimate $t\Theta-1$ in $C^0$. The fourth term is critical. The result follows by applying $-\p_t$ to $\dens{k}$, using \eqref{eq: k equation higher order} and integrating over $\Omega_t$. 
\end{proof}

\begin{lemma} \label{lemma: mean curvature higher order estimates}
    We have
    \[
    \begin{split}
        \int_{\Omega_t} -\p_t \dens{\Theta} \,dx &\leq \frac{2\lpar}{t} \sobprod[(\Omega_t)]{\ell}{-e_Ie_I(\ln N)}{t\Theta-1} - \frac{2(\lpar + 1)}{t} \int_{\Omega_t} \dens{\Theta} \,dx\\
        &\quad + \int_{\Omega_t} \frac{4\lpar}{tN}\big( k_{IJ}k_{IJ} + e_0(\s)^2 \big) \sum_{1 \leq |\I| \leq \ell} (\pI N)\pI(t\Theta-1) \,dx\\
        &\quad + \int_{\Omega_t} \frac{4\lpar}{t} \sum_{1 \leq |\I| \leq \ell} \big( k_{IJ}\pI k_{IJ} + e_0(\s) \pI e_0(\s) \big) \pI(t\Theta-1) \,dx\\
        &\quad + \frac{C_\ell}{t} \ce_{\ell-1}^{1/2} \ce_\ell^{1/2}  + C_\ell t^{-1+\delta} \ce_\ell
    \end{split}
    \]
    for all $t \leq t_0$.
\end{lemma}

\begin{proof}
    We can write \eqref{mean curvature high order equation} as
    \begin{equation} \label{eq: mean curvature equation higher order}
    \begin{split}
        -\frac{1}{t}\p_t(t\Theta-1) &= -N^2e_Ie_I(\ln N) - \frac{\lpar + 2}{t^2}(t\Theta-1) + N^2\big(k_{IJ}k_{IJ} + e_0(\s)^2\big) - \frac{1}{t^2} + \subc, 
    \end{split}
    \end{equation}
    where
    \[
    \subc = -N^2e_I(\ln N)e_I(\ln N) + N^2 * \gamma * \ve(\ln N) - N^2V\circ\s - \frac{\lpar+1}{t^2}(t\Theta-1)^2,
    \]
    with $k = \ell$. Let $|\I| \leq \ell$. Note that $\pI\big(-N^2e_Ie_I(\ln N)\big)$ can be written as a sum of terms of the form
    \[
    -N^2\pI e_Ie_I(\ln N), \qquad \pI[1] N * \pI[2] \p N * \pI[3]e * \pI[4] \p\ve(\ln N),
    \]
    where $|\I_1| + |\I_2| + |\I_3| + |\I_4| \leq \ell-1$. The first term we keep. The second term is subcritical with $k = 0$. Next, consider $\pI\big(N^2(k_{IJ}k_{IJ} + e_0(\s)^2) - \frac{1}{t^2}\big)$. If $\I = 0$, it can be estimated as the subcritical terms as a consequence of \eqref{eq: kasner relations estimate} and Lemma~\ref{lemma: basic blow up of energy}. If $\I \neq 0$, then it can be written as a sum of 
    \[
    2N(\pI N)\big(k_{IJ}k_{IJ} + e_0(\s)^2\big), \qquad 2N^2\big( k_{IJ}\pI k_{IJ} + e_0(\s) \pI e_0(\s) \big),
    \]
    and terms where at least two derivatives hit two different factors. The first two terms we keep. The third type of term is critical.

    Now note that
    \[
    -\p_t \dens{\Theta} = 2(e_0N) \dens{\Theta} + \frac{2}{t}\dens{\Theta} + \sum_{|\I| \leq \ell} \frac{2\lpar}{tN^2} \pI\big( -\tfrac{1}{t}\p_t(t\Theta-1)  \big) \pI(t\Theta-1).
    \]
    The result follows by using \eqref{lapse equation global}, \eqref{eq: mean curvature equation higher order} and integrating over $\Omega_t$.
\end{proof}

\subsection{The lapse}

\begin{lemma} \label{lemma: lapse higher order estimates}
    We have
    \[
    \begin{split}
        \int_{\Omega_t} -\p_t\dens{N} \,dx &\leq \frac{\lpar + 3}{t} \int_{\Omega_t} \dens{N} \,dx + \frac{1}{t} \Big(1 + \frac{1}{\lpar}\Big)\ck{0}{N}^4 \int_{\Omega_t} \dens{\Theta} \,dx\\
        &\quad + \frac{C_\ell}{t} \ce_{\ell-1}^{1/2} \ce_\ell^{1/2} + C_\ell t^{-1+\delta} \ce_\ell 
    \end{split}
    \]
    for all $t \leq t_0$.
\end{lemma}

\begin{proof}
    we write \eqref{lapse equation global} as
    \begin{equation} \label{eq: lapse equation higher order}
        -\frac{1}{t}\p_t N = -\frac{\lpar + 1}{t^2} N(t\Theta-1).
    \end{equation}
    If $|\I| \leq \ell$, then $\pI$ of the right-hand side of this equation can be written as a sum of terms of the form
    \[
    -\frac{\lpar+1}{t^2} N \pI(t\Theta-1), \qquad \frac{1}{t^2} \pI[1] \p N * \pI[2](t\Theta-1).
    \]
    The second type of term is critical. For the first term, by Young's inequality,
    \[
    \begin{split}
        -\frac{2(\lpar+1)}{t^3}N \sum_{|\I| \leq \ell} \pI(t\Theta-1) \pI N &\leq \frac{2(\lpar+1)N^2}{t\lpar^{1/2}} \dens{\Theta}^{1/2} \dens{N}^{1/2}\\
        &\leq \frac{\lpar+1}{t} \Big( \frac{N^4}{\lpar} \dens{\Theta} + \dens{N} \Big).
    \end{split}
    \]
    The result follows by applying $-\p_t$ to $\dens{N}$, using \eqref{eq: lapse equation higher order} and integrating over $\Omega_t$.
\end{proof}

\begin{lemma} \label{lemma: lapse derivatives higher order estimates}
    We have
    \[
    \begin{split}
        \int_{\Omega_t} -\p_t \dens{\ve(\ln N)} \,dx &\leq 2\sobprod[(\Omega_t)]{\ell}{-\lpar e_I\Theta - Ne_Jk_{JI}}{e_I(\ln N)}\\
        &\quad + 2\int_{\Omega_t} Nk_{JI} \sum_{|\I| \leq \ell} (\pI a_J)\big(\pI e_I(\ln N)\big) \,dx\\
        &\quad + 2\int_{\Omega_t} Nk_{JK}\sum_{|\I| \leq \ell} (\pI \gamma_{JIK})\big(\pI e_I(\ln N)\big) \,dx\\
        &\quad + 2\int_{\Omega_t} Ne_0(\s)\sum_{|\I| \leq \ell} (\pI e_I\s)\big(\pI e_I(\ln N)\big) \,dx\\
        &\quad + \frac{C_\ell}{t} \ce_{\ell-1}^{1/2} \ce_\ell^{1/2} + C_\ell t^{-1+\delta} \ce_\ell
    \end{split}
    \]
    for all $t \leq t_0$.
\end{lemma}

\begin{proof}
    Equation~\eqref{lapse derivatives high order equation} can be written as
    \begin{equation} \label{eq: lapse derivatives equation higher order}
        \begin{split}
            -\p_te_I(\ln N) &= -\lpar e_I\Theta - Ne_Jk_{JI} + Na_Jk_{JI} + N\gamma_{JIK}k_{JK}\\
            &\quad + Ne_0(\s)e_I(\s) - \Theta e_I(\ln N) + Nk_{IJ}e_J(\ln N).
        \end{split}
    \end{equation}
    If $|\I| \leq \ell$, then $\pI(N a_J k_{JI})$ can be written as a sum of terms of the form 
    \[
    N(\pI a_J)k_{JI}, \qquad \pI[1] \p N * \pI[2] a * \pI[3] k, \qquad \pI[1] N * \pI[2] a * \pI[3] \p k,
    \]
    where $|\I_1| + |\I_2| + |\I_3| \leq \ell-1$. The first term we keep. The second and third are critical. Similar observations hold for the last four terms on the right-hand side of \eqref{eq: lapse derivatives equation higher order}. Hence, up to critical terms, we are left with
    \[
    Nk_{JK}(\pI\gamma_{JIK}), \qquad Ne_0(\s)(\pI e_I\s), \qquad -\Theta\big(\pI e_I(\ln N)\big), \qquad Nk_{IJ}\big(\pI e_J(\ln N)\big).
    \]
    Putting the third and fourth terms together, we have
    \[
    \sum_{|\I| \leq \ell} \Big( Nk_{IJ} \big(\pI e_J(\ln N)\big) - \Theta \big(\pI e_I(\ln N)\big) \Big)\big(\pI e_I(\ln N)\big) \leq \frac{1 - t\Theta+Ct^\delta}{t}\dens{\ve(\ln N)},
    \]
    by the Cauchy-Schwarz inequality on the sum over $I$ and $J$ only and \eqref{eq: kasner relations estimate}. So that after integrating, the total contribution of these terms is subcritical. The result now follows by applying $-\p_t$ to $\dens{\ve(\ln N)}$, using \eqref{eq: lapse derivatives equation higher order} and integrating over $\Omega_t$. 
\end{proof}

\subsection{The structure coefficients}

\begin{lemma} \label{lemma: gammas higher order estimate}
    We have
    \[
    \begin{split}
        \int_{\Omega_t} -\p_t \dens{\gamma} \,dx &\leq 2\sobprod[(\Omega_t)]{\ell}{Ne_I k_{JK}}{\gamma_{IJK}} + \frac{6}{t}\int_{\Omega_t} \dens{\gamma} \,dx\\
        &\quad + 2\int_{\Omega_t} Nk_{JK} \sum_{|\I| \leq \ell} \big(\pI e_I(\ln N)\big)(\pI\gamma_{IJK}) \,dx + \frac{C_\ell}{t} \ce_{\ell-1}^{1/2} \ce_\ell^{1/2} + C_\ell t^{-1+\delta} \ce_\ell
    \end{split}
    \]
    for all $t \leq t_0$.
\end{lemma}

\begin{proof}
    Equation~\eqref{gamma equation global} can be written as
    \begin{equation} \label{eq: gamma equation higher order}
        \begin{split}
            -\p_t \gamma_{IJK} &= -Ne_Jk_{IK} + Ne_Ik_{JK} - Nk_{KL}\gamma_{IJL} - Nk_{IL}\gamma_{JLK}\\
            &\quad - Nk_{JL}\gamma_{LIK} - Ne_J(\ln N)k_{IK} + Ne_I(\ln N)k_{JK}.
        \end{split}
    \end{equation}
    If $|\I| \leq \ell$, after applying $\pI$ to the right-hand side of the equation, the last five terms can be written as the sum of 
    \begin{gather*}
        - Nk_{KL}(\pI\gamma_{IJL}), \quad - Nk_{IL}(\pI\gamma_{JLK}), \quad - Nk_{JL}(\pI\gamma_{LIK}),\\
        - Nk_{IK}\big(\pI e_J(\ln N)\big), \quad Nk_{JK}\big(\pI e_I(\ln N)\big),
    \end{gather*}
    plus critical terms. The last two terms we keep. For the first three terms, note that
    \[
    \begin{split}
        &- Nk_{KL}\sum_{|\I| \leq \ell}(\pI\gamma_{IJL})(\pI\gamma_{IJK})\\
        &- Nk_{IL}\sum_{|\I| \leq \ell}(\pI\gamma_{JLK})(\pI\gamma_{IJK})\\
        &- Nk_{JL}\sum_{|\I| \leq \ell}(\pI\gamma_{LIK})(\pI\gamma_{IJK}) \leq \frac{3+Ct^\delta}{t} \sum_{I,J,K} \sum_{|\I| \leq \ell}(\pI\gamma_{IJK})^2
    \end{split}
    \]
    by \eqref{eq: kasner relations estimate}. Now the result follows by applying $-\p_t$ to $\dens{\gamma}$, using \eqref{eq: gamma equation higher order}, integrating over $\Omega_t$ and using the antisymmetry of $\gamma_{IJK}$. 
\end{proof}

\begin{lemma} \label{lemma: a higher order estimate}
    We have
    \[
    \begin{split}
        \int_{\Omega_t} -\p_t \dens{a} \,dx &\leq -2\int_{\Omega_t} Nk_{JK} \sum_{|\I| \leq \ell} (\pI\gamma_{JIK})(\pI a_I) \,dx\\
        &\quad - 2\int_{\Omega_t} Ne_0(\s) \sum_{|\I| \leq \ell} (\pI e_I\s)(\pI a_I) \,dx\\
        &\quad -2\int_{\Omega_t} Nk_{IK} \sum_{|\I| \leq \ell} \big(\pI e_K(\ln N)\big)(\pI a_I) \,dx\\
        &\quad + \frac{1}{t} \int_{\Omega_t} \dens{\ve(\ln N)} + \dens{a} \,dx + \frac{C_\ell}{t} \ce_{\ell-1}^{1/2} \ce_\ell^{1/2} + C_\ell t^{-1+\delta} \ce_\ell
    \end{split}
    \]
    for all $t \leq t_0$.
\end{lemma}

\begin{proof}
    We write \eqref{eq: a equation} as
    \begin{equation} \label{eq: a equation higher order}
        -\p_t a_I = -N\gamma_{JIK}k_{JK} - Ne_0(\s)e_I(\s) - Nk_{IK}e_K(\ln N) + \Theta e_I(\ln N).
    \end{equation}
    If $|\I| \leq \ell$, then $\pI$ of the right-hand side can be written as the sum of
    \[
    -Nk_{JK}(\pI\gamma_{JIK}), \quad - Ne_0(\s)(\pI e_I\s), \quad - Nk_{IK}\big(\pI e_K(\ln N)\big), \quad \Theta \big(\pI e_I(\ln N)\big),
    \]
    plus critical terms. The last term we write as
    \[
    \frac{1}{t} \pI e_I(\ln N) + \frac{1}{t}(t\Theta-1)\big(\pI e_I(\ln N)\big).
    \]
    Note that the second term can be estimated in $L^2$ as a subcritical term. The result follows by differentiating $\dens{a}$ and using \eqref{eq: a equation higher order}. 
\end{proof}

\subsection{The scalar field}

\begin{lemma} \label{lemma: phi0 higher order estimate}
    We have
    \[
    \begin{split}
        \int_{\Omega_t} -\p_t \dens{e_0\s} \,dx &\leq 2\sobprod[(\Omega_t)]{\ell}{-Ne_Ie_I\s}{e_0\s} + \frac{2}{t} \int_{\Omega_t} \dens{e_0\s} \,dx\\
        &\quad + 2 \int_{\Omega_t} e_0(\s) \sum_{1 \leq |\I| \leq \ell} (\pI\Theta)(\pI e_0\s) \,dx + \frac{C_\ell}{t} \ce_{\ell-1}^{1/2} \ce_\ell^{1/2} + C_\ell t^{-1+\delta} \ce_\ell
    \end{split}
    \]
    for all $t \leq t_0$.
\end{lemma}

\begin{proof}
    We write \eqref{eq: phi equation global} as
    \begin{equation} \label{eq: phi equation higher order}
        -\p_te_0\s = -Ne_Ie_I\s + \Theta e_0\s + \subc,
    \end{equation}
    where
    \[
    \subc = Na_Ie_I\s - Ne_I(\s)e_I(\ln N) + NV'\circ\s
    \]
    with $k = \ell$. If $|\I| \leq \ell$, we can write $\pI(\Theta e_0\s)$ as the sum of
    \[
    \frac{1}{t} \pI e_0\s, \qquad \frac{t\Theta-1}{t} \pI e_0\s, \qquad e_0(\s) \pI\Theta,
    \]
    plus critical terms, where the third term is omitted if $\I = 0$. The result follows by differentiating $\dens{e_0\s}$ and using \eqref{eq: phi equation higher order}.
\end{proof}

\begin{lemma} \label{lemma: phi derivatives higher order estimate}
    We have
    \[
    \begin{split}
        \int_{\Omega_t} -\p_t \dens{\ve\s} \,dx &\leq 2\sobprod[(\Omega_t)]{\ell}{-Ne_Ie_0\s}{e_I\s} + \frac{2}{t} \int_{\Omega_t} \dens{\ve\s} \,dx\\
        &\quad - 2 \int_{\Omega_t} Ne_0(\s) \sum_{|\I| \leq \ell} \big(\pI e_I(\ln N)\big)(\pI e_I\s) \,dx + \frac{C_\ell}{t} \ce_{\ell-1}^{1/2} \ce_\ell^{1/2} + C_\ell t^{-1+\delta} \ce_\ell
    \end{split}
    \]
    for all $t \leq t_0$.
\end{lemma}

\begin{proof}
    Write \eqref{eq: phi derivatives equation global} as
    \begin{equation} \label{eq: phi derivatives equation higher order}
        -\p_te_I\s = -Ne_Ie_0\s + Nk_{IJ}e_J\s - Ne_I(\ln N) e_0\s.
    \end{equation}
    If $|\I| \leq \ell$, then $\pI$ of the last two terms can be written as the sum of 
    \[
    Nk_{IJ} \pI e_J\s, \qquad -Ne_0(\s) \pI e_I(\ln N),
    \]
    plus critical terms. After differentiating $\dens{\ve\s}$ and using \eqref{eq: phi derivatives equation higher order}, the contribution from the first term is estimated in the usual way, by the Cauchy-Schwarz inequality and \eqref{eq: kasner relations estimate}. The result follows. 
\end{proof}

\begin{lemma} \label{lemma: phi higher order estimate}
    We have
    \[
    \int_{\Omega_t} -\p_t \dens{\s} \,dx \leq \frac{2-2\delta}{t} \int_{\Omega_t} \dens{\s} \,dx + C_\ell t^{-1+\delta} \ce_\ell 
    \]
    for all $t \leq t_0$.
\end{lemma}

\begin{proof}
    This is a direct consequence of the definition of $\dens{\s}$. Recall that $\p_t\s = Ne_0\s$.
\end{proof}

\subsection{The energy estimate}

Similarly as in Lemma~\ref{lemma: constraint equations trick}, we need to control some terms by using the constraint equations.

\begin{lemma} \label{lemma: contstraint trick higher order}
    We have
    \begin{align*}
        \begin{split}
            2\sobprod[(\Omega_t)]{\ell}{a_J}{e_J\Theta - Ne_Ik_{IJ}} &\leq - 2 \int_{\Omega_t} Nk_{IK}\sum_{|\I| \leq \ell}(\pI a_J)(\pI \gamma_{IJK}) \,dx\\
            &\quad - 2 \int_{\Omega_t} Ne_0(\s)\sum_{|\I| \leq \ell}(\pI a_J)(\pI e_J\s) \,dx\\
            &\quad + \frac{3}{t} \int_{\Omega_t} \dens{a} \,dx + \frac{1}{t} \int_{\Omega_t} \dens{\ve(\ln N)} \,dx\\
            &\quad + \frac{C_\ell}{t} \ce_\ell^{1/2} \ce_\ell^{1/2} + C_\ell t^{-1+\delta} \ce_\ell,
        \end{split}\\
        \begin{split}
            \frac{2}{t} \sobprod[(\Omega_t)]{\ell}{e_Ja_J}{t\Theta-1} &\leq \int_{\Omega_t} \frac{2}{tN}\big(k_{IJ}k_{IJ} + e_0(\s)^2\big)\sum_{1 \leq |\I| \leq \ell}(\pI N) \pI(t\Theta-1) \,dx\\
            &\quad + \int_{\Omega_t} \frac{2}{t}\sum_{1 \leq |\I| \leq \ell}\big(k_{IJ}\pI k_{IJ} + e_0(\s)\pI e_0(\s)\big)\pI(t\Theta-1) \,dx\\
            &\quad - \frac{2}{\lpar t} \int_{\Omega_t} \dens{\Theta} \,dx + \frac{C_\ell}{t} \ce_\ell^{1/2} \ce_\ell^{1/2} + C_\ell t^{-1+\delta} \ce_\ell
        \end{split}
    \end{align*}
    for all $t \leq t_0$.
\end{lemma}

\begin{proof}
    For the first inequality, note that the momentum constraint \eqref{momentum constraint with structure coeffs} implies that
    \[
    e_J\Theta - Ne_Ik_{IJ} = -Na_Ik_{IJ} - N\gamma_{IJK}k_{IK} - Ne_0(\s)e_J(\s) + \Theta e_J(\ln N).
    \]
    If $|\I| \leq \ell$, then up to critical terms, we can write $\pI$ of the right-hand side as the sum of
    \[
    -Nk_{IJ}\pI a_I, \qquad -Nk_{IK}\pI \gamma_{IJK}, \qquad -Ne_0(\s) \pI e_J(\s), \qquad \Theta \pI e_J(\ln N).
    \]
    This proves the first inequality. For the second inequality, by the Hamiltonian constraint \eqref{hamiltonian constraint with structure coeffs},
    \begin{equation} \label{eq: hamiltonian constraint ready for the trick}
        2e_Ja_J = \frac{1}{t^2N^2} \big(t^2N^2\big(k_{IJ}k_{IJ} + e_0(\s)^2\big) - 1\big) -\frac{2}{t^2N^2}(t\Theta-1) + \subc, 
    \end{equation}
    where 
    \[
    \subc = \gamma * \gamma + e_I(\s)e_I(\s) + 2V\circ\s - \frac{1}{t^2N^2}(t\Theta-1)^2
    \]
    with $k = \ell$. If $1 \leq |\I| \leq \ell$, then $\pI$ of the first term on the right-hand side of \eqref{eq: hamiltonian constraint ready for the trick} can be written as the sum of 
    \[
    \frac{2}{N}\big(k_{IJ}k_{IJ} + e_0(\s)^2\big) \pI N, \qquad 2\big(k_{IJ} \pI k_{IJ} + e_0(\s) \pI e_0(\s)\big),
    \]
    plus critical terms. If $\I = 0$, then it is subcritical with $k = 0$, by \eqref{eq: kasner relations estimate} and Lemma~\ref{lemma: basic blow up of energy}. Moving on, if $|\I| \leq \ell$, then $\pI$ of the second term on the right-hand side of \eqref{eq: hamiltonian constraint ready for the trick} can be written as
    \[
    -\frac{2}{t^2N^2}\pI(t\Theta-1),
    \]
    plus critical terms. Putting these observations together, we conclude that the second inequality holds.
\end{proof}

\begin{lemma} \label{lemma: differential energy estimate}
    If $2^{-1/4} \leq N \leq 2^{1/4}$ and $1 \leq \lpar \leq \frac{4}{3}$, then
    \[
    \int_{\Omega_t} -\p_t \rho_\ell \,dx \leq \frac{8}{t} \ce_\ell + \frac{C_\ell}{t} \ce_{\ell-1}^{1/2} \ce_\ell^{1/2} + C_\ell t^{-1+\delta} \ce_\ell + \cf,
    \]
    where
    \[
    \begin{split}
        \cf &= 2\sobprod[(\Omega_t)]{\ell}{Ne_K\gamma_{KIJ}}{k_{IJ}} + 2\sobprod[(\Omega_t)]{\ell}{Ne_Ik_{JK}}{\gamma_{IJK}}\\
        &\quad + 2\sobprod[(\Omega_t)]{\ell}{Ne_Ia_J}{k_{IJ}} + 2\sobprod[(\Omega_t)]{\ell}{a_J}{Ne_Ik_{IJ}}\\
        &\quad - \frac{2}{t} \sobprod[(\Omega_t)]{\ell}{a_J}{e_J(t\Theta-1)} - \frac{2}{t} \sobprod[(\Omega_t)]{\ell}{e_Ja_J}{t\Theta-1} \\
        &\quad + 2\sobprod[(\Omega_t)]{\ell}{-Ne_Ie_J(\ln N)}{k_{IJ}} + 2\sobprod[(\Omega_t)]{\ell}{-Ne_Jk_{JI}}{e_I(\ln N)}\\
        &\quad + \frac{2\lpar}{t}\sobprod[(\Omega_t)]{\ell}{-e_Ie_I(\ln N)}{t\Theta-1} + \frac{2\lpar}{t}\sobprod[(\Omega_t)]{\ell}{-e_I(t\Theta-1)}{e_I(\ln N)}\\
        &\quad + 2\sobprod[(\Omega_t)]{\ell}{-Ne_Ie_I\s}{e_0\s} + 2\sobprod[(\Omega_t)]{\ell}{-Ne_Ie_0\s}{e_I\s},
    \end{split}
    \]
    for all $t \leq t_0$.
\end{lemma}

\begin{proof}
    Putting Lemmas~\ref{lemma: frame higher order estimate}, \ref{lemma: k higher order estimate}, \ref{lemma: mean curvature higher order estimates}, \ref{lemma: lapse higher order estimates}, \ref{lemma: lapse derivatives higher order estimates}, \ref{lemma: gammas higher order estimate}, \ref{lemma: a higher order estimate}, \ref{lemma: phi0 higher order estimate}, \ref{lemma: phi derivatives higher order estimate}, \ref{lemma: phi higher order estimate} and \ref{lemma: contstraint trick higher order} together yields
    \[
    \begin{split}
        \int_{\Omega_t} -\p_t \rho_\ell \,dx &\leq \frac{2(1+\delta)}{t} \int_{\Omega_t} \dens{e} + \dens{\omega} \,dx + \frac{\lpar+3}{t}\int_{\Omega_t} \dens{N} \,dx + \frac{6}{t} \int_{\Omega_t} \dens{\gamma} \,dx\\
        &\quad + \frac{2}{t} \int_{\Omega_t} \dens{k} + \dens{e_0\s} + \dens{\ve\s} \,dx + \frac{4}{t} \int_{\Omega_t} \dens{a} \,dx + \frac{2-2\delta}{t} \int_{\Omega_t} \dens{\s} \,dx\\
        &\quad + \frac{2}{t} \int_{\Omega_t} \dens{\ve(\ln N)} \,dx + \frac{1}{t} \Big[ \Big(1+\frac{1}{\lpar}\Big)\ck{0}{N}^4-2\Big(\lpar + 1 + \frac{1}{\lpar}\Big) \Big] \int_{\Omega_t} \dens{\Theta} \,dx\\
        &\quad + \frac{4\lpar+4}{t} \int_{\Omega_t} \sum_{1 \leq |\I| \leq \ell}\big(k_{IJ}\pI k_{IJ} + e_0(\s)\pI e_0(\s)\big) \pI(t\Theta-1) \,dx\\
        &\quad + \frac{4\lpar+2}{t} \int_{\Omega_t} \frac{1}{N} (k_{IJ}k_{IJ} + e_0(\s)^2)\sum_{1 \leq |\I| \leq \ell}(\pI N) \pI(t\Theta-1) \,dx\\
        &\quad -4 \int_{\Omega_t} \sum_{|\I| \leq \ell} \big(Nk_{JK} \pI\gamma_{JIK} + Ne_0(\s) \pI e_I\s\big) (\pI a_I) \,dx\\
        &\quad + \frac{C_\ell}{t} \ce_{\ell-1}^{1/2} \ce_\ell^{1/2} + C_\ell t^{-1+\delta} \ce_\ell + \cf,
    \end{split}
    \]
    where we have written $\sobprod{\ell}{Ne_Ia_J}{k_{IJ}}$ similarly as at the end of the proof of Lemma~\ref{lemma: constraint equations trick}. By a few applications of the Cauchy-Schwarz inequality, \eqref{eq: kasner relations estimate} and Young's inequality,
    \[
    \begin{split}
        &\frac{4\lpar+4}{t} \int_{\Omega_t} \sum_{1 \leq |\I| \leq \ell} \big(k_{IJ}\pI k_{IJ} + e_0(\s)\pI e_0(\s)\big) \pI(t\Theta-1) \,dx\\
        &\leq \frac{4\lpar+4}{t}\ck{0}{t^2N^2(k_{IJ}k_{IJ} + e_0(\s)^2)}^{1/2} \Big( \int_{\Omega_t} \dens{k} + \dens{e_0\s} \,dx \Big)^{1/2} \Big( \frac{1}{\lpar} \int_{\Omega_t} \dens{\Theta} \,dx \Big)^{1/2}\\
        &\leq \frac{2\lpar+2}{t} \Big( \int_{\Omega_t} \dens{k} + \dens{e_0\s} \,dx + \frac{1}{\lpar} \int_{\Omega_t} \dens{\Theta} \Big) + C_\ell t^{-1+\delta} \ce_\ell.
    \end{split}
    \]
    Next,
    \[
    \begin{split}
        &\frac{4\lpar+2}{t} \int_{\Omega_t} \frac{1}{N} \big(k_{IJ}k_{IJ} + e_0(\s)^2\big)\sum_{1 \leq |\I| \leq \ell}(\pI N) \pI(t\Theta-1) \,dx\\
        &\leq \frac{4\lpar+2}{t} \ck{0}{1/N}^2 \ck{0}{t^2N^2(k_{IJ}k_{IJ} + e_0(\s)^2)} \Big( \int_{\Omega_t} \dens{N} \,dx \Big)^{1/2} \Big( \frac{1}{\lpar} \int_{\Omega_t} \dens{\Theta} \,dx \Big)^{1/2}\\
        &\leq \frac{2\lpar+1}{t} \Big( \int_{\Omega_t} \dens{N} \,dx + \frac{1}{\lpar} \ck{0}{1/N}^4 \int_{\Omega_t} \dens{\Theta} \,dx\Big) + C_\ell t^{-1+\delta} \ce_\ell.
    \end{split}
    \]
    Moving on, by Young's inequality,
    \[
    \begin{split}
        &-4 \int_{\Omega_t} \sum_{|\I| \leq \ell} \big(Nk_{JK} \pI\gamma_{JIK} + Ne_0(\s) \pI e_I\s\big) (\pI a_I) \,dx\\
        &\hspace{3cm}\leq \frac{1}{t} \int_{\Omega_t} \sum_{I,J,K} \sum_{|\I| \leq \ell} (\pI\gamma_{IJK})^2 \,dx + \frac{1}{t} \int_{\Omega_t} \sum_I \sum_{|\I| \leq \ell} (\pI e_I\s)^2 \,dx\\
        &\hspace{3cm}\quad + 4t \int_{\Omega_t} N^2\big(k_{JK}k_{JK} + e_0(\s)^2\big)\sum_I \sum_{|\I| \leq \ell}(\pI a_I)^2 \,dx\\
        &\hspace{3cm}\leq \frac{2}{t} \int_{\Omega_t} \dens{\gamma} \,dx + \frac{1}{t} \int_{\Omega_t} \dens{\ve\s} \,dx + \frac{4}{t} \int_{\Omega_t} \dens{a} \,dx + C_\ell t^{-1+\delta} \ce_\ell.
    \end{split}
    \]
    Putting these observations together, we obtain
    \[
    \begin{split}
        \int_{\Omega_t} -\p_t \rho_\ell \,dx &\leq \frac{2(1+\delta)}{t} \int_{\Omega_t} \dens{e} + \dens{\omega} \,dx + \frac{3\lpar+4}{t}\int_{\Omega_t} \dens{N} \,dx + \frac{8}{t} \int_{\Omega_t} \dens{\gamma} \,dx\\
        &\quad + \frac{2\lpar+4}{t} \int_{\Omega_t} \dens{k} + \dens{e_0\s} \,dx + \frac{3}{t} \int_{\Omega_t} \dens{\ve\s} \,dx + \frac{8}{t} \int_{\Omega_t} \dens{a} \,dx\\
        &\quad + \frac{2-2\delta}{t} \int_{\Omega_t} \dens{\s} \,dx + \frac{2}{t} \int_{\Omega_t} \dens{\ve(\ln N)} \,dx\\
        &\quad + \frac{1}{t} \Big[ \Big( 2 + \frac{1}{\lpar} \Big) \ck{0}{1/N}^4 + \Big(1+\frac{1}{\lpar}\Big)\ck{0}{N}^4 - 2\lpar \Big] \int_{\Omega_t} \dens{\Theta} \,dx\\
        &\quad + \frac{C_\ell}{t} \ce_{\ell-1}^{1/2} \ce_\ell^{1/2} + C_\ell t^{-1+\delta} \ce_\ell + \cf.
    \end{split}
    \]
    If $2^{-1/4} \leq N \leq 2^{1/4}$, then the constant in front of the integral of $\dens{\Theta}$ is less or equal than $6 + \frac{4}{\lpar} -2\lpar$. But this number is less or equal than $8$ if and only if $\lpar \geq 1$. If additionally we have $\lpar \leq \frac{4}{3}$, we obtain the result. 
\end{proof}

Now we are ready to prove energy estimates of all orders.

\begin{proposition} \label{prop: higher order energy estimate}
    If $2^{-1/4} \leq N \leq 2^{1/4}$ and $1 \leq \lpar \leq \frac{4}{3}$, then for every $\ell$ there is a constant $C_\ell$ such that
    \[
    \ce_\ell(t) \leq C_\ell t^{-9}
    \]
    for all $t \leq T$.
\end{proposition}

\begin{proof}
    We prove this by induction. Fix some $\lambda \in (0,\frac{1}{2})$, and suppose that there are constants $C_{\ell-1}$ and $b_{\ell-1}$ such that
    \begin{equation}
        \ce_{\ell-1}(t) \leq C_{\ell-1}\langle \ln t \rangle^{b_{\ell-1}} t^{-8-2\lambda}
    \end{equation}
    for all $t \leq T$. Note that our assumptions imply that this is satisfied for $\ell=2$.
    
    Young's inequality and Lemma~\ref{lemma: differential energy estimate} imply
    \[
    \int_{\Omega_t} -\p_t \rho_\ell \,dx \leq \frac{8+\lambda}{t} \ce_\ell + \frac{1}{t} \frac{C_\ell^2}{4\lambda} \ce_{\ell-1} + C_\ell t^{-1+\delta}\ce_\ell + \cf
    \]
    for all $t \leq t_0$. Applying Lemma~\ref{divergence theorem lemma} to $t^{8+2\lambda}\rho_\ell$, we get
    \[
    \begin{split}
        &s^{8+2\lambda} \ce_\ell(s) + \int_{\side_{[s,t_\ell]}} t^{8+2\lambda}\rho_\ell\, \frac{t^{-1+3\delta}}{|df|} \,\mu_{\side_{[s,t_\ell]}}\\
        &\leq t_\ell^{8+2\lambda} \ce_\ell(t_\ell) + \int_s^{t_\ell} t^{8+2\lambda}\Big(  -\frac{\lambda}{t} \ce_\ell(t) + \frac{C_\ell}{t} \ce_{\ell-1}(t) + C_\ell t^{-1+\delta} \ce_\ell(t) + \cf(t) \Big) \,dt 
    \end{split}
    \]
    for some $t_\ell \leq t_0$, where we allow $C_\ell$ to depend on $\lambda$. By Lemma~\ref{lemma: symmetric hyperbolic estimate}, we have 
    \[
    \Big| \int_s^{t_\ell} t^{8+2\lambda} \cf(t) \,dx \Big| \leq \int_s^{t_\ell} C_\ell t^{-1+\delta} t^{8+2\lambda} \ce_\ell(t) \,dt + \int_{\side_{[s,t_\ell]}} C_\ell t^{8+2\lambda} \rho_\ell\, \frac{t^{-1+4\delta}}{|df|} \,\mu_{\side_{[s,t_\ell]}}.
    \]
    These two observations together yield
    \[
    \begin{split}
        &s^{8+2\lambda} \ce_\ell(s) + \int_{\side_{[s,t_\ell]}} t^{8+2\lambda} \rho_\ell\, \frac{t^{-1+3\delta}}{|df|} \,\mu_{\side_{[s,t_\ell]}}\\
        &\leq t_\ell^{8+2\lambda} \ce_\ell(t_\ell) + \int_s^{t_\ell} t^{8+2\lambda}\bigg(  \frac{C_\ell t^\delta - \lambda}{t} \ce_\ell(t) + \frac{C_\ell}{t} \ce_{\ell-1}(t) \bigg) \,dt + \int_{\side_{[s,t_\ell]}} C_\ell t^{8+2\lambda} \rho_\ell\, \frac{t^{-1+4\delta}}{|df|} \,\mu_{\side_{[s,t_\ell]}} 
    \end{split}
    \]
    (recall that in the global setting the side boundary integrals do not appear). Now we choose $t_\ell$ small enough such that $C_\ell t_\ell^\delta \leq \lambda < 1$. Then the total contribution of the boundary integrals can be discarded. Thus
    \[
    \begin{split}
        s^{8+2\lambda} \ce_\ell(s) - t_\ell^{8+2\lambda} \ce_\ell(t_\ell) \leq \int_s^{t_\ell} \frac{C_\ell}{t} t^{8+2\lambda} \ce_{\ell-1}(t) \,dt \leq \int_s^{t_\ell} \frac{C_\ell}{t} \langle \ln t \rangle^{b_{\ell-1}} \,dt \leq C_\ell \langle \ln s \rangle^{b_{\ell-1}+1},
    \end{split}
    \]
    where we have used the inductive assumption. That is,
    \[
    t^{8+2\lambda} \ce_\ell(t) \leq C_\ell \langle \ln t \rangle^{b_{\ell-1}+1}
    \]
    for all $t \leq t_\ell$. This proves the inductive step, since this inequality also holds for all $t \in [t_\ell,T]$ by smoothness of the solution. Finally, note that $\langle \ln t \rangle^{b_\ell}t^{1-2\lambda}$ is bounded for all $t \leq T$. 
\end{proof}

\begin{proof}[Proof of Theorem~\ref{thm: smooth estimates}]
    By interpolation and Proposition~\ref{prop: higher order energy estimate},
    \[
    \sob[(\Omega_t)]{\ell}{e} \leq C_{\ell,k} t^{(-1+4\delta)(1-\ell/k)} t^{-9\ell/(2k)}.
    \]
    Taking $k$ large enough, depending only on $\ell$ and $\delta$, we thus get
    \[
    \sob[(\Omega_t)]{\ell}{e} \leq C_\ell t^{-1+3\delta}.
    \]
    Since this is true for every $\ell$, it is also true for every $C^\ell$ norm, by Sobolev embedding. The results for $\omega$, $\gamma$, $t\Theta-1$, $\ve(\ln N)$ and $\ve\s$ follow similarly. A similar argument yields
    \[
    t\ck[(\overline{\Omega}_t)]{\ell}{k} + t\ck[(\overline{\Omega}_t)]{\ell}{e_0\s} + \ck[(\overline{\Omega}_t)]{\ell}{N} \leq C_\ell t^{-\delta}.
    \]
    Now we need to improve on these estimates. To that end, note that the proof of Lemma~\ref{lemma: potential estimates 2} plus a similar argument as above implies that
    \[
    \sob[(\Omega_t)]{\ell}{\s} \leq C_\ell t^{-\delta}
    \]
    for all $\ell$. Therefore, by Sobolev embedding and Lemma~\ref{lemma: potential estimates 2},
    \[
    \ck[(\overline{\Omega}_t)]{\ell}{V\circ\s} + \ck[(\overline{\Omega}_t)]{\ell}{V'\circ\s} \leq C_\ell t^{-2+4\delta}.
    \]
    From \eqref{k equation global}, \eqref{eq: phi equation global} and \eqref{lapse equation global} we see that
    \begin{align*}
    \begin{split}
        \p_t(tk_{IJ}) &= (1-t\Theta)k_{IJ} + tN * e * \p\gamma + tN * e * \p\ve(\ln N) + t N * \gamma * \gamma\\
        &\quad + tNe_I(\ln N)e_J(\ln N) + tN * \gamma * \ve(\ln N) + tNe_I(\s)e_J(\s) + tN(V\circ\s)\delta_{IJ}, 
    \end{split}\\
    \p_t(te_0\s) &= (1-t\Theta)e_0\s + tNe_Ie_I\s - tN\gamma_{IKK}e_I\s + tNe_I(\s)e_I(\ln N) - tNV'\circ\s,\\
    \p_tN &= \frac{\lpar+1}{t} N(t\Theta-1).
    \end{align*}
    Therefore, using the estimates we have obtained so far, we conclude that \eqref{eq: estimates on time derivatives} holds. It follows that for $r \leq s$ and $|\I| \leq \ell$,
    \[
    |r\pI k_{IJ}(r,x) - s\pI k_{IJ}(s,x)| \leq \int_r^s \big|\pI\big(\p_t(tk_{IJ})\big)(t,x)\big| \,dt \leq C_\ell (s^{\delta} - r^{\delta}), 
    \]
    for all $x \in \overline{\Omega}_r$. As a consequence, we obtain the desired $C^\ell$ bounds for $tk_{IJ}$ and we can find the functions $\mfkr_{IJ} \in C^\infty(\overline{\Omega}_0)$. In a similar way we can obtain the $C^\ell$ estimates for $e_0\s$ and $N$, and the existence of $\nr$. To find $\psir$, note that
    \[
    \p_t\Psi = -\frac{N^2}{(t\Theta)^2} \p_t(t\theta) te_0\s + \frac{N}{t\Theta} \p_t(te_0\s).
    \]
    Whence $\ck[(\overline{\Omega}_t)]{\ell}{\p_t\Psi} \leq C_\ell t^{-1+\delta}$, by \eqref{eq: estimates on time derivatives}.
    Next, to find $\phir$, we have that
    \[
    \p_t\Phi = \Big( \Theta - \frac{1}{t} \Big)\Psi + (\p_t\Psi)\ln\theta + \frac{N}{t\Theta} \Psi \p_t(t\theta).
    \]
    It follows that $\ck[(\overline{\Omega}_t)]{\ell}{\p_t\Phi} \leq C_\ell t^{-1+\delta/2}$, since $\ln\theta = \ln(t\Theta) - \ln t - \ln N$. This finishes the proof.
\end{proof}

\section{Proof of localized big bang formation} \label{sec: proof of main theorem}

Now we can prove Theorem~\ref{thm: main theorem} by combining Theorems~\ref{thm: global existence} and \ref{thm: smooth estimates} with \cite[Theorem~24]{franco_complete_asymptotics_2026}.

\begin{proof}[Proof of Theorem~\ref{thm: main theorem}]
    Fix $1 \leq \lpar \leq \frac{4}{3}$ and $k_1$ such that Theorem~\ref{thm: global existence} is applicable. Taking into account that $T = 1/\bar\theta(0)$, let $\zeta_1$ be large enough such that Theorem~\ref{thm: global existence} is applicable. Then we deduce the existence of a solution $(\Omega_{(0,T]},g,\s)$ to the Einstein--nonlinear scalar field equations with potential $V$, as described in Proposition~\ref{prop: reduced equations struct coeffs}, such that the initial data induced on $\Omega_T$ coincide with $\mathfrak{I}$. Next, we wish to apply Theorem~\ref{thm: smooth estimates}. Note that as a consequence of \eqref{eq: estimates on the global solution}, we have that \eqref{eq: C1 assumptions} holds, and that the required bounds on $N$ hold, by taking $\zeta_1$ larger if necessary. In order to deduce that \eqref{eq: kasner relations estimate} holds, note that
    \[
    \begin{split}
        t^2N^2\big(k_{IJ}k_{IJ} + e_0(\s)^2\big) &= \theta^{-2}k_{IJ}k_{IJ} + \theta^{-2}e_0(\s)^2\\
        &\quad + (t\Theta-1)(t\Theta+1)\big(\theta^{-2}k_{IJ}k_{IJ} + \theta^{-2}e_0(\s)^2\big)
    \end{split}
    \]
    and apply Lemma~\ref{estimate on the kasner relations}. We conclude that Theorem~\ref{thm: smooth estimates} is applicable. Now, to obtain the existence of the initial data on the singularity, we would like to apply \cite[Theorem~24]{franco_complete_asymptotics_2026}. To that end, we proceed to obtain asymptotics for the eigenvalues of $\mck$.
    
    For square $n \times n$ matrices $A, \Tilde{A}$, define
    \[
    \md(A,\Tilde{A}) := \min_\sigma\big\{ \max_i |\Tilde{\lambda}_{\sigma(i)} - \lambda_i| \big\},
    \]
    where $\lambda_i, \Tilde{\lambda}_i$ are the eigenvalues of $A$ and $\Tilde{A}$ respectively, and $\sigma$ runs over all permutations of $\{1,\ldots,n\}$. Then, by \cite[Chapter~IV, Theorem~3.3, p. 192]{stewart_matrix_1990} and Theorem~\ref{thm: smooth estimates}, there is a constant $C$ such that
    \[
    \md\big( Tk_{IJ}(T,x), tk_{IJ}(t,x) \big) \leq C\tsum_{I,J}\ck[(\overline{\Omega}_t)]{0}{Tk_{IJ}(T) - tk_{IJ}(t)} \leq CT^{\delta},
    \]
    for all $x \in \overline{\Omega}_t$. Recall that $Tk_{IJ}(T) = T\bar\theta\bar p_I\delta_{IJ}$. Hence, taking into account Lemma~\ref{lemma: estimates on initial mean curvature}, we conclude that there is an ordering $q_I$ of the eigenvalues of $(tk_{IJ})$ such that
    \[
    \ck[(\overline{\Omega}_t)]{0}{q_I(t) - \bar p_I} \leq CT^{\delta}.
    \]
    In particular, if $T$ is small enough, then $|q_I - q_J| \geq 1/(2\zeta_0)$ for $I \neq J$. Let $P(t,x,\lambda) := \det(tk_{IJ}(t,x) - \lambda \delta_{IJ})$ be the characteristic polynomial of $(tk_{IJ})$. Then
    \[
    \p_\lambda P\big(t,x,q_I(t,x)\big) = -\prod_{J \neq I}\big(q_J(t,x) - q_I(t,x)\big) \neq 0.
    \]
    By the implicit function theorem, it follows that the $q_I$ are smooth in $\Omega_{(0,T)}$ and that
    \begin{equation} \label{eq: derivatives of the eigenvalues}
        Xq_I = \frac{X\big( \det(tk_{IJ}(t,x) - \lambda \delta_{IJ} )\big)\big|_{\lambda = q_I}}{\prod_{J \neq I}(q_J - q_I)}
    \end{equation}
    for $X \in \mfx(\Omega_{(0,T)})$. If we let $X = \p_i$, then iteratively applying more derivatives to \eqref{eq: derivatives of the eigenvalues}, we conclude that $q_I$ are bounded in $C^\ell(\overline{\Omega}_t)$, for all $\ell$. If we now let $X = \p_t$ in \eqref{eq: derivatives of the eigenvalues}, then we can conclude that $\ck[(\overline{\Omega}_t)]{\ell}{\p_tq_I} \leq C_\ell t^{-1+\delta}$ for all $\ell$. Now we are ready to consider $\mck$. Note that $\mck(e_I,\omega^J) = \frac{1}{\theta}k_{IJ}$. It follows that $p_I := \frac{1}{t\theta}q_I$ are the eigenvalues of $\mck$. Note that \eqref{eq: estimates on time derivatives}, together with what we know about $\p_t q_I$ implies that $\ck[(\overline{\Omega}_t)]{\ell}{\p_t p_I} \leq C_\ell t^{-1+\delta}$ for all $\ell$. Therefore, if $s \leq t$,
    \[
    \ck[(\overline{\Omega}_s)]{\ell}{p_I(t) - p_I(s)} \leq C_\ell (t^{\delta} - s^{\delta}).
    \]
    In particular, there are $\pr_I \in C^\infty(\overline{\Omega}_0)$ such that 
    \[
    \ck[(\overline{\Omega}_0)]{\ell}{p_I(t) - \pr_I} \leq C_\ell t^{\delta}
    \]
    for all $\ell$. Moreover, $|p_I - p_J| \geq \frac{N}{2\zeta_0t\Theta}$ for $I \neq J$, which means that $p_I$ and $\pr_I$ are distinct everywhere. Next, if $I \neq J$, then
    \[
    \pr_I + \pr_J - \pr_K \leq \bar p_I + \bar p_J - \bar p_K + CT^{\delta} < 1 - 6\delta + CT^{\delta} < 1, 
    \]
    provided that $T$ is small enough. Finally, we have
    \[
    \tsum_I \pr_I^2 + \psir^2 = 1,
    \]
    which follows as a consequence of letting $t \to 0$ in Lemma~\ref{estimate on the kasner relations}. At this point, we can apply \cite[Theorem~24 and Remark~25]{franco_complete_asymptotics_2026} to conclude that there are smooth initial data on the singularity $(\overline{\Omega}_0,\mchr,\mckr,\phir,\psir)$ as desired. Curvature blow up follows similarly as in the proof of \cite[Theorem~132]{oude_groeniger_formation_2023}.

    Now we verify the statements about the map $\Xi$. Let $\gamma: (a,b) \to \Omega_{(0,T)}$ be a future directed causal curve. Up to a reparameterization, we can write it as $\gamma(\tau) = \big(\tau,x(\tau)\big)$ with $(a,b) \subset (0,T)$. If $\tau_1 \leq \tau_2$, then
    \[
    \begin{split}
        |x(\tau_2) - x(\tau_1)| \leq \int_{\tau_1}^{\tau_2} |x'(\tau)| \,d\tau \leq \int_{\tau_1}^{\tau_2} |x'(\tau)|_h \tsum_I \big|e_I\big(\gamma(\tau)\big)\big| \,d\tau.
    \end{split}
    \]
    $\gamma$ being causal means that $|x'(\tau)|_h \leq N\big(\tau,x(\tau)\big)$. Therefore, due to \eqref{eq: estimates on the global solution},
    \[
    |x(\tau_2) - x(\tau_1)| < \int_{\tau_1}^{\tau_2} \tau^{-1+3\delta} \,d\tau = \frac{1}{3\delta}( \tau_2^{3\delta} - \tau_1^{3\delta} ).
    \]
    But then there is some $y \in \R^3$ such that $x(\tau) \to y$ as $\tau \to b$. Moreover, the inequality above shows that $y \in \Omega_b$. It follows that all future inextendible causal curves in $\Omega_{(0,T]}$ must terminate at $\{T\} \times \Omega_T$. In other words, $\{T\} \times \Omega_T$ is a past Cauchy surface for $\Omega_{(0,T]}$. The existence of the map $\Xi: \Omega_{(0,T]} \to D^-\big( \iota(\Omega_T) \big) \subset M$ now follows as a consequence of the abstract properties of the maximal globally hyperbolic development.

    For the statement about $J^-\big(\iota(0)\big)$, let $\gamma: (a,0] \to M$ be past inextendible future directed causal, such that $\gamma(0) = \Xi(T,0) = \iota(0)$. Then there is an $a' \geq a$ such that 
    \[
    \Xi^{-1} \circ \gamma|_{(a',0]}(\tau) = \big(t(\tau),x(\tau)\big)
    \]
    is past inextendible in $\Omega_{(0,T]}$. A similar calculation as above then shows that $|x(\tau)| < \frac{1}{3\delta}T^{3\delta}$ for all $\tau \in (a',0]$. If $\frac{1}{3\delta}T^{3\delta} \leq \varepsilon$, this implies that $\big(t(\tau),x(\tau)\big) \in (0,T] \times \Omega_0$ for all $\tau \in (a',0]$, and $\big(t(\tau),x(\tau)\big)$ must approach some point in $\{0\} \times \Omega_0$ as $\tau \to a'$. Curvature blow up now implies that we must have $a' = a$. 
    
    Finally, suppose that $\gamma:(a,0] \to \Omega_{(0,T]}$ is a past inextendible future directed causal geodesic with $\gamma(0) = (T,0)$. Define $v^\mu(\tau)$ and $\bar v(\tau)$ by the relations
    \[
    \gamma'(\tau) = v^0(\tau) e_0 |_{\gamma(\tau)} + \bar v(\tau) = v^0(\tau) e_0|_{\gamma(\tau)} + v^I(\tau)e_I|_{\gamma(\tau)},
    \]
    and let $u(\tau) := v^0(\tau)(\theta\circ\gamma)(\tau)$. Note that $\gamma$ being future directed means that $v^0 > 0$, and $\gamma$ being causal means that $|\bar v|_h \leq v^0$. Since $\gamma$ is a geodesic,
    \[
    \frac{d}{d\tau}v^0 = -\frac{d}{d\tau}g(\gamma', e_0|_{\gamma}) = -v^Iv^Jk_{IJ}\circ\gamma - v^Iv^0g( e_I,\n_{e_0}e_0 )\circ\gamma, 
    \]
    It follows that
    \[
    u' = -u^2 - v^Iv^J(\theta k_{IJ})\circ\gamma - u v^I g(e_I,\n_{e_0}e_0)\circ\gamma + (v^0)^2( e_0\theta + \theta^2)\circ\gamma + v^0v^I e_I\theta\circ\gamma.
    \]
    Note that
    \[
    |v^Iv^J(\theta k_{IJ})\circ\gamma| \leq (\theta\circ\gamma)^2 \textstyle\max_I\{|p_I|\circ\gamma\}|\bar v|_h^2 \leq (1-6\delta)u^2, 
    \]
    where we have used that $\gamma$ is causal. Also, 
    \[
    |uv^I g(e_I,\n_{e_0}e_0)\circ\gamma| = |uv^Ie_I(\ln N)\circ\gamma| \leq C(t\circ\gamma)^{-1+4\delta}uv^0 \leq C(t\circ\gamma)^{4\delta}u^2, 
    \]
    where we have used that $\frac{N}{t\Theta}$ is bounded in $\Omega_{(0,T]}$. Moreover,
    \[
    v^0v^Ie_I\theta \circ\gamma= \frac{1}{N\circ\gamma}v^0v^Ie_I\Theta\circ\gamma - uv^Ie_I(\ln N)\circ\gamma,
    \]
    which similarly as above, implies that $|v^0v^Ie_I\theta\circ\gamma| \leq C(t\circ\gamma)^{4\delta}u^2$. Note that, due to \eqref{k equation global}, we can conclude that $(v^0)^2|(e_0\theta + \theta^2)\circ\gamma| \leq C(t\circ\gamma)^{4\delta}u^2$ as well. Taking $\zeta_1$ even larger if necessary, it follows that
    \[
    u' \leq -6\delta u^2 + CT^{4\delta}u^2 \leq -\delta u^2.
    \]
    Since $u > 0$, we conclude that $u$ must blow up in finite parameter time to the past. As a consequence, $\gamma$ is past incomplete. 
\end{proof}

\begin{appendices}

\section{Conventions, identities and inequalities}

\subsection{Conventions} \label{ssec: conventions appendix}

Let $\Omega \subset \R^3$ be open. We denote by $\p_i$ the standard coordinate vector fields on $\Omega$.

\begin{definition}
    A \emph{multiindex} $\I$ is an element of $\{1,2,3\}^k$. We use the notation $|\I| = k$, and if $\I = (i_1,\ldots,i_k)$, then $\p_{\hspace{1pt}\I} = \p_{i_1} \cdots \p_{i_k}$. We write $\J \cup \K = \I$ if there is some $n$, with $1 \leq n \leq k$, such that $\J = (i_{\ell_1},\ldots,i_{\ell_n})$ and $\K = (i_{\ell_{n+1}},\ldots,i_{\ell_k})$, where $(\ell_1,\ldots,\ell_k)$ is a permutation of $(1,\ldots,k)$, such that $\ell_1 < \ell_2 < \cdots < \ell_n$ and $\ell_{n+1} < \cdots < \ell_k$. We use similar conventions for larger unions of multiindices.
\end{definition}

If $u$ is a function on $\Omega$, the $C^k$ and Sobolev norms that we work with are defined by
\begin{align*}
    \ck[(\overline{\Omega})]{k}{u} &:= \sum_{|\I| \leq k} \sup_{x \in \Omega} |\pI u(x)|,\\
    \sob[(\Omega)]{k}{u} &:= \Bigg( \sum_{|\I| \leq k} \int_\Omega |\pI u(x)|^2\,dx \Bigg)^{1/2}.
\end{align*}
For the $C^k$ and Sobolev norms of the indexed objects that appear throughout the paper, such as $e_I^i$ and $e_I(\ln N)$, we use the following conventions:
\[
\ck[(\overline{\Omega})]{k}{e} := \tsum_{I,i} \ck[(\overline{\Omega})]{k}{e_I^i}, \qquad \sob[(\Omega)]{k}{\ve(\ln N)} := \Big( \tsum_I \sob[(\Omega)]{k}{e_I(\ln N)}^2 \Big)^{1/2},
\]
and so on. If $T$ is a tensor on $\Omega$, we can define its $C^k$ and Sobolev norms by considering the components of $T$ in terms of the coordinate frame $\{\p_i\}$ and using our conventions for indexed objects. For the global setting, let $(\Sigma,\refmetric)$ be a 3-dimensional closed Riemannian manifold with a fixed reference global frame $\{E_i\}$ which is orthonormal with respect to $\refmetric$. Then the corresponding $C^k$ and Sobolev norms on $\Sigma$ are defined in a similar way, but with respect to the frame $\{E_i\}$, instead of the coordinate frame $\{\p_i\}$ above.

\begin{remark} \label{rmk: paralellizable}
    Assuming that a closed manifold is parallelizable is, in general, a topological restriction. However, all 3-dimensional closed orientable manifolds are parallelizable; see \cite{Benedetti_framing_2018}. 
\end{remark}

We use the standard conventions of round and square brackets to denote symmetrization and antisymmetrization of indices, so that, for instance,
\[
A_{(IJ)} = \tfrac{1}{2}(A_{IJ} + A_{JI}), \qquad A_{[IJ]} = \tfrac{1}{2}(A_{IJ} - A_{JI}).
\]
Moreover, for an indexed quantity $A_{IJ}$, with $I,J = 1,2,3$, define its trace-free symmetrization by
\[
A_{\langle IJ\rangle} := \tfrac{1}{2}( A_{IJ} + A_{JI} )- \tfrac{1}{3}\tsum_KA_{KK}\delta_{IJ}.
\]
In case there are more indices involved we write
\[
A_{\langle I|JK|L\rangle} = \tfrac{1}{2}(A_{IJKL} + A_{LJKI}) - \tfrac{1}{3}\tsum_MA_{MJKM} \delta_{IL},
\]
etc., so that the indices surrounded by vertical bars are unaffected.

\subsection{Inequalities}

Here we list some inequalities that are used throughout the paper.

\begin{lemma} \label{commutator estimate}
    Let $\Omega \subset \R^3$ be open and $u,v \in C^\infty(\overline{\Omega})$. Then
    \[
    \begin{split}
        \sum_{|\I| = k} |\p_{\hspace{1pt}\I}(uv)| &\leq 2^k \sum_{|\J| \leq k} |\p_\J u| \sum_{|\K| \leq k} |\p_{\hspace{1pt}\K} v|,\\
        \sum_{|\I| = k} \big([\p_{\hspace{1pt}\I}, u]v\big)^2 &\leq 4^k \sum_{|\J| \leq k} |\p_\J u|^2 \sum_{|\K| \leq k-1} |\p_\K v|^2.
    \end{split}
    \]
\end{lemma}

\begin{proof}
    See the proof of \cite[Lemma~134]{oude_groeniger_formation_2023}.
\end{proof}

\begin{lemma}[Interpolation inequality] \label{lemma: interpolation}
    Let $\Omega \subset \R^3$ be open and bounded with smooth boundary, and $\ell \leq k$ nonnegative integers. Then  
    \[
    \sob[(\Omega)]{\ell}{u} \leq C\|u\|_{L^2(\Omega)}^{1-\ell/k} \sob[(\Omega)]{k}{u}^{\ell/k}
    \]
    for all $u \in C^\infty(\overline{\Omega})$, where $C$ is a constant depending only on $\Omega$, $k$ and $\ell$.
\end{lemma}

\begin{proof}
    See \cite[Theorem~5.2]{adams_sobolev_2003}.
\end{proof}

\begin{lemma}[Moser estimates] \label{lemma: moser}
    Let $\Omega \subset \R^3$ be open and bounded with smooth boundary, let $\ell = |\I_1| + \cdots + |\I_k|$, and $u_i \in C^\infty(\overline{\Omega})$ for $i = 1,\ldots,k$. Then
    \[
    \|\pI[1]u_1 \cdots \pI[k]u_k\|_{L^2(\Omega)} \leq C\sum_{i=1}^k \sob[(\Omega)]{\ell}{u_i} \prod_{j \neq i} \ck[(\overline{\Omega})]{0}{u_j},
    \]
    where $C$ is a constant depending only on $\Omega$ and $\ell$.
\end{lemma}

\begin{proof}
    Let $j < m$ be nonnegative integers and $|\I| = j$. By the Gagliardo-Nirenberg inequality \cite[Theorem~10.1]{friedman_pde_1969} with $q = \infty$, $r = 2$ and $a = j/m$,
    \[
    \|\pI u\|_{L^{2m/j}(\Omega)} \leq C\|u\|_{C^0(\overline{\Omega})}^{1-j/m} \|u\|_{H^m(\Omega)}^{j/m}
    \]
    for all $u \in C^\infty(\overline{\Omega})$, where $C$ is a constant depending only on $\Omega$, $j$ and $m$. With this inequality at hand, the proof is the same as that of \cite[Lemma~6.16]{ringstrom_cauchy_2009}.
\end{proof}

\begin{lemma} \label{lemma: complicated lemma}
    Let $\Omega \subset \R^3$ be open and bounded with smooth boundary, and let $\s_i, \psi_i, \chi_i, \pi_{ij} \in C^\infty(\overline{\Omega})$ for $i,j = 1,\ldots,m$. If $\ell \geq 2$, then
    \begin{align}
        \begin{split}
            \bigg\|\sum_{i=1}^m \s_i \psi_i\bigg\|_{H^\ell(\Omega)} &\leq \Big( \eta + \|\tsum_{i=1}^m \s_i^2\|_{C^0(\overline{\Omega})}^{1/2} \Big) \|\psi\|_{H^\ell(\Omega)}\\
            &\quad + C\langle 1/\eta \rangle^{\ell-1} \big\langle \|\s\|_{H^\ell(\Omega)} \big\rangle^\ell \|\psi\|_{C^2(\overline{\Omega})},
        \end{split}\label{complicated estimate 1}\\
        \begin{split}
            \bigg|\sum_{i=1}^m \langle \s_i \psi_i, \chi_i \rangle_{H^\ell(\Omega)}\bigg| &\leq \Big( \eta + \max_i \|\s_i\|_{C^0(\overline{\Omega})} \Big)\|\psi\|_{H^\ell(\Omega)} \|\chi\|_{H^\ell(\Omega)}\\
            &\quad + C\langle 1/\eta \rangle^{\ell-1} \big\langle \|\s\|_{H^\ell(\Omega)} \big\rangle^\ell \|\psi\|_{C^2(\overline{\Omega})} \|\chi\|_{H^\ell(\Omega)},
        \end{split}\label{complicated estimate 2}\\
        \begin{split}
            \sum_{i=1}^m \langle \s_i \psi_i, \psi_i \rangle_{H^\ell(\Omega)} &\leq \Big( \eta + \textstyle\max_i \sup_{x \in \overline{\Omega}} \s_i(x) \Big)\|\psi\|_{H^\ell(\Omega)}^2\\
            &\quad + C\langle 1/\eta \rangle^{\ell-1} \big\langle \|\s\|_{H^\ell(\Omega)} \big\rangle^\ell \|\psi\|_{C^2(\overline{\Omega})} \|\psi\|_{H^\ell(\Omega)},
        \end{split}\label{complicated estimate 3}\\
        \begin{split}
            \bigg|\sum_{i,j=1}^m \langle \s_i \psi_j, \pi_{ij} \rangle_{H^\ell(\Omega)}\bigg| &\leq \Big( \eta + \|\tsum_{i=1}^m \s_i^2\|_{C^0(\overline{\Omega})}^{1/2} \Big) \|\psi\|_{H^\ell(\Omega)} \|\pi\|_{H^\ell(\Omega)}\\
            &\quad + C\langle 1/\eta \rangle^{\ell-1} \big\langle \|\s\|_{H^\ell(\Omega)} \big\rangle^\ell \|\psi\|_{C^2(\overline{\Omega})} \|\pi\|_{H^\ell(\Omega)}
        \end{split}\label{complicated estimate 4}
    \end{align}
    for every $\eta > 0$, where we use conventions similar to those of Subsection~\ref{ssec: conventions appendix} for norms of indexed objects.
\end{lemma}

\begin{proof}
    The proof is exactly the same as that of \cite[Lemma~140]{oude_groeniger_formation_2023}.
\end{proof}

In the global setting, analogs of Lemmas~\ref{lemma: interpolation}, \ref{lemma: moser} and \ref{lemma: complicated lemma} also hold; see \cite[Lemmas~137, 139 and 140]{oude_groeniger_formation_2023}.

\subsection{The localized spacetime domain} \label{app: the spacetime domain}

For $0 < t_0 < t_1$, consider $\Omega_{[t_0,t_1]}$ as in Definition~\ref{def: spacetime domain}. Define the function $f:[0,\infty) \times \R^3 \to \R$ by 
\[
f(t,x) := |x| - \frac{1}{3\delta}t^{3\delta} - \varepsilon.
\]
Then the boundary of $\Omega_{[t_0,t_1]}$ decomposes as follows,
\[
\p\Omega_{[t_0,t_1]} = \Omega_{t_0} \cup \side_{[t_0,t_1]} \cup \Omega_{t_1},
\]
where $\side_{[t_0,t_1]} := f^{-1}(0) \cap ([t_0,t_1]\times \R^3)$ denotes the ``side" part of the boundary. Note that if $x \neq 0$,
\[
df_{(t,x)} = -t^{-1+3\delta}dt + \sum_i \frac{x^i}{|x|}dx^i.
\]
We denote by $\eta$ the Euclidean metric on $[0,\infty) \times \R^3$. Moreover, given $X \in \mfx([0,\infty) \times \R^3)$ and $\omega \in \Omega^1([0,\infty) \times \R^3)$, we use the notation $|X| = |X|_\eta$ and $|\omega| = |\omega|_\eta$. In order to derive energy estimates, we shall need the following consequences of the divergence theorem.

\begin{lemma} \label{divergence theorem lemma}
    For $0 < t_0 < t_1$, let $u \in C^\infty(\overline{\Omega}_{[t_0,t_1]})$ and $X \in \mfx(\overline{\Omega}_{[t_0,t_1]})$ be tangent to $\Omega_t$ for all $t \in [t_0,t_1]$. Then
    \begin{align*}
        \int_{\Omega_{t_0}} u\,dx + \int_{\side_{[t_0,t_1]}} u \,\frac{t^{-1+3\delta}}{|df|} \,\mu_{\side_{[t_0,t_1]}} &= \int_{\Omega_{t_1}} u\,dx - \int_{t_0}^{t_1} \int_{\Omega_t} \p_tu \,dx dt,\\
        \int_{t_0}^{t_1} \int_{\Omega_t} Xu \,dxdt &= -\int_{t_0}^{t_1} \int_{\Omega_t} u\, \diver_{\eta} X \,dxdt + \int_{S_{[t_0,t_1]}} u\, \frac{df(X)}{|df|} \,\mu_{\side_{[t_0,t_1]}}.
    \end{align*}
    Here $\mu_{\side_{[t_0,t_1]}}$ denotes the volume form of the metric induced on $\side_{[t_0,t_1]}$ by $\eta$.
\end{lemma}

\begin{proof}
    For any vector field $X \in \mfx(\overline{\Omega}_{[t_0,t_1]})$,
    \[
    \diver_\eta(uX) = Xu + u\,\diver_\eta X.
    \]
    In particular,
    \[
    \diver_\eta(u\p_t) = \p_tu.
    \]
    Integrating over $\Omega_{[t_0,t_1]}$ and applying the divergence theorem, we obtain
    \[
    \begin{split}
        \int_{\Omega_{[t_0,t_1]}} \p_tu \,dtdx &= \int_{\Omega_{t_1}} \eta( u\p_t, \p_t )\,dx + \int_{\Omega_{t_0}} \eta(u\p_t, -\p_t)\,dx + \int_{\side_{[t_0,t_1]}} \frac{df(u\p_t)}{|df|} \,\mu_{\side_{[t_0,t_1]}}\\
        &= \int_{\Omega_{t_1}} u\,dx - \int_{\Omega_{t_0}}u\,dx - \int_{\side_{[t_0,t_1]}} u\, \frac{t^{-1+3\delta}}{|df|} \,\mu_{\side_{[t_0,t_1]}}.
    \end{split}
    \]
    If $X$ is tangent to $\Omega_t$ for all $t \in [t_0,t_1]$, similarly, we obtain 
    \[
    \int_{\Omega_{[t_0,t_1]}} Xu + u \,\diver_{\eta} X \,dtdx = \int_{\side_{[t_0,t_1]}} \frac{df(uX)}{|df|} \,\mu_{\side_{[t_0,t_1]}},
    \]
    since $\eta(X,\p_t) = 0$. The result follows.
\end{proof}

\begin{lemma} \label{lemma: symmetric hyperbolic estimate}
    For $0 < t_0 < t_1$ and $\ell$ a positive integer, let $u,v \in C^\infty(\overline{\Omega}_{[t_0,t_1]})$, and $X \in \mfx(\overline{\Omega}_{[t_0,t_1]})$ with $X$ tangent to $\Omega_t$ for all $t \in [t_0,t_1]$. Then
    \[
    \begin{split}
        &\bigg| \int_{t_0}^{t_1} \big\langle Xu,v \big\rangle_{H^{\ell}(\Omega_t)} + \big\langle u,Xv \big\rangle_{H^{\ell}(\Omega_t)} \,dt \bigg|\\
        &\hspace{3cm}\leq \int_{t_0}^{t_1} \Big( \|\diver_\eta X\|_{C^0(\overline{\Omega}_t)} + C\|X\|_{C^1(\overline{\Omega}_t)} \Big) \|u\|_{H^\ell(\Omega_t)} \|v\|_{H^\ell(\Omega_t)}\\
        &\hspace{3cm}\quad + C\|X\|_{H^\ell(\Omega_t)} \Big( \|u\|_{C^1(\overline{\Omega}_t)} \|v\|_{H^\ell(\Omega_t)} + \|v\|_{C^1(\overline{\Omega}_t)} \|u\|_{H^\ell(\Omega_t)} \Big) \,dt\\
        &\hspace{3cm}\quad + \sum_{|\I| \leq \ell} \int_{\side_{[t_0,t_1]}} |\p_{\hspace{1pt}\I}u||\p_{\hspace{1pt}\I}v| \frac{|df(X)|}{|df|} \,\mu_{\side_{[t_0,t_1]}},
    \end{split}
    \]
    where $C$ is a constant depending only on $\Omega_{[t_0,t_1]}$ and $\ell$.
\end{lemma}

\begin{proof}
    Note that
    \[
    \begin{split}
        &\int_{t_0}^{t_1} \big\langle Xu,v \big\rangle_{H^{\ell}(\Omega_t)} + \big\langle u,Xv \big\rangle_{H^{\ell}(\Omega_t)} \,dt\\
        &= \sum_{|\I| \leq \ell} \int_{t_0}^{t_1} \int_{\Omega_t} X\big( (\p_{\hspace{1pt}\I}u)(\p_{\hspace{1pt}\I}v) \big) + [\p_{\hspace{1pt}\I},X](u) \p_{\hspace{1pt}\I}v + [\p_{\hspace{1pt}\I},X](v) \p_{\hspace{1pt}\I}u \,dxdt.
    \end{split}
    \]
    The first term on the right-hand side equals
    \[
    \begin{split}
        - \sum_{|\I| \leq \ell} \int_{t_0}^{t_1} \int_{\Omega_t} (\p_{\hspace{1pt}\I}u)(\p_{\hspace{1pt}\I}v)\, \diver_\eta X \,dxdt + \sum_{|\I| \leq \ell} \int_{\side_{[t_0,t_1]}} (\p_{\hspace{1pt}\I}u)(\p_{\hspace{1pt}\I}v) \frac{df(X)}{|df|} \,\mu_{\side_{[t_0,t_1]}},
    \end{split}
    \]
    by Lemma~\ref{divergence theorem lemma}. Note that the first commutator can be written as a sum of terms of the form
    \[
    (\p_{\J_1}X^i)(\p_{\J_2}\p_iu),
    \]
    where $|\J_1|+|\J_2| = |\I|$ and $|\J_2| \leq |\I|-1$. Hence, by Moser estimates,
    \[
    \sum_{|\I| \leq \ell} \int_{\Omega_t} [\p_{\hspace{1pt}\I},X](u) \p_{\hspace{1pt}\I}v \,dx \leq C\Big( \|X\|_{C^1(\overline{\Omega}_t)} \|u\|_{H^{\ell}(\Omega_t)} + \|u\|_{C^1(\overline{\Omega}_t)} \|X\|_{H^{\ell}(\Omega_t)} \Big) \|v\|_{H^{\ell}(\Omega_t)}.
    \]
    The second commutator is similar.
\end{proof}

\end{appendices}

\printbibliography[heading=bibintoc]

\end{document}